\documentclass{article}
\usepackage{multirow}

     \usepackage[preprint,nonatbib]{neurips_2020}

\usepackage[utf8]{inputenc} 
\usepackage[T1]{fontenc}    
\usepackage{hyperref}       
\usepackage{url}            
\usepackage{booktabs}       
\usepackage{amsfonts}       
\usepackage{nicefrac}       
\usepackage{microtype}      
\usepackage{float}
\usepackage{bbm}
\usepackage{amsthm,amsmath,amssymb}
\usepackage{mathtools}
\usepackage{subcaption}
\usepackage{threeparttable}

\usepackage{caption}
\usepackage{tikz}
\usetikzlibrary{arrows,automata, calc, positioning}
\usepackage{adjustbox}
\usepackage{wrapfig}

\usepackage{accents}

\newcommand{\indep}{\raisebox{0.05em}{\rotatebox[origin=c]{90}{$\models$}}}
\DeclareMathOperator*{\argmin}{\arg\!\min}
\DeclareMathOperator*{\argmax}{\arg\!\max}

\newtheorem{proposition}{Proposition}[section]
\newtheorem{definition}{Definition}[section]
\newtheorem{remark}{Remark}[section]
\newtheorem{lemma}{Lemma}[section]
\newtheorem{theorem}{Theorem}[section]
\newtheorem{algorithm}{Algorithm}[section]
\newtheorem{corollary}{Corollary}[section]

\newtheorem{assumption}{Assumption}[section]

\title{Kernel Methods for Causal Functions:
Dose, Heterogeneous, and Incremental Response Curves}

\author{%
  Rahul Singh \\
  MIT Economics \\
  \texttt{rahul.singh@mit.edu} \\
   \And
   Liyuan Xu \\
Gatsby Unit, UCL \\
   \texttt{liyuan.jo.19@ucl.ac.uk} \\
   \And
   Arthur Gretton \\
   Gatsby Unit, UCL \\
   \texttt{arthur.gretton@gmail.com} \\
}

\begin{document}

\maketitle

\begin{abstract}
We propose estimators based on kernel ridge regression for nonparametric causal functions such as dose, heterogeneous, and incremental response curves. Treatment and covariates may be discrete or continuous in general spaces. Due to a decomposition property specific to the RKHS, our estimators have simple closed form solutions. We prove uniform consistency with finite sample rates via original analysis of generalized kernel ridge regression. We extend our main results to counterfactual distributions and to causal functions identified by front and back door criteria. We achieve state-of-the-art performance in nonlinear simulations with many covariates, and conduct a policy evaluation of the US Job Corps training program for disadvantaged youths. 
\end{abstract}



\section{Introduction}\label{sec:intro}

\subsection{Motivation}

Program evaluation aims to measure the counterfactual relationship between treatment $D$ and outcome $Y$, which may vary for different subpopulations: if we intervened on treatment, setting $D=d$, what would be the expected counterfactual outcome $Y^{(d)}$ for individuals with characteristics $V=v$? When treatment is binary, the causal parameter is a function $\theta_0(v)=E\{Y^{(1)}-Y^{(0)} \mid V=v\}$ called the heterogeneous treatment effect; when treatment is continuous, it is a function $\theta_0(d,v)=E\{Y^{(d)} \mid V=v\}$ that we call a heterogeneous response curve. Assuming selection on observable covariates $(V,X)$, the causal function $\theta_0(d,v)$ can be recovered by integrating the regression function $\gamma_0(d,v,x)=E(Y \mid D=d,V=v,X=x)$ according to the conditional distribution $\text{\normalfont pr}(x \mid v)$: $\theta_0(d,v)=\int \gamma_0(d,v,x)\mathrm{d}\text{\normalfont pr}(x \mid v)$ \cite{rosenbaum1983central,robins1986new}, which may be complex when there are many covariates. 

The same is true for other causal functions such as dose
and incremental response curves, 
and even counterfactual distributions, albeit with different regressions and reweightings. Therefore nonparametric estimation of a causal function involves three challenging steps: estimating a nonlinear regression, with possibly many covariates; estimating the distribution for reweighting, which may be conditional; and using the nonparametric distribution to integrate the nonparametric regression. For this reason, flexible estimation of nonparametric causal functions, such as $\theta_0(d,v)$, is often deemed too computationally demanding to be practical for program evaluation. 



Our key insight is that the reproducing kernel Hilbert space (RKHS), a popular nonparametric setting in machine learning, is precisely the class of functions for which the steps of nonparametric causal estimation can be separated. This decomposition follows almost immediately from the definition of the RKHS, and it is a specific strength of our framework; random forests, for example, do not allow such decoupling. Our key insight follows from a more fundamental one. Evaluation of a causal function is generally not a bounded functional over all of $\mathbb{L}^2$ \cite{van1991differentiable,newey1994asymptotic}. We prove that the evaluation of a causal function is a bounded functional over the RKHS $\mathcal{H}$, which is a subset of $\mathbb{L}^2$. By the classic Riesz representation theorem of functional analysis, a bounded functional over a Hilbert space admits a decoupled inner product representation within the Hilbert space. We show how to use this representation to separate the steps of nonparametric causal estimation. This insight appears to be original.

We adapt kernel ridge regression, a classic machine learning algorithm that generalizes splines \cite{wahba1990spline}, to address the computational challenges of estimating causal functions such as dose, heterogeneous, and incremental response curves. Nonparametric estimation with kernels is quite simple: the nonlinear regression with many covariates can be estimated by simple matrix operations; the conditional distribution can be expressed as a regression problem and estimated by simple matrix operations as well; and the step of integration can be performed by taking the product of the results. The final nonparametric estimator for the causal function has a one line, closed form solution, unlike previous work. This simplicity makes the family of estimators highly practical. The proposed estimators are substantially simpler yet outperform some leading alternatives in nonlinear simulations with many covariates; see Supplement~\ref{sec:experiments}. As extensions, we generalize our new algorithmic techniques to counterfactual distributions in Supplement~\ref{sec:distribution} as well as causal functions and distributions identified by front and back door criteria in Supplement~\ref{sec:graphical}.

Theoretically, our statistical guarantees rely on smoothness of the causal function and spectral decay of the covariance operator rather than the explicit dimension of treatment and covariates. In economic modelling, many variables may matter for labor market decisions, yet economic theory suggests that the effect of different intensities of job training should be well approximated by smooth functions. The emphasis on smoothness in the causal interpretation of RKHS assumptions generalizes standard Sobolev assumptions, and it differs from the emphasis on sparsity in lasso-type assumptions. Our causal function estimators are uniformly consistent with rates that combine minimax optimal rates for smooth nonparametric regressions. En route to our main results, we prove an improved rate for conditional expectation operators. Our main results are nonasymptotic and imply asymptotic uniform validity.

\subsection{Contribution}

Conceptually, we illustrate how to use RKHS techniques in order to separate the steps of nonparametric causal estimation. In doing so, we provide a template for researchers to develop simple kernel estimators for complex causal estimands. Specifically, we clarify five assumptions under which we derive our various results: (i) identification, from the social scientific problem at hand; (ii) basic regularity conditions on the kernels, which are satisfied by all of the kernels typically used in practice; (iii) basic regularity on the  outcome, treatment, and covariates, allowing them to be discrete or continuous variables that take values in general spaces (even texts, images, or graphs); (iv) smoothness of the causal estimand; and (v) spectral decay of the covariance operator. We combine these five assumptions to estimate causal functions, providing insight into the meaning and applicability of RKHS approximation assumptions for causal inference.

Statistically, we prove uniform consistency: our estimators converge to causal functions in $\sup$ norm, which encodes caution about worst case scenarios when informing policy decisions. Our finite sample rates of convergence explicitly account for each source of error at any finite sample size. Our rates do not directly depend on the data dimension, but rather the smoothness of the causal estimand and spectral decay of the covariance operator. The rates may indirectly depend on dimension; see Section~\ref{sec:consistency} for discussion in the context of the Sobolev space, which is a special case of an RKHS. Of independent interest, we provide a technical innovation to justify our main results: relative to previous work, we prove faster rates of convergence in Hilbert--Schmidt norm for conditional expectation operators. We generalize our main results to prove convergence in distribution for counterfactual distributions. The analysis of uniform confidence bands for our causal function estimators is an open question that we pose for future research, since uniform inference for kernel ridge regression remains an open question in statistics.

Computationally, we demonstrate state-of-the-art performance in nonlinear simulations with many covariates, despite the relative simplicity of our proposal compared to existing machine learning approaches. In order to simplify the causal estimation problem, we assume that underlying conditional expectation functions are elements in an RKHS. We propose a family of global estimators with closed form solutions, avoiding density estimation and sampling even for complex integrals. Throughout, the only hyperparameters are kernel hyperparameters and ridge regression penalties. The former have well established tuning procedures, and the latter are easily tuned using the closed form solution for generalized cross validation (which is asymptotically optimal) or leave-one-out cross validation (which we derive). In practice, the tunings are similar, and the asymptotically optimal choice aligns with our statistical theory.

Empirically, our kernel ridge regression approach allows for simple yet flexible estimation of nuanced causal estimands. Such estimands provide meaningful insights about the Job Corps, the largest job training program for disadvantaged youth in the US. Our key statistical assumption is that different intensities of job training have smooth effects on counterfactual employment, and those effects are smoothly modified by age. In our program evaluation in Supplement~\ref{sec:experiments}, we find that the effect of job training on employment substantially varies by class hours and by age; a targeted policy will be more effective. Our program evaluation confirms earlier findings while also uncovering meaningful heterogeneity. 
 We demonstrate how kernel methods for causal functions are a practical addition to the empirical economic toolkit.

\section{Related work}\label{sec:related}

We view nonparametric causal functions as reweightings of an underlying regression, synthesizing the $g$ formula \cite{robins1986new} and partial means \cite{newey1994kernel} frameworks. To express causal functions in this way, we build on canonical identification theorems under the assumption of selection on observables \cite{rosenbaum1983central,robins1986new,altonji2005cross}. We propose simple, global estimators that combine kernel ridge regressions. Previous works that take a global view include \cite{van2003unified,luedtke2016super,diaz2013targeted,kennedy2020optimal}, and references therein. A broad literature instead views causal functions as collections of localized treatment effects and proposes local estimators with Nadaraya--Watson smoothing, e.g. \cite{imai2004causal,rubin2005general,rubin2006extending,galvao2015uniformly,luedtke2016statistical,kennedy2017nonparametric,semenova2021debiased,kallus2018policy,chernozhukov2022debiased,fan2019estimation,zimmert2019nonparametric,colangelo2020double}, and references therein. By taking a global view rather than a local view, we propose simple estimators that can be computed once and evaluated at any value of a continuous treatment, rather than a computationally intensive procedure that must be reimplemented at any treatment value.

Our work appears to be the first to reduce estimation of dose, heterogeneous, and incremental response curves to kernel ridge regressions. Previous works incorporating the RKHS into nonparametric estimation focus on different causal functions: nonparametric instrumental variable regression \cite{carrasco2007linear,darolles2011nonparametric,singh2019kernel}, and heterogeneous treatment effect conditional on the full vector of covariates \cite{nie2021quasi}. 
\cite{nie2021quasi} propose the R learner to estimate the heterogeneous treatment effect $\theta_0(x)=E\{Y^{(1)}-Y^{(0)} \mid X=x\}$. \cite[Section 3]{nie2021quasi} reviews the extensive literature that considers this estimand. The R learner minimizes a loss that contains inverse propensities and different regularization \cite[eq. A24]{nie2021quasi}, and it does not appear to have a closed form solution. The authors prove oracle mean square error rates. By contrast, we pursue a more general heterogeneous response curve with discrete or continuous treatment, conditional on some interpretable subvector $V$ \cite{abrevaya2015estimating}: $\theta_0(d,v)=E\{Y^{(d)} \mid V=v\}$. Unlike previous work on nonparametric causal functions in the RKHS, we (i) consider dose, heterogeneous, and incremental response curves; (ii) propose estimators with closed form solutions; and (iii) prove uniform consistency, which is an important norm for policy evaluation.

We extend the framework from causal functions to counterfactual distributions. Existing work focuses on distributional generalizations of average treatment effect (ATE) or average treatment on the treated (ATT) for binary treatment \cite{firpo2007efficient,cattaneo2010efficient,chernozhukov2013inference}, e.g. $\theta_0=\text{\normalfont pr}\{Y^{(1)}\}-\text{\normalfont pr}\{Y^{(0)}\}$. \cite{muandet2021counterfactual} propose an RKHS approach for distributional ATE and ATT with binary treatment using inverse propensity scores and an assumption on the smoothness of a ratio of densities, which differs from our approach. Unlike previous work, we (i) allow treatment to be continuous; (ii) avoid inversion of propensity scores and densities; and (iii) study a broad class of counterfactual distributions for the full population, subpopulations, and alternative populations, e.g. $\theta_0(d,v)=\text{\normalfont pr}\{Y^{(d)} \mid V=v\}$. 

We provide a detailed comparison with kernel methods for binary treatment effects in Section~\ref{sec:detail}. Whereas we study causal functions, these works study causal scalars \cite{kallus2020generalized,hirshberg2019minimax,singh2021debiased}. We clarify the sense in which our causal function estimators generalize known estimators for treatment effects to new estimators for causal functions. Previous work is inherently tied to the $\mathbb{L}^2$ bounded functional perspective. However, evaluation of a causal function is not a bounded functional over all of $\mathbb{L}^2$ \cite{van1991differentiable,newey1994asymptotic}. Therefore our algorithms extend the conceptual framework of kernel methods for causal inference in a new direction. Our statistical contribution is a new, uniform analysis of response curves that goes beyond pointwise approximation of response curves by local treatment effects. 

This paper subsumes our previous draft \cite[Section 2]{singh2020kernel}.
\section{Causal functions}\label{sec:problem}


A causal function summarizes the expected counterfactual outcome $Y^{(d)}$ given a hypothetical intervention on continuous treatment that sets $D=d$. The causal inference literature studies a rich variety of causal functions with nuanced interpretation, which we define below. Unless otherwise noted, expectations are with respect to the population distribution $\text{\normalfont pr}$.

\begin{definition}[Causal functions]\label{def:causal_param}
We define
\begin{enumerate}
    \item Dose response: $\theta_0^{ATE}(d)=E\{Y^{(d)}\}$ is the counterfactual mean outcome given intervention $D=d$ for the entire population.
     \item Dose response with distribution shift: $ \theta_0^{DS}(d,\tilde{\text{\normalfont pr}})=E_{\tilde{\text{\normalfont pr}}}\{Y^{(d)}\}$ is the counterfactual mean outcome given intervention $D=d$ for an alternative population with data distribution $\tilde{\text{\normalfont pr}}$. 
    \item Conditional response: $ \theta_0^{ATT}(d,d')=E\{Y^{(d')} \mid D=d\}$ is the counterfactual mean outcome given intervention $D=d'$ for the subpopulation who actually received treatment $D=d$.
     \item Heterogeneous response: $\theta_0^{CATE}(d,v)=E\{Y^{(d)} \mid V=v\}$ is the counterfactual mean outcome given intervention $D=d$ for the subpopulation with subcovariate value $V=v$.
\end{enumerate}
Likewise we define incremental functions, e.g. $\theta_0^{\nabla:ATE}(d)=E\{\nabla_d Y^{(d)}\}$ where $\nabla_d$ means $\partial / \partial d$.
\end{definition}
The superscript of each nonparametric causal function corresponds to its familiar parametric analogue. Results for means of potential outcomes immediately imply results for differences thereof. See Supplement~\ref{sec:distribution} for counterfactual distributions and Supplement~\ref{sec:graphical} for graphical models.

The  dose response curves $\theta_0^{ATE}(d)$ and $\theta_0^{DS}(d,\tilde{\text{\normalfont pr}})$ are causal functions for entire populations. 
 The second argument of $\theta_0^{DS}(d,\tilde{\text{\normalfont pr}})$ gets to the heart of external validity: though our data were drawn from population $\text{\normalfont pr}$, what would be the dose response curve for a different population $\tilde{\text{\normalfont pr}}$? For example, a job training study may be conducted in Virginia, yet we may wish to inform policy in Arkansas, a state with different demographics \cite{hotz2005predicting}. 
Predictive questions of this nature are widely studied in machine learning under the names of transfer learning, distribution shift, and covariate shift \cite{quinonero2009dataset}.

$\theta_0^{ATE}(d)$ and $\theta_0^{DS}(d,\tilde{\text{\normalfont pr}})$ are dose response curves for entire populations, but causal functions may vary for different subpopulations. Towards the goal of personalized or targeted interventions, an analyst may ask another nuanced counterfactual question: what would have been the effect of treatment $D=d'$ for the subpopulation who actually received treatment $D=d$? When treatment is continuous, we may define the conditional response $ \theta_0^{ATT}(d,d')=E\{Y^{(d')} \mid D=d\}$. 

In $\theta_0^{ATT}(d,d')$, heterogeneity is indexed by treatment $D$. Heterogeneity may instead be indexed by some interpretable covariate subvector $V$, e.g. age, race, or gender \cite{abrevaya2015estimating}. An analyst may therefore prefer to measure heterogeneous effects for subpopulations characterized by different values of $V$. For simplicity, we will write covariates as $(V,X)$ for this setting, where $X$ are additional identifying covariates besides the interpretable covariates $V$. While many works focus on the special case where treatment is binary, our definition of heterogeneous response curve $\theta_0^{CATE}(d,v)=E\{Y^{(d)} \mid V=v\}$ allows for continuous treatment.





\begin{lemma}[Identification of causal functions \cite{rosenbaum1983central,robins1986new}]\label{theorem:id_treatment}
Under standard assumptions of selection on observables and covariate shift in Supplement~\ref{sec:id},
$\theta_0^{ATE}(d)=\int \gamma_0(d,x)\mathrm{d}\text{\normalfont pr}(x)$, $\theta_0^{DS}(d,\tilde{\text{\normalfont pr}})=\int \gamma_0(d,x)\mathrm{d}\tilde{\text{\normalfont pr}}(x)$, $\theta_0^{ATT}(d,d')=\int \gamma_0(d',x)\mathrm{d}\text{\normalfont pr}(x \mid d)$, and $\theta_0^{CATE}(d,v)=\int \gamma_0(d,v,x)\mathrm{d}\text{\normalfont pr}(x \mid v)$, 
where $\gamma_0(d,x)=E(Y \mid D=d,X=x)$ and $\gamma_0(d,v,x)=E(Y \mid D=d,V=v,X=x)
$. Likewise we identify incremental functions, e.g. $\theta_0^{\nabla:ATE}(d)=\int \nabla_d \gamma_0(d,x)\mathrm{d}\text{\normalfont pr}(x)$ \cite{altonji2005cross}.
\end{lemma}

Lemma~\ref{theorem:id_treatment} clarifies the data requirements for estimating each causal function. The dose response $\theta_0^{ATE}(d)$ and conditional response $\theta_0^{ATT}(d,d')$ require observations of outcome $Y$, treatment $D$, and covariates $X$ drawn from the population $\text{\normalfont pr}$. The dose response with distribution shift $\theta_0^{DS}(d,\tilde{\text{\normalfont pr}})$ additionally requires observations of covariates $\tilde{X}$ drawn from the alternative population $\tilde{\text{\normalfont pr}}$. For the heterogeneous response $\theta_0^{CATE}(d,v)$, we abuse notation by denoting the covariates by $(V,X)$, where $V$ is the subcovariate of interest and selection is with respect to the union $(V,X)$. An analyst requires observations of $(Y,D,V,X)$ drawn from the population $\text{\normalfont pr}$.

In particular, Lemma~\ref{theorem:id_treatment} expresses each causal function as an integral of the regression function $\gamma_0$ according to a marginal or conditional distribution. As previewed in Section~\ref{sec:intro}, nonparametric estimation of $\theta_0^{CATE}(d,v)$ involves three steps: estimating a nonlinear regression $\gamma_0(d,v,x)$, which may involve many covariates $X$; estimating the conditional distribution $\text{\normalfont pr}(x \mid v)$ for reweighting; and using the latter to integrate the former. In the next section, we propose original estimators that achieve all three steps in a one line, closed form solution. 
\section{Algorithm}\label{sec:algorithm}

\subsection{RKHS background}

To present the algorithm, we provide background on the RKHS. The essential property of a function $\gamma$ in an RKHS $\mathcal{H}$ is the eponymous reproducing property: $\gamma(w)=\langle \gamma,\phi(w)\rangle_{\mathcal{H}}$ where $\phi(w)$ are features, formally defined below, that serve as the basis functions for $\mathcal{H}$. Our key algorithmic insight is to interpret the reproducing property as a way to separate the function $\gamma$ from the features $\phi(w)$. We use this defining property of the RKHS to decouple the three steps of nonparametric causal estimation. After providing RKHS background material, we prove an inner product representation that formalizes the decoupling, then introduce the causal estimators.

A scalar-valued RKHS $\mathcal{H}$ is a Hilbert space with elements that are functions $\gamma:\mathcal{W}\rightarrow\mathbb{R}$, on which the operator of evaluation is bounded \cite{berlinet2011reproducing}. Polynomial, spline, and Sobolev spaces are widely used examples of RKHSs. $\mathcal{W}$ can be any Polish space, so a value $w\in \mathcal{W}$ can be discrete or continuous. An RKHS is fully characterized by its feature map, which takes a point $w$ in the original space $\mathcal{W}$ and maps it to a feature $\phi(w)$ in the RKHS $\mathcal{H}$. The closure of $span\{\phi(w)\}_{w\in \mathcal{W}}$ is the RKHS $\mathcal{H}$. In other words, $\{\phi(w)\}_{w\in\mathcal{W}}$ can be viewed as the dictionary of basis functions for the RKHS $\mathcal{H}$. The kernel $k:\mathcal{W}\times \mathcal{W}\rightarrow \mathbb{R}$ is the inner product of features $\phi(w)$ and $\phi(w')$: $
k(w,w')=\langle \phi(w),\phi(w') \rangle_{\mathcal{H}}
$. A real-valued kernel $k$ is continuous, symmetric, and positive definite. Though we have constructed the kernel from the feature map, the Moore--Aronszajn Theorem states that, for any positive definite kernel $k$, there exists a unique RKHS $\mathcal{H}$ with feature map $\phi:w\mapsto k(w,\cdot)$. We have already seen that if $\gamma\in\mathcal{H}$, then $\gamma:\mathcal{W}\rightarrow \mathbb{R}$. With the additional notation of the feature map, we write
$
\gamma(w)= \langle \gamma,\phi(w) \rangle_{\mathcal{H}}
$. If $\mathcal{W}$ is separable and $\phi$ is continuous, then $\mathcal{H}$ is separable and may be infinite dimensional.

The RKHS is a practical hypothesis space for nonparametric regression. Consider output $Y\in\mathbb{R}$, input $W\in\mathcal{W}$, and the goal of estimating the conditional expectation function $\gamma_0(w)=E(Y \mid W=w)$. A kernel ridge regression estimator of $\gamma_0$ is
\begin{equation}\label{eq:cef_loss}
\hat{\gamma}=\argmin_{\gamma \in\mathcal{H}} n^{-1}\sum_{i=1}^n \{Y_i-\langle\gamma,\phi(W_i)\rangle_{\mathcal{H}}\}^2 + \lambda \|\gamma\|^2_{\mathcal{H}}.
\end{equation}
$\lambda>0$ is a hyperparameter on the ridge penalty $\|\gamma\|^2_{\mathcal{H}}$, which imposes smoothness in estimation. The solution to the optimization problem has a well known closed form \cite{kimeldorf1971some}, which we exploit and generalize throughout this work:
\begin{equation}\label{eq:cef_form}
\hat{\gamma}(w)=Y^{\top}(K_{WW}+n\lambda  I )^{-1}K_{Ww}.
\end{equation}
The closed form solution involves the kernel matrix $K_{WW}\in\mathbb{R}^{n\times n}$ with $(i,j)$th entry $k(W_i,W_j)$, and the kernel vector $K_{Ww}\in\mathbb{R}^n$ with $i$th entry $k(W_i,w)$. To tune the ridge hyperparameter $\lambda$, both generalized cross validation and leave-one-out cross validation have closed form solutions, and the former is asymptotically optimal \cite{craven1978smoothing,li1986asymptotic}.

We have seen that the feature map takes a value in the original space $w\in\mathcal{W}$ and maps it to a feature in the RKHS $\phi(w)\in\mathcal{H}$. Now we generalize this idea, from the embedding of a value $w$ to the embedding of a distribution $\text{\normalfont q}$. Just as a value $w$ in the original space  is embedded as an element $\phi(w)$ in the RKHS, so too the distribution $\text{\normalfont q}$ over the original space can be embedded as an element  $\mu=E_{\text{\normalfont q}}\{\phi(W)\}$ in the RKHS \cite{smola2007hilbert,berlinet2011reproducing}. Boundedness of the kernel implies existence of the mean embedding as well as Bochner integrability, which permits us to exchange the expectation and inner product. Mean embeddings facilitate the evaluation of expectations of RKHS functions: for $\gamma \in \mathcal{H}$,
$E_{\text{\normalfont q}}\{\gamma(W)\}=E_{\text{\normalfont q}}\left\{\langle \gamma,\phi(W) \rangle_{\mathcal{H}}\right\}=\langle \gamma,\mu \rangle_{\mathcal{H}}$. The final expression foreshadows how we will use the technique of mean embeddings to decouple the nonparametric regression step from the nonparametric reweighting step in the estimation of causal functions. A natural question is whether the embedding $\text{\normalfont q}\mapsto E_{\text{\normalfont q}}\{\phi(W)\}$ is injective, i.e. whether the RKHS element representation is unique. This is called the characteristic property of the kernel $k$, and it holds for commonly used RKHSs e.g. the exponentiated quadratic kernel \cite{sriperumbudur2010relation}.

The tensor product RKHS is one way to construct an RKHS for functions with multiple arguments. Consider the RKHSs $\mathcal{H}_{1}$ and $\mathcal{H}_{2}$ with positive definite kernels $k_{1}:\mathcal{W}_1\times\mathcal{W}_1\rightarrow \mathbb{R}$ and  
$k_{2}:\mathcal{W}_2\times\mathcal{W}_2\rightarrow \mathbb{R}$, respectively. An element $\gamma_1\in \mathcal{H}_{1}$ is a function $\gamma_1:\mathcal{W}_1\rightarrow \mathbb{R}$ and an element $\gamma_2\in \mathcal{H}_{2}$ is a function $\gamma_2:\mathcal{W}_2\rightarrow \mathbb{R}$. The tensor product RKHS $\mathcal{H}=\mathcal{H}_{1}\otimes \mathcal{H}_{2}$ is the RKHS with the product kernel
$
k:(\mathcal{W}_1\times \mathcal{W}_2) \times (\mathcal{W}_1\times \mathcal{W}_2)\rightarrow \mathbb{R},\; \{(w_1,w_2),(w'_1,w'_2)\}\mapsto k_{1}(w_1,w_1')  k_{2}(w_2,w_2')
$. Equivalently, the tensor product RKHS $\mathcal{H}$ has feature map $\phi(w_1)\otimes \phi(w_2)$ such that $\|\phi(w_1)\otimes \phi(w_2)\|_{\mathcal{H}}=\|\phi(w_1)\|_{\mathcal{H}_1}\|\phi(w_2)\|_{\mathcal{H}_2}$. Formally, tensor product notation means $(a\otimes b)c=a \langle b,c\rangle$. An element of the tensor product RKHS $\gamma \in\mathcal{H}$ is a function $\gamma:\mathcal{W}_1\times \mathcal{W}_2 \rightarrow \mathbb{R}$. We assume that the regression function $\gamma_0(w_1,w_2)=E(Y \mid W=w_1,w_2=w_2)$ is an element of a tensor product RKHS, i.e. $\gamma_0\in \mathcal{H}$. As such, the different arguments of $\gamma_0$ are decoupled, which we exploit when calculating partial means. 

Finally, we introduce the RKHS $\mathcal{L}_2(\mathcal{H}_{1},\mathcal{H}_{2})$ that we employ for conditional expectation operators. Rather than being a space of real-valued functions, it is a space of Hilbert--Schmidt operators from one RKHS to another. If the operator $E$ is an element of $ \mathcal{L}_2(\mathcal{H}_{1},\mathcal{H}_{2})$, then $E:\mathcal{H}_{1} \rightarrow \mathcal{H}_{2}$. Formally, it can be shown that $\mathcal{L}_2(\mathcal{H}_{1},\mathcal{H}_{2})$ is an RKHS in its own right with an appropriately defined kernel and feature map. $\mathcal{L}_2(\mathcal{H}_{1},\mathcal{H}_{2})$ is an example of a vector-valued RKHS; see \cite{micchelli2005learning} for a more general discussion. In the present work, we assume the conditional expectation operator $E_0:\gamma_1(\cdot)\mapsto E\{\gamma_1(W_1) \mid W_2=\cdot\}$ is an element of this RKHS, i.e. $E_0\in \mathcal{L}_2(\mathcal{H}_{1},\mathcal{H}_{2})$. We estimate $E_0$ by a kernel ridge regression in $\mathcal{L}_2(\mathcal{H}_{1},\mathcal{H}_{2})$, which coincides with estimating the conditional mean embedding $\mu_{w_1}(w_2)=E\{\phi(W_1) \mid W_2=w_2\}$ via the kernel ridge regression of $\phi(W_1)$ on $\phi(W_2)$; see the derivation of Algorithm~\ref{algorithm:treatment} below.

\subsection{Decoupled representation}

Lemma~\ref{theorem:id_treatment} makes precise how each causal function is identified as a partial mean of the form $\int \gamma_0(d,x)\mathrm{d}\text{\normalfont q}$ for some distribution $\text{\normalfont q}$. To facilitate estimation, we now assume that $\gamma_0$ is an element of an RKHS. 
In our construction, we define scalar valued RKHSs for treatment $D$ and covariates $(V,X)$, then assume that the regression is an element of the tensor product space. Let $k_{\mathcal{D}}:\mathcal{D}\times \mathcal{D}\rightarrow \mathbb{R}$, $k_{\mathcal{V}}:\mathcal{V}\times \mathcal{V}\rightarrow \mathbb{R}$, and $k_{\mathcal{X}}:\mathcal{X}\times \mathcal{X}\rightarrow \mathbb{R}$ be measurable positive definite kernels corresponding to scalar valued RKHSs $\mathcal{H}_{\mathcal{D}}$, $\mathcal{H}_{\mathcal{V}}$, and $\mathcal{H}_{\mathcal{X}}$. Denote the feature maps
$
\phi_{\mathcal{D}}:\mathcal{D}\rightarrow \mathcal{H}_{\mathcal{D}}, \; d\mapsto k_{\mathcal{D}}(d,\cdot ); \;
\phi_{\mathcal{V}}:\mathcal{V}\rightarrow \mathcal{H}_{\mathcal{V}}, \; v\mapsto k_{\mathcal{V}}(v,\cdot );
 \; \phi_{\mathcal{X}}:\mathcal{X}\rightarrow \mathcal{H}_{\mathcal{X}}, \; x\mapsto k_{\mathcal{X}}(x,\cdot )
$.
To lighten notation, we suppress subscripts when arguments are provided. 

For $\theta_0^{ATE}$, $\theta_0^{DS}$, and $\theta_0^{ATT}$, we assume the regression $\gamma_0$ is an element of the RKHS $\mathcal{H}$ with the kernel $k(d,x;d',x')=k_{\mathcal{D}}(d,d')k_{\mathcal{X}}(x,x')$. We appeal to the fact that the product of positive definite kernels for $\mathcal{H}_{\mathcal{D}}$ and $\mathcal{H}_{\mathcal{X}}$ defines a new positive definite kernel for $\mathcal{H}$. The product construction provides a rich composite basis; $\mathcal{H}$ has the tensor product feature map $\phi(d)\otimes \phi(x)$ and $\mathcal{H}=\mathcal{H}_{\mathcal{D}}\otimes \mathcal{H}_{\mathcal{X}}$. In this RKHS,
$
\gamma_0(d,x)=\langle \gamma_0, \phi(d)\otimes \phi(x)\rangle_{\mathcal{H}} 
$.
Likewise for $\theta_0^{CATE}$ we assume $\gamma_0\in \mathcal{H}=\mathcal{H}_{\mathcal{D}}\otimes\mathcal{H}_{\mathcal{V}}\otimes  \mathcal{H}_{\mathcal{X}}$. We place regularity conditions on this RKHS construction in order to represent causal functions as inner products in $\mathcal{H}$. In anticipation of counterfactual distributions in Supplement~\ref{sec:distribution}, we also include conditions for an outcome RKHS in parentheses.

\begin{assumption}[RKHS regularity conditions]\label{assumption:RKHS}
Assume 
\begin{enumerate}
    \item $k_{\mathcal{D}}$, $k_{\mathcal{V}}$, $k_{\mathcal{X}}$ (and $k_{\mathcal{Y}}$) are continuous and bounded. Formally,
    $
    \sup_{d\in\mathcal{D}}\| \phi(d)\|_{\mathcal{H}_{\mathcal{D}}}\leq \kappa_d$, $ \sup_{v\in\mathcal{V}}\|\phi(v)\|_{\mathcal{H}_{\mathcal{V}}}\leq \kappa_v$, $ \sup_{x\in\mathcal{X}}\|\phi(x)\|_{\mathcal{H}_{\mathcal{X}}}\leq \kappa_x
    $ \{and $ \sup_{y\in\mathcal{Y}}\|\phi(y)\|_{\mathcal{H}_{\mathcal{Y}}}\leq \kappa_y$\}.
    \item $\phi(d)$, $\phi(v)$, $\phi(x)$ \{and $\phi(y)$\} are measurable.
    \item $k_{\mathcal{X}}$ (and $k_{\mathcal{Y}}$) are characteristic.
\end{enumerate}
For incremental functions, further assume $\mathcal{D}\subset \mathbb{R}$ is an open set and $\nabla_d  \nabla_{d'} k_{\mathcal{D}}(d,d')$ exists and is continuous, hence $\sup_{d\in\mathcal{D}}\|\nabla_d\phi(d)\|_{\mathcal{H}}\leq \kappa_d'$.
\end{assumption}
Commonly used kernels are continuous and bounded. Measurability is a similarly weak condition. The characteristic property ensures injectivity of the mean embeddings.

\begin{theorem}[Decoupling via kernel mean embeddings]\label{theorem:representation_treatment}
Suppose the conditions of Lemma~\ref{theorem:id_treatment}, Assumption~\ref{assumption:RKHS}, and $\gamma_0\in\mathcal{H}$ hold. Then
\begin{enumerate}
    \item $\theta_0^{ATE}(d)=\langle \gamma_0, \phi(d)\otimes \mu_x\rangle_{\mathcal{H}} $ where $\mu_x=\int\phi(x) \mathrm{d}\text{\normalfont pr}(x) $;
    \item $\theta_0^{DS}(d,\tilde{\text{\normalfont pr}})=\langle \gamma_0, \phi(d)\otimes \nu_x\rangle_{\mathcal{H}} $ where $\nu_x=\int\phi(x) \mathrm{d}\tilde{\text{\normalfont pr}}(x) $;
    \item $\theta_0^{ATT}(d,d')=\langle \gamma_0, \phi(d')\otimes \mu_x(d)\rangle_{\mathcal{H}} $ where $\mu_x(d)=\int\phi(x) \mathrm{d}\text{\normalfont pr}(x \mid d)$;
    \item $\theta_0^{CATE}(d,v)=\langle \gamma_0, \phi(d)\otimes \phi(v)\otimes \mu_{x}(v)\rangle_{\mathcal{H}} $ where $\mu_{x}(v)= \int \phi(x) \mathrm{d}\text{\normalfont pr}(x \mid v)$.
\end{enumerate}
Likewise for incremental functions, e.g. $\theta_0^{\nabla:ATE}(d)=\langle \gamma_0, \nabla_d\phi(d)\otimes \mu_x\rangle_{\mathcal{H}} $.
\end{theorem}

\begin{proof}[Sketch] 
Consider $\theta_0^{CATE}(d,v)$. Boundedness of the kernel implies Bochner integrability, which allows us to exchange the integral and inner product:
\begin{align*}
    \int \gamma_0(d,v,x)\mathrm{d}\text{\normalfont pr}(x \mid v)=\int \langle \gamma_0, \phi(d)\otimes \phi(v)\otimes \phi(x)\rangle_{\mathcal{H}}  \mathrm{d}\text{\normalfont pr}(x \mid v) =\langle \gamma_0, \phi(d)\otimes \mu_x(v) \rangle_{\mathcal{H}}.
\end{align*}
\end{proof}
See Supplement~\ref{sec:derivation} for the full proof. $\mu_x(v)=\int\phi(x) \text{\normalfont pr}(x \mid v)$ is the mean embedding of the conditional distribution $\text{\normalfont pr}(x \mid v)$. It encodes the distribution $\text{\normalfont pr}(x \mid v)$ as a function $\mu_x(v)\in\mathcal{H}_{\mathcal{X}}$ such that the causal function $\theta_0^{CATE}(d,v)$ can be expressed as an inner product in $\mathcal{H}$.

\subsection{Closed form solution}

The representation in Theorem~\ref{theorem:representation_treatment}  is essential to the algorithm derivation. In particular, the representation cleanly separates the three steps necessary to estimate a causal function: estimating a nonlinear regression, which may involve many covariates; estimating the distribution for reweighting; and using the nonparametric distribution to integrate the nonparametric regression. For example, for $\theta_0^{CATE}(d,v)$, our estimator is $\hat{\theta}^{CATE}(d,v)=\langle \hat{\gamma}, \phi(d)\otimes \phi(v)\otimes \hat{\mu}_x(v)\rangle_{\mathcal{H}}$. The nonlinear regression estimator $\hat{\gamma}$ is a standard kernel ridge regression of $Y$ on $\phi(D)\otimes \phi(V)\otimes \phi(X)$; the reweighting distribution estimator $\hat{\mu}_x(v)$ is a generalized kernel ridge regression of $\phi(X)$ on $\phi(V)$; and the latter can be used to integrate the former by simply multiplying the two. This algorithmic insight is a key innovation of the present work, and the reason why our estimators have simple closed form solutions despite complicated causal integrals.
\begin{algorithm}[Estimation of causal functions]\label{algorithm:treatment}
Denote the empirical kernel matrices
$
K_{DD}, K_{VV}, K_{XX}\in\mathbb{R}^{n\times n}
$ calculated from observations drawn from population $\text{\normalfont pr}$. Let $\tilde{X}_i$ $(i=1,...,\tilde{n})$ be observations drawn from population $\tilde{\text{\normalfont pr}}$. Denote by $\odot$ the elementwise product. Causal function estimators have the closed form solutions
\begin{enumerate}
    \item $\hat{\theta}^{ATE}(d)=n^{-1}\sum_{i=1}^n Y^{\top}(K_{DD}\odot K_{XX}+n\lambda  I )^{-1}(K_{Dd}\odot K_{Xx_i})  $;
     \item $\hat{\theta}^{DS}(d,\tilde{\text{\normalfont pr}})=\tilde{n}^{-1}\sum_{i=1}^{\tilde{n}} Y^{\top}(K_{DD}\odot K_{XX}+n\lambda  I )^{-1}(K_{Dd}\odot K_{X\tilde{x}_i}) $;
    \item $\hat{\theta}^{ATT}(d,d')=Y^{\top}(K_{DD}\odot K_{XX}+n\lambda  I )^{-1}[K_{Dd'}\odot \{K_{XX}(K_{DD}+n\lambda_1  I )^{-1}K_{Dd}\}]$;
    \item $\hat{\theta}^{CATE}(d,v)=Y^{\top}(K_{DD}\odot K_{VV}\odot K_{XX} +n\lambda  I )^{-1}[K_{Dd}\odot K_{Vv}\odot \{K_{XX}(K_{VV}+n\lambda_2  I )^{-1}K_{Vv} \}] $;
\end{enumerate}
where $(\lambda,\lambda_1,\lambda_2)$ are ridge regression penalty hyperparameters.  Likewise for incremental functions, e.g. $\hat{\theta}^{\nabla:ATE}(d)=n^{-1}\sum_{i=1}^n Y^{\top}(K_{DD}\odot K_{XX}+n\lambda  I )^{-1}(\nabla_d K_{D{d}}\odot K_{Xx_i})  $ where $(\nabla_d  K_{D{d}})_i=\nabla_d k(D_i,d)$.
\end{algorithm}

\begin{proof}[Sketch]
Consider $\theta_0^{CATE}(d,v)$. Analogously to~\eqref{eq:cef_loss}, the kernel ridge regression estimators of the regression $\gamma_0$ and the conditional mean embedding $\mu_x(v)$ are given by
\begin{align*}
    \hat{\gamma}&=\argmin_{\gamma \in\mathcal{H}} n^{-1}\sum_{i=1}^n \{Y_i-\langle\gamma, \phi(D_i)\otimes \phi(V_i) \otimes\phi (X_i)\rangle_{\mathcal{H}}\}^2 + \lambda \|\gamma\|^2_{\mathcal{H}}, \\
    \hat{E}&=\argmin_{E\in\mathcal{L}_2(\mathcal{H}_{\mathcal{X}},\mathcal{H}_{\mathcal{V}})} n^{-1}\sum_{i=1}^n \{\phi(X_i)-E^*\phi(V_i)\}^2 + \lambda_2 \|E\|^2_{\mathcal{L}_2(\mathcal{H}_{\mathcal{X}},\mathcal{H}_{\mathcal{V}})},
\end{align*}
where $\hat{\mu}_x(v)=\hat{E}^*\phi(v)$ and $E^*$ is the adjoint of $E$.
Analogously to~\eqref{eq:cef_form}, the closed forms are
\begin{align*}
    \hat{\gamma}(d,v,\cdot)&=Y^{\top}(K_{DD}\odot K_{VV}\odot K_{XX}+n\lambda  I )^{-1}\{K_{Dd}\odot K_{Vv}\odot K_{X(\cdot)}\},\\
    [\hat{\mu}_x(v)](\cdot)&=K_{(\cdot) X}(K_{VV}+n\lambda_2  I )^{-1}K_{Vv}.
\end{align*}
To arrive at the main result, match the empty arguments $(\cdot)$ of the kernel ridge regressions.
\end{proof}
See Supplement~\ref{sec:derivation} for the full derivation and a comparison to series estimation. We give theoretical values for $(\lambda,\lambda_1,\lambda_2)$ that optimally balance bias and variance in Theorem~\ref{theorem:consistency_treatment} below. Supplement~\ref{sec:tuning} gives practical tuning procedures based on generalized and leave-one-out cross validation to empirically balance bias and variance, the former of which is asymptotically optimal. 

\section{Comparison to kernel methods for causal scalars}\label{sec:detail}


We now connect our kernel methods for causal functions with related kernel methods for treatment effects. Recall the definition $\theta_0^{ATE}(d)=E\{Y^{(d)}\}$. We allow treatment to be continuous, so $\theta_0^{ATE}$ is a causal function called the dose response. In related work, treatment is binary, so $\theta_0^{ATE}$ is a vector of two causal scalars $\theta^{ATE}_0(1),\theta^{ATE}_0(0)$ whose difference is the treatment effect. 

We clarify three points. (i) There is a sense in which our algorithms generalize known estimators for treatment effects to new estimators for causal functions. (ii) A treatment effect is a bounded functional over $\mathbb{L}^2$ with a balancing weight representation, while a response curve is not. Our key insight is that a response curve is a bounded functional over the RKHS $\mathcal{H}$, which is a subset of $\mathbb{L}^2$. (iii) Our theoretical contribution is a new, uniform analysis of response curves. The analysis goes beyond pointwise approximation of response curves by local treatment effects.


We begin by reviewing the theory of balancing weights, which are popular in causal inference with binary treatments. For clarity, in this section we emphasize a fixed treatment value by writing $d^*\in \mathcal{D}$.  The following representation is well known.

\begin{proposition}[Existence for treatment effects; Point 3.1 of \cite{hernan2010causal}]\label{prop:balance_exists}
Suppose selection on observables (stated in Supplement~\ref{sec:id}) holds and treatment is binary. Fix $d^*\in\mathcal{D}$. If $\text{\normalfont pr}(D=d^* \mid X)$ is bounded away from zero almost surely, then there exists balancing weight $\alpha_0\in \mathbb{L}^2$ such that for all $\gamma\in \mathbb{L}^2$, $\int \gamma(d^*,x)\mathrm{d}\text{\normalfont pr}(x)=\langle \gamma,\alpha_0 \rangle_{\mathbb{L}^2}$. In particular, $\theta_0^{ATE}(d^*)=\int y\alpha_0(d,x)\mathrm{d}\text{\normalfont pr}(d,x,y)=\langle \gamma_0,\alpha_0 \rangle_{\mathbb{L}^2}$ and the balancing weight is $\alpha_0(d,x)=1(d=d^*)/\text{\normalfont pr}(D=d^* \mid x)$.
\end{proposition}

In summary, a treatment effect has two representations: the primal representation of Lemma~\ref{theorem:id_treatment} as a partial mean of the regression $\gamma_0(d,x)=E(Y \mid D=d,X=x)$, and the dual representation of Proposition~\ref{prop:balance_exists} as a reweighting of the outcome $Y$ using the balancing weight $\alpha_0(d,x)=1(d=d^*)/\text{\normalfont pr}(D=d^* \mid x)$. Clearly, the two representations are related by the law of iterated expectations. Moreover, from the closed form of $\alpha_0$, we require $\text{\normalfont pr}(D=d^* \mid X)>0$ for $\alpha_0$ to exist. This property keenly relies on the treatment being discrete. Indeed, it is well known that a balancing weight representation does not exist for response curves.

\begin{proposition}[Non-existence for response curves \cite{van1991differentiable,newey1994asymptotic}]\label{prop:balance_dne}
Suppose selection on observables (stated in Supplement~\ref{sec:id}) holds and treatment is continuous. Fix $d^*\in\mathcal{D}$. Even if the density $f(d^* \mid X)$ is bounded away from zero almost surely, there does not exist a balancing weight $\alpha_0\in \mathbb{L}^2$ such that for all $\gamma\in \mathbb{L}^2$, $\int \gamma(d^*,x)\mathrm{d}\text{\normalfont pr}(x)=\langle \gamma,\alpha_0 \rangle_{\mathbb{L}^2}$. In particular, without further restrictions, there does not exist $\alpha_0\in \mathbb{L}^2$ such that  $\theta_0^{ATE}(d^*)=\int y\alpha_0(d,x)\mathrm{d}\text{\normalfont pr}(d,x,y)=\langle \gamma_0,\alpha_0 \rangle_{\mathbb{L}^2}$.
\end{proposition}

 Whereas a binary treatment effect is a bounded functional over $\mathbb{L}^2$ with a balancing weight representation, a dose response is not a bounded functional over $\mathbb{L}^2$ and does not have a balancing weight representation in the classic sense. From a functional analytic perspective, this discrepancy is the reason why the problems we study are nonparametric whereas previous work on kernel methods for treatment effects are semiparametric. See Supplement~\ref{sec:balancing} for discussion.

Our key insight is that the dose response is a bounded functional over the RKHS $\mathcal{H}$, which is a subset of $\mathbb{L}^2$. This fact follows from three simple observations: (i) the dose response is a partial mean; (ii) in the RKHS, a partial mean can be reformulated as a kind of evaluation; and (iii) the RKHS $\mathcal{H}$ is the subset of $\mathbb{L}^2$ for which evaluation is a bounded functional. Through this lens, Theorem~\ref{theorem:representation_treatment} shows that there can exist a function $\tilde{\alpha}_0\in\mathcal{H}$ such that $\theta_0^{ATE}(d)=\langle \gamma_0,\tilde{\alpha}_0 \rangle_{\mathcal{H}}$ even when there does not exist a function $\alpha_0\in \mathbb{L}^2$ such that $\theta_0^{ATE}(d)=\langle \gamma_0,\alpha_0 \rangle_{\mathbb{L}^2}$.


What is the relationship between between our kernel methods for causal functions and existing kernel methods for treatment effects? There is a sense in which our dose response estimator, which is the simplest case of our framework, is a relaxation of kernel balancing weight estimators from binary treatment to continuous treatment. We formalize this connection as follows.

\begin{corollary}[Relaxation of balancing weight estimators]\label{cor:connect}
Suppose treatment is binary, and take $k_{\mathcal{D}}(d,d')=1(d=d')$ to be the treatment kernel. Then $\hat{\theta}^{ATE}(d)=n^{-1}\sum_{i=1}^n Y_i \hat{\alpha}_i$, where $\hat{\alpha}_i=\hat{\alpha}(D_i,X_i)$ and $\hat{\alpha}$ is a ridge regularized estimator of $\alpha_0 \in \mathbb{L}^2$.
\end{corollary}

See Supplement~\ref{sec:balancing} for the proof. The balancing weight estimator $\hat{\alpha}$ minimizes a generalized balancing weight loss with ridge regularization; see \cite[eq. 8]{kallus2020generalized}, \cite[eq. 1]{hirshberg2019minimax}, and \cite[Definition 3.2]{singh2021debiased} for various formulations. Corollary~\ref{cor:connect} provides intuition for our tensor product RKHS construction. Our product kernel construction ensures that using the binary treatment kernel amounts to subsetting, which recovers previous algorithms. The tensor product RKHS provides a natural way to relax binary treatment to continuous treatment while retaining computational tractability.

As argued in Proposition~\ref{prop:balance_dne}, the balancing weight $\alpha_0\in \mathbb{L}^2$ does not exist for the dose response. Nonetheless, our key insight in Theorem~\ref{theorem:representation_treatment} is that a function $\tilde{\alpha}_0\in\mathcal{H}$ does exist to serve a similar purpose. By combining the partial mean perspective with the technique of kernel mean embedding, we demonstrate that our framework easily extends to conditional nonparametric causal functions, e.g. the heterogeneous response curve $\theta_0^{CATE}(d,v)$, which are substantially more challenging than unconditional nonparametric causal functions, e.g. the dose response $\theta^{ATE}_0(d)$.

Perhaps the most surprising consequence of our construction is the closed form solution for causal functions. In particular, each closed form solution is a reweighting of the observed outcomes with empirical weights that we characterize even though a population balancing weight in $\mathbb{L}^2$ does not exist. 
In sum, previous work \cite{kallus2020generalized,hirshberg2019minimax,singh2021debiased} on kernel methods for treatment effects is inherently tied to the $\mathbb{L}^2$ population balancing weight perspective; our algorithms apply the conceptual framework of kernel methods to new classes of causal functions. The following corollary reinterprets Algorithm~\ref{algorithm:treatment} through this lens. 

\begin{corollary}[Closed form reweighting even when balancing weight does not exist]\label{cor:extend}
Suppose treatment is continuous, with $k_{D}$ that is continuous and bounded. Then 
$\hat{\theta}^{ATE}(d)=n^{-1}\sum_{i=1}^n Y_i\hat{\alpha}_i^{ATE}$,
$\hat{\theta}^{DS}(d,\tilde{\text{\normalfont pr}})=n^{-1}\sum_{i=1}^{n}Y_i\hat{\alpha}_i^{DS}$,  $\hat{\theta}^{ATT}(d,d')=n^{-1}\sum_{i=1}^n Y_i\hat{\alpha}_i^{ATT}$, and  
$\hat{\theta}^{CATE}(d,v)=n^{-1}\sum_{i=1}^n Y_i\hat{\alpha}_i^{CATE} $, where the weights have closed form solutions given in Supplement~\ref{sec:balancing}. Likewise for incremental functions, e.g. $\hat{\theta}^{\nabla:ATE}(d)=n^{-1}\sum_{i=1}^n Y_i\hat{\alpha}_i^{\nabla:ATE}$.
\end{corollary}


Each of our proposed causal function estimators is global. In particular, within Corollary~\ref{cor:extend}, the weights $(\hat{\alpha}_j^{ATE},\hat{\alpha}_j^{DS},\hat{\alpha}_j^{ATT},\hat{\alpha}_j^{CATE})$ $(j=1,...,n)$ depend on all of the observations as refracted through the ridge regularized empirical covariance and the kernel evaluations $k(D_i,d)$. This approach departs from a localization approach to causal functions whereby the weight assigned to each observation is determined by Nadaraya--Watson smoothing \cite{kennedy2017nonparametric,kallus2018policy,colangelo2020double,chernozhukov2022debiased}. In the localization approach, the weight is $k^{NW}\{(D_i-d)/h\}$ where $k^{NW}$ is a Nadaraya--Watson kernel and $h$ is a vanishing bandwidth. By contrast, we consider a fixed kernel and vanishing ridge regularization.

The global perspective has three main advantages. First, our estimators can be computed once and evaluated at any value of a continuous treatment.  By contrast, a localized estimator is a computationally intensive procedure that must be reimplemented at any treatment value. Second, our estimators are constructed from function classes with designed-in smoothness properties, which leads to smoother and therefore more plausible response curves. We compare our smooth estimate with a jagged localizing estimate in the program evaluation of Supplement~\ref{sec:experiments}. Third, we prove uniform consistency of response curves, whereas localizations of previous results would only lead to pointwise consistency. These uniform guarantees are the focus of the next section.

\section{Uniform consistency}\label{sec:consistency}

\subsection{RKHS background}

In Section~\ref{sec:problem}, we defined the causal functions of interest, and identified them as partial means. In Section~\ref{sec:algorithm}, we introduced the tensor product RKHS as the function space in which the three steps of nonparametric causal estimation may be decoupled. We then proposed estimators based on kernel ridge regression with closed form solutions. In Section~\ref{sec:detail}, we demonstrated that our estimators generalize known estimators for the binary treatment case. In this section, we prove uniform consistency of the estimators, with finite sample rates that combine minimax optimal rates. To do so, we define our key approximation assumptions, which are standard in RKHS learning theory: smoothness and spectral decay.

To state our key assumptions, we must introduce a certain eigendecomposition. Recall the example of a generic RKHS $\mathcal{H}$ with kernel $k:\mathcal{W}\times \mathcal{W}\rightarrow \mathbb{R}$ consisting of functions $\gamma:\mathcal{W}\rightarrow \mathbb{R}$. Let $\nu$ be any Borel measure on $\mathcal{W}$. We denote by $\mathbb{L}^2_{\nu}(\mathcal{W})$ the space of square integrable functions with respect to measure $\nu$. Given the kernel, define the integral operator
$
L:\mathbb{L}_{\nu}^2(\mathcal{W})\rightarrow \mathbb{L}_{\nu}^2(\mathcal{W}),\; \gamma \mapsto \int k(\cdot,w)\gamma(w)\mathrm{d}\nu(w)
$. 
If the kernel $k$ is defined on $\mathcal{W}\subset\mathbb{R}^d$ and shift invariant, then $L$ is a convolution of $k$ and $\gamma$. If $k$ is smooth, then $L\gamma$ is a smoothed version of $\gamma$. $L$ is a self adjoint, positive, compact operator, so by the spectral theorem we can denote its countable eigenvalues by $(\eta_j)$ and its countable eigenfunctions, which are equivalence classes, by $\{(\varphi_j)_{\nu}\}$:
$$
L\gamma =\sum_{j=1}^{\infty} \eta_j\langle (\varphi_j)_{\nu},\gamma \rangle_{\mathbb{L}^2_{\nu}(\mathcal{W})}  (\varphi_j)_{\nu},\quad (\varphi_j)_{\nu}=\{f:\nu(f\neq \varphi_j)=0\}.
$$
Without loss of generality, $\eta_j\geq \eta_{j+1}$, and these are also the eigenvalues of the feature covariance operator $T=E\{\phi(W)\otimes \phi(W)\}$. For simplicity, we assume $(\eta_j)>0$ in this discussion; see \cite[Remark 3]{cucker2002mathematical} for the more general case. $\{(\varphi_j)_{\nu}\}$ form an orthonormal basis of $\mathbb{L}_{\nu}^2(\mathcal{W})$. By the generalized Mercer's Theorem for Polish spaces \cite[Corollary 3.5]{steinwart2012mercer}, we can express the kernel as $k(w,w')=\sum_{j=1}^{\infty}\eta_j \varphi_j(w)\varphi_j(w')$, where $(w,w')$ are in the support of $\nu$, $\varphi_j$ is a continuous element in the equivalence class $(\varphi_j)_{\nu}$, and the convergence is absolute and uniform.

With this notation, we express $\mathbb{L}_{\nu}^2(\mathcal{W})$ and the RKHS $\mathcal{H}$ in terms of the series $\{(\varphi_j)_{\nu}\}$. If $\gamma\in \mathbb{L}_{\nu}^2(\mathcal{W})$, then $\gamma$ can be uniquely expressed as
$
\gamma=\sum_{j=1}^{\infty}\gamma_j(\varphi_j)_{\nu}
$
and the partial sums $\sum_{j=1}^J \gamma_j (\varphi_j)_{\nu}$ converge to $\gamma$ in $\mathbb{L}^2_{\nu}(\mathcal{W})$. Indeed, for $\gamma=\sum_{j=1}^{\infty}\gamma_j(\varphi_j)_{\nu}$ and $\gamma'=\sum_{j=1}^{\infty}\gamma_j'(\varphi_j)_{\nu}$,
$$
 \mathbb{L}^2_{\nu}(\mathcal{W})=\left\{\gamma=\sum_{j=1}^{\infty}\gamma_j(\varphi_j)_{\nu}:\; \sum_{j=1}^{\infty}\gamma_j^2<\infty\right\},\quad \langle \gamma,\gamma' \rangle_{\mathbb{L}^2_{\nu}(\mathcal{W})}=\sum_{j=1}^{\infty} \gamma_j\gamma_j'.
$$
By \cite[Theorem 4]{cucker2002mathematical}, the RKHS $\mathcal{H}$ can be explicitly represented as
$$
\mathcal{H}=\left(\gamma=\sum_{j=1}^{\infty}\gamma_j\varphi_j:\;\sum_{j=1}^{\infty} \frac{\gamma_j^2}{\eta_j}<\infty\right),\quad \langle \gamma,\gamma' \rangle_{\mathcal{H}}=\sum_{j=1}^{\infty} \frac{\gamma_j\gamma_j'}{\eta_j}.
$$
To interpret this result, recall that $(\eta_j)$ is a weakly decreasing sequence. The RKHS $\mathcal{H}$ is the subset of functions in $\mathbb{L}^2_{\nu}(\mathcal{W})$ which are continuous and for which higher order terms in the series $\{(\varphi_j)_{\nu}\}$ have a smaller contribution. The RKHS inner product penalizes higher order coefficients, and the magnitude of the penalty corresponds to how small the eigenvalue is.

We have seen how to conduct kernel ridge regression with the RKHS $\mathcal{H}$. To analyze the bias from ridge regularization, we place a smoothness assumption called the source condition on the regression function $\gamma_0(w)=E(Y \mid W=w)$ \cite{smale2007learning,caponnetto2007optimal,carrasco2007linear}. 
Formally, we place assumptions of the form
\begin{equation}\label{eq:prior}
    \gamma_0\in \mathcal{H}^c=\left(f=\sum_{j=1}^{\infty}\gamma_j\varphi_j:\;\sum_{j=1}^{\infty} \frac{\gamma_j^2}{\eta^c_j}<\infty\right)\subset \mathcal{H},\quad c\in(1,2].
\end{equation}
While $c=1$ recovers correct specification $\gamma_0\in \mathcal{H}$, $c\in(1,2]$ is a stronger condition: $\gamma_0$ is a particularly smooth element of $\mathcal{H}$, well approximated by the leading terms in the series $\{(\varphi_j)_{\nu}\}$. Smoothness delivers uniform consistency. A larger value of $c$ corresponds to a smoother target $\gamma_0$ and a faster convergence rate for $\hat{\gamma}$. Rates do not further improve for $c>2$, which is known as the saturation effect for ridge regularization. 



To analyze the variance of kernel ridge regression, we place a spectral decay assumption called the effective dimension of the basis $(\varphi_j)$ for the RKHS $\mathcal{H}$. 
To obtain faster convergence rates, we place a direct assumption on the rate at which the eigenvalues $(\eta_j)$, and hence the importance of the eigenfunctions $(\varphi_j)$, decay: we assume there exists some constant $C$ such that for all $j$
\begin{equation}\label{eq:b}
  \eta_j\leq C j^{-b},\quad b\geq 1.
\end{equation} 
A bounded kernel, which we have already assumed, implies $b=1$ \cite[Lemma 10]{fischer2017sobolev}. The limit $b\rightarrow \infty$ may be interpreted as a finite dimensional RKHS \cite{caponnetto2007optimal}. For intermediate values of $b$, the polynomial rate of spectral decay quantifies the effective dimension of the RKHS $\mathcal{H}$ in light of the measure $\nu$. Intuitively, a higher value of $b$ corresponds to a lower effective dimension and a faster convergence rate for $\hat{\gamma}$.

For intuition, we relate the source condition and effective dimension to a familiar notion of smoothness in the Sobolev space. The restriction that defines an RKHS generalizes higher order smoothness in a Sobolev space. Indeed, certain Sobolev spaces are RKHSs. Let $\mathcal{W}\subset \mathbb{R}^p$. Denote by $\mathbb{H}_2^s$ the Sobolev space with $s>p/2$ derivatives that are square integrable. This space can be generated by the Mat\`ern kernel, which converges to the popular exponentiated quadratic kernel as $s\rightarrow \infty$. Suppose $\mathcal{H}=\mathbb{H}_2^s$ is chosen as the RKHS for estimation. Suppose the measure $\nu$ supported on $\mathcal{W}$ is absolutely continuous with respect to the uniform distribution and bounded away from zero. If $\gamma_0\in \mathbb{H}_2^{s_0}$, then $c=s_0/s$ \cite{pillaud2018statistical}. Written another way, $(\mathbb{H}_2^{s})^c=\mathbb{H}_2^{s_0}$. In this sense, $c$ precisely quantifies the additional smoothness of $\gamma_0$ relative to $\mathcal{H}$. Moreover, in this Sobolev space, $b=2s/p>1$ \cite{fischer2017sobolev}. The effective dimension is increasing in the input dimension $p$ and decreasing in the degree of smoothness $s$. The minimax optimal rate in Sobolev norm is $n^{-(c-1)/\{2(c+1/b)\}}=n^{-(s_0-s)/(2s_0+p)}$, which is achieved by kernel ridge regression with the rate optimal regularization $\lambda=n^{-1/(c+1/b)}=n^{-2s/(2s_0+p)}$. Our analysis applies to Sobolev spaces over $\mathbb{R}^p$ as a special case; our results are much more general, allowing treatment and covariates to be in Polish spaces. 

\subsection{Finite sample rates}

Towards a guarantee of uniform consistency, we place regularity conditions on the original spaces. In anticipation of counterfactual distributions in Supplement~\ref{sec:distribution}, we also include conditions for the outcome space in parentheses.
\begin{assumption}[Original space regularity conditions]\label{assumption:original}
Assume $\mathcal{D}$, $\mathcal{V}$, $\mathcal{X}$ (and $\mathcal{Y}$) are Polish spaces. Further assume $\mathcal{Y}\subset \mathbb{R}$, $\int y^2 \mathrm{d}\text{\normalfont pr}(y)< \infty$, and a moment condition holds: there exist constants $\sigma,\tau$ such that for all $m\geq 2$, $\int |y-\gamma_0(D,X)|^m\mathrm{d}\text{\normalfont pr}(y \mid D,X)\leq m! \sigma^2 \tau^{m-2}/2$ almost surely. For $\theta_0^{CATE}$, replace $X$ with $(V,X)$.
\end{assumption}
A Polish space is a separable and completely metrizable topological space. Random variables with support in a Polish space may be discrete or continuous and may even be infinite dimensional. Bounded $Y$ implies the moment condition.


Next, we assume the regression $\gamma_0$ is smooth in the sense of~\eqref{eq:prior}, and $\mathcal{H}$ has low effective dimension in the sense of~\eqref{eq:b}. Denote the $j$th eigenvalue of the convolution operator for $\mathcal{H}$ by $\eta_j(\mathcal{H})$. Recall that $\eta_j(\mathcal{H})$ is also the $j$th eigenvalue of the feature covariance operator.
\begin{assumption}[Smoothness and spectral decay for regression]\label{assumption:smooth_gamma}
Assume $\gamma_0\in\mathcal{H}^c$ with $c\in(1,2]$, and $\eta_j(\mathcal{H})\leq C j^{-b}$ with $b\geq 1$.
\end{assumption}
See Supplement~\ref{sec:proof} for alternative ways of writing and interpreting Assumption~\ref{assumption:smooth_gamma}. We place similar smoothness and spectral decay conditions on the conditional mean embeddings $\mu_x(d)$ and $\mu_x(v)$,  which are generalized conditional expectation functions. We articulate this assumption abstractly for the conditional mean embedding $\mu_{a}(b)=\int \phi(a)\mathrm{d}\text{\normalfont pr}(a \mid b)$ where $a\in\mathcal{A}_{\ell}$ and $b\in\mathcal{B}_{\ell}$. All one has to do is specify $\mathcal{A}_{\ell}$ and $\mathcal{B}_{\ell}$ to specialize the assumption. For $\mu_x(d)$, $\mathcal{A}_1=\mathcal{X}$ and $\mathcal{B}_1=\mathcal{D}$; for $\mu_x(v)$, $\mathcal{A}_2=\mathcal{X}$ and $\mathcal{B}_2=\mathcal{V}$. For fixed $\mathcal{A}_{\ell}$ and $\mathcal{B}_{\ell}$, we parametrize smoothness by $c_{\ell}$ and spectral decay by $b_{\ell}$.

Formally, define the conditional expectation operator $E_{\ell}:\mathcal{H}_{\mathcal{A}_{\ell}}\rightarrow\mathcal{H}_{\mathcal{B}_{\ell}}$, $f(\cdot)\mapsto E\{f(A_{\ell}) \mid B_{\ell}=\cdot\}$. By construction, $E_{\ell}$ encodes the same information as $\mu_{a}(b)$ since
$$
\{\mu_{a}(b)\}(\cdot)=\int \phi(a)\mathrm{d}\text{\normalfont pr}(a \mid b) =\{E_{\ell}\phi(\cdot)\}(b) =\{E_{\ell}^* \phi(b)\}(\cdot),\quad a\in\mathcal{A}_{\ell}, \quad b\in\mathcal{B}_{\ell},
$$
where $E_{\ell}^*$ is the adjoint of $E_{\ell}$. We denote the space of Hilbert--Schmidt operators between $\mathcal{H}_{\mathcal{A}_{\ell}}$ and $\mathcal{H}_{\mathcal{B}_{\ell}}$ by $\mathcal{L}_2(\mathcal{H}_{\mathcal{A}_{\ell}},\mathcal{H}_{\mathcal{B}_{\ell}})$. \cite{grunewalder2013smooth,singh2019kernel} prove that $\mathcal{L}_2(\mathcal{H}_{\mathcal{A}_{\ell}},\mathcal{H}_{\mathcal{B}_{\ell}})$ is an RKHS in its own right, for which we can assume smoothness in the sense of~\eqref{eq:prior} and spectral decay in the sense of~\eqref{eq:b}.


\begin{assumption}[Smoothness and spectral decay for mean embedding]\label{assumption:smooth_op}
Assume the following: $E_{\ell}\in \mathcal{L}_2(\mathcal{H}_{\mathcal{A}_{\ell}},\mathcal{H}^{c_{\ell}}_{\mathcal{B}_{\ell}})$ with $c_{\ell}\in(1,2]$, and $\eta_{\ell}(\mathcal{H}_{\mathcal{B}_{\ell}})\leq C j^{-b_{\ell}}$ with $b_{\ell}\geq 1$.
\end{assumption}
Just as we place approximation assumptions for $\gamma_0$ in terms of $\mathcal{H}$, which provides the features onto which we project $Y$, we place approximation assumptions for $E_{\ell}$ in terms of $\mathcal{H}_{\mathcal{B}_{\ell}}$, which provides the features $\phi(B_{\ell})$ onto which we project $\phi(A_{\ell})$. Under these conditions, we arrive at our main theoretical guarantee.
\begin{theorem}[Uniform consistency of causal functions]\label{theorem:consistency_treatment}
Suppose the conditions of Lemma~\ref{theorem:id_treatment} hold, as well as Assumptions~\ref{assumption:RKHS},~\ref{assumption:original}, and~\ref{assumption:smooth_gamma}. Set $(\lambda,\lambda_1,\lambda_2)=\{n^{-1/(c+1/b)},n^{-1/(c_1+1/b_1)},n^{-1/(c_2+1/b_2)}\}$, which is rate optimal regularization.
\begin{enumerate}
    \item Then with high probability
    $
    \|\hat{\theta}^{ATE}-\theta_0^{ATE}\|_{\infty}=O\left[n^{-(c-1)/\{2(c+1/b)\}}\right]$ and $\|\hat{\theta}^{DS}(\cdot,\tilde{\text{\normalfont pr}})-\theta_0^{DS}(\cdot,\tilde{\text{\normalfont pr}})\|_{\infty}=O\left[ n^{-(c-1)/\{2(c+1/b)\}}+\tilde{n}^{-1/2}\right].
    $
    \item If in addition Assumption~\ref{assumption:smooth_op} holds with $\mathcal{A}_1=\mathcal{X}$ and $\mathcal{B}_1=\mathcal{D}$, then with high probability
      $
    \|\hat{\theta}^{ATT}-\theta_0^{ATT}\|_{\infty}=O\left[n^{-(c-1)/\{2(c+1/b)\}}+n^{-(c_1-1)/\{2(c_1+1/b_1)\}}\right].
    $
    \item If in addition Assumption~\ref{assumption:smooth_op} holds with $\mathcal{A}_2=\mathcal{X}$ and $\mathcal{B}_2=\mathcal{V}$, then with high probability
      $
    \|\hat{\theta}^{CATE}-\theta_0^{CATE}\|_{\infty}=O\left[n^{-(c-1)/\{2(c+1/b)\}}+n^{-(c_2-1)/\{2(c_2+1/b_2)\}}\right].
    $
\end{enumerate}
Likewise for incremental functions, e.g. $
    \|\hat{\theta}^{\nabla:ATE}-\theta_0^{\nabla:ATE}\|_{\infty}=O\left[n^{-(c-1)/\{2(c+1/b)\}}\right].
    $
\end{theorem}
Explicit constants hidden by the $O(\cdot)$ notation, as well as explicit specializations of Assumption~\ref{assumption:smooth_op}, are indicated in Appendices~\ref{sec:proof} and~\ref{sec:operator}. These rates approach $n^{-1/4}$ when $(c,c_1,c_2)=2$ and $(b,b_1,b_2)\rightarrow \infty$, i.e. when the regressions are smooth and when the effective dimensions are finite. Interestingly, each rate combines minimax optimal rates in RKHS norm: $n^{-(c-1)/\{2(c+1/b)\}}$ for standard nonparametric regression \cite[Theorem 2]{fischer2017sobolev}; $\tilde{n}^{-1/2}$ for unconditional mean embeddings \cite[Theorem 1]{tolstikhin2017minimax}; and, in contemporaneous work, $n^{-(c_{\ell}-1)/\{2(c_{\ell}+1/b_{\ell})\}}$ for conditional mean embeddings \cite[Theorem 3]{li2022optimal}.

\begin{remark}[Technical innovation]
Our conditional mean embedding rate builds on original analysis of conditional expectation operators in Supplement~\ref{sec:operator} that is of independent interest. We improve the rate from $n^{-(c_{\ell}-1)/\{2(c_{\ell}+1)\}}$ \cite[Theorem 2]{singh2019kernel} to $n^{-(c_{\ell}-1)/\{2(c_{\ell}+1/b_{\ell})\}}$. Our consideration of Hilbert--Schmidt norm departs from \cite{park2020measure} and \cite{talwai2022sobolev}, who study surrogate risk and operator norm, respectively. Our assumptions also depart from \cite[Hypothesis 5]{singh2019kernel}, \cite[Theorem 4.5]{park2020measure}, and \cite[Assumptions 3 and 4]{talwai2022sobolev}. Instead, Assumption~\ref{assumption:smooth_op} directly generalizes \cite[Conditions SRC and EVD]{fischer2017sobolev} from RKHS functions to Hilbert--Schmidt operators.
\end{remark}

Overall, rates slower than $n^{-1/4}$ reflect the challenge of a $\sup$ norm guarantee, which is stronger than a mean square error guarantee and encodes caution about worst case scenarios when informing policy decisions. For comparison, the minimax optimal Sobolev norm rate for learning an $s_0$-smooth regression, using $\mathbb{H}_2^s$ over $\mathbb{R}^p$, is $n^{-(c-1)/\{2(c+1/b)\}}=n^{-(s_0-s)/(2s_0+p)}$.





\newpage


\appendix

\section{Identification}\label{sec:id}

In seminal work, \cite{rosenbaum1983central,robins1986new} state sufficient conditions under which causal functions, philosophical quantities defined in terms of potential outcomes $\{Y^{(d)}\}$, can be measured from empirical quantities such as outcomes $Y$, treatments $D$, and covariates $(V,X)$. Colloquially, this collection of sufficient conditions is known as selection on observables. We assume selection on observables in the main text, and Pearl's front and back door criteria in Supplement~\ref{sec:graphical}.

\begin{assumption}[Selection on observables]\label{assumption:selection}
Assume
\begin{enumerate}
    \item No interference: if $D=d$ then $Y=Y^{(d)}$.
    \item Conditional exchangeability: $\{Y^{(d)}\}\indep D \mid X$.
    \item Overlap: if $f(x)>0$ then $f(d \mid x)>0$, where $f(x)$ and $f(d \mid x)$ are densities. 
\end{enumerate}
For $\theta_0^{CATE}$, replace $X$ with $(V,X)$.
\end{assumption}

No interference is also called the stable unit treatment value assumption. It rules out network effects, also called spillovers. Conditional exchangeability states that conditional on covariates $X$, treatment assignment is as good as random. Overlap ensures that there is no covariate stratum $X=x$ such that treatment has a restricted support.
 To handle $\theta_0^{DS}$, we place a standard assumption in transfer learning.
\begin{assumption}[Distribution shift]\label{assumption:covariate}
Assume
\begin{enumerate}
    \item $\tilde{\text{\normalfont pr}}(Y,D,X)=\text{\normalfont pr}(Y \mid D,X)\tilde{\text{\normalfont pr}}(D,X)$;
    \item $\tilde{\text{\normalfont pr}}(D,X)$ is absolutely continuous with respect to $\text{\normalfont pr}(D,X)$.
\end{enumerate}
\end{assumption}
Populations $\text{\normalfont pr}$ and $\tilde{\text{\normalfont pr}}$ differ only in the distribution of treatments and covariates. Moreover, the support of $\text{\normalfont pr}$ contains the support of $\tilde{\text{\normalfont pr}}$.  An immediate consequence is that the regression $\gamma_0(d,x)=E(Y \mid D=d,X=x)$ remains the same across the different populations $\text{\normalfont pr}$ and $\tilde{\text{\normalfont pr}}$. 

\section{Simulations and program evaluation}\label{sec:experiments}

\subsection{Simulations}

\begin{figure}[ht]
\begin{centering}
     \begin{subfigure}[b]{0.48\textwidth}
         \centering
         \includegraphics[width=\textwidth]{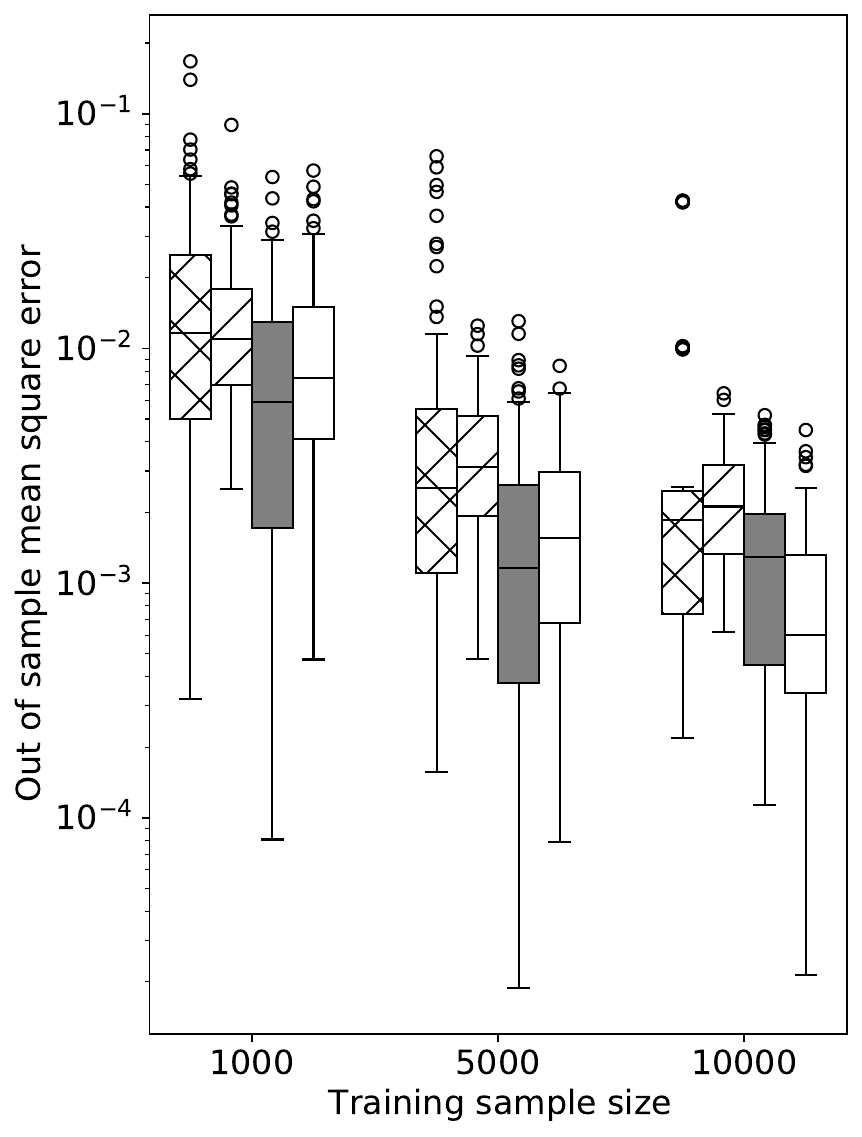}
         \caption{Dose response curve.}
     \end{subfigure}
     \hfill
     \begin{subfigure}[b]{0.48\textwidth}
         \centering
         \includegraphics[width=\textwidth]{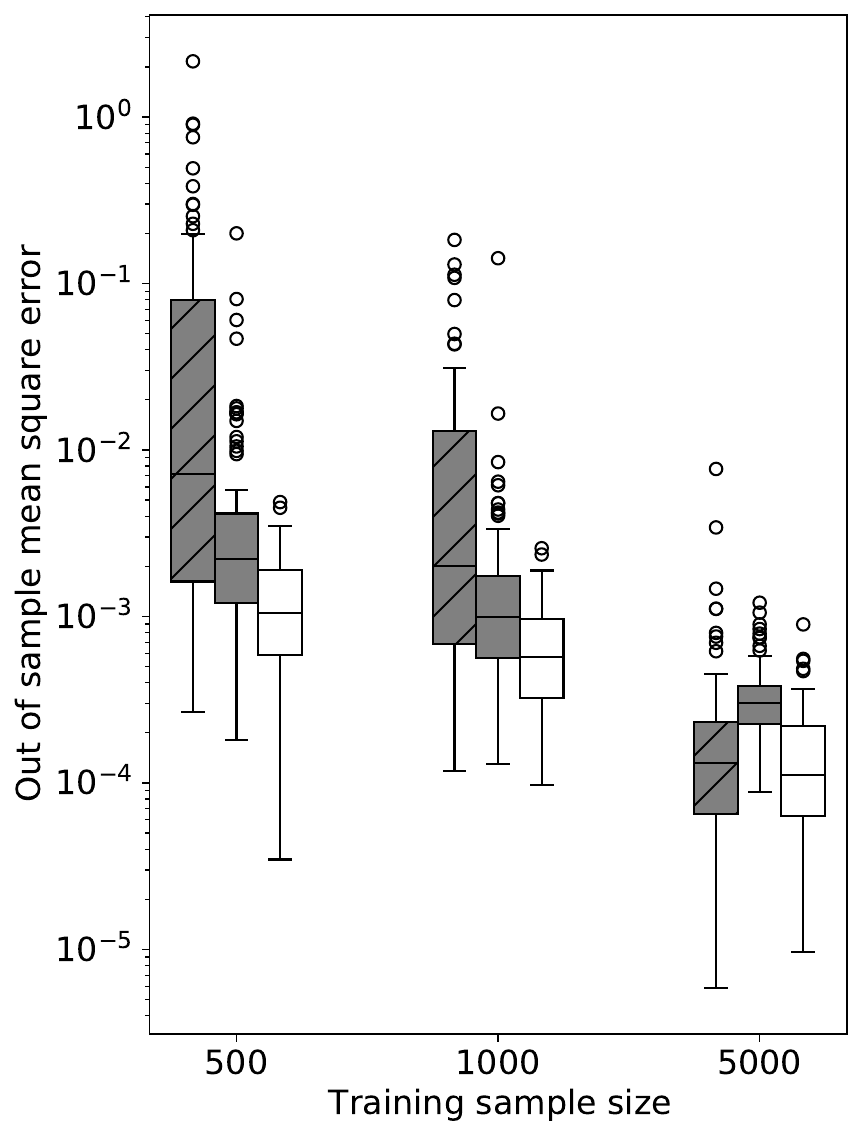}
         \caption{Heterogeneous treatment effect.}
     \end{subfigure}
\par
\caption{\label{fig:cont}
Nonparametric causal function simulations. We implement the estimators of \cite{abrevaya2015estimating} (\texttt{IPW}, lined gray), \cite{kennedy2017nonparametric} (\texttt{DR1}, checkered white), \cite{colangelo2020double} (\texttt{DR2}, lined white), and \cite{semenova2021debiased} (\texttt{DR-series}, gray), in addition to our own (\texttt{RKHS}, white).}
\end{centering}
\end{figure}
We demonstrate that our nonparametric causal function estimators outperform some leading alternatives in nonlinear simulations with many covariates, despite the relative simplicity of our proposed approach. For each causal function design and sample size, we implement 100 simulations and calculate mean square error with respect to the true causal function. Figure~\ref{fig:cont} visualizes results. A lower mean square error is desirable. See Supplement~\ref{sec:simulations} for a full exposition of the data generating processes and implementation details.  

The dose response curve design \cite{colangelo2020double} involves learning the causal function $\theta_0^{ATE}(d)=1.2d+d^2$. A single observation consists of the triple $(Y,D,X)$ for outcome, treatment, and high dimensional covariates where $Y,D\in\mathbb{R}$ and $X\in\mathbb{R}^{100}$. In addition to our one-line nonparametric estimator (\texttt{RKHS}, white), we implement the estimators of \cite{kennedy2017nonparametric} (\texttt{DR1}, checkered white), \cite{colangelo2020double} (\texttt{DR2}, lined white), and \cite{semenova2021debiased} (\texttt{DR-series}, gray). \texttt{DR1} and \texttt{DR2} are local estimators that involve Nadaraya--Watson smoothing around doubly robust estimating equations. \texttt{DR-series} uses series regression with debiased pseudo outcomes, and we give it the advantage of correct specification as a quadratic function. By the Wilcoxon rank sum test, \texttt{RKHS} significantly outperforms alternatives at sample size 10,000, with p value less than $10^{-3}$, despite its relative simplicity.

Though our approach allows for heterogeneous response of a continuous treatment, we implement a design for heterogeneous effect of a binary treatment in order to facilitate comparison with existing methods. The heterogeneous treatment effect design \cite{abrevaya2015estimating} involves learning the causal functions $\theta_0^{CATE}(0,v)=0$ and $\theta_0^{CATE}(1,v)=v(1+2v)^2(v-1)^2$. A single observations consists of the tuple $(Y,D,V,X)$ for outcome, treatment, covariate of interest, and other covariates. In this design, $Y,D,V\in\mathbb{R}$ and  $X\in\mathbb{R}^3$. In addition to our one-line nonparametric estimator (\texttt{RKHS}, white), we implement the estimators of \cite{abrevaya2015estimating} (\texttt{IPW}, lined gray) and \cite{semenova2021debiased} (\texttt{DR-series}, gray). The former involves Nadaraya--Watson smoothing around an inverse propensity estimator, and the latter involves (correctly specified) series regression with a debiased pseudo outcome. The R learner \cite{nie2021quasi} cannot be implemented since $V\neq X$. The simple \texttt{RKHS} approach significantly outperforms alternatives at sample sizes 500 and 1,000 by the Wilcoxon rank sum test, with p values less than $10^{-5}$. 
\subsection{Program evaluation: US Job Corps}

\begin{figure}[ht]
\begin{centering}
     \begin{subfigure}[b]{0.45\textwidth}
         \centering
         \includegraphics[width=\textwidth]{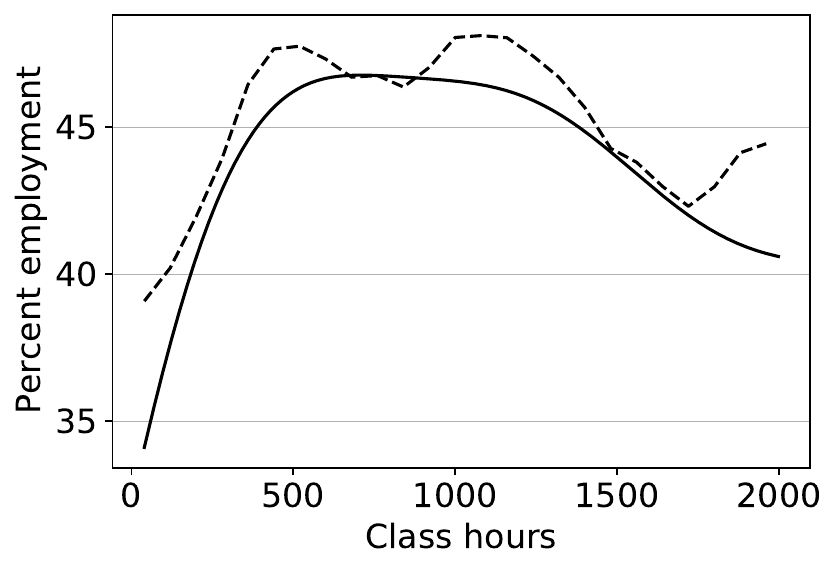}
         \caption{Dose response curve.}
     \end{subfigure}
     \hfill
     \begin{subfigure}[b]{0.45\textwidth}
         \centering
         \includegraphics[width=\textwidth]{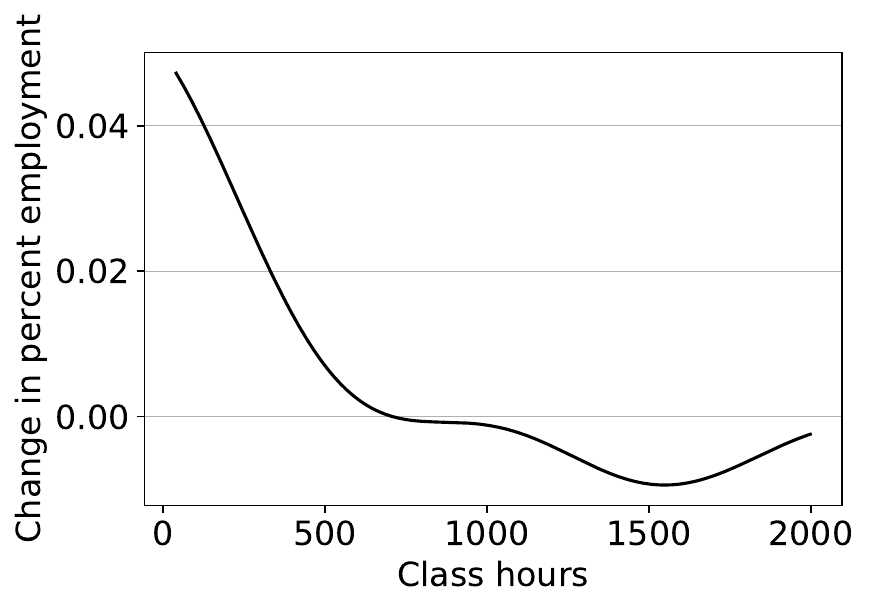}
         \caption{Incremental response curve.}
     \end{subfigure}
     \begin{subfigure}[b]{0.45\textwidth}
         \centering
         \includegraphics[width=\textwidth]{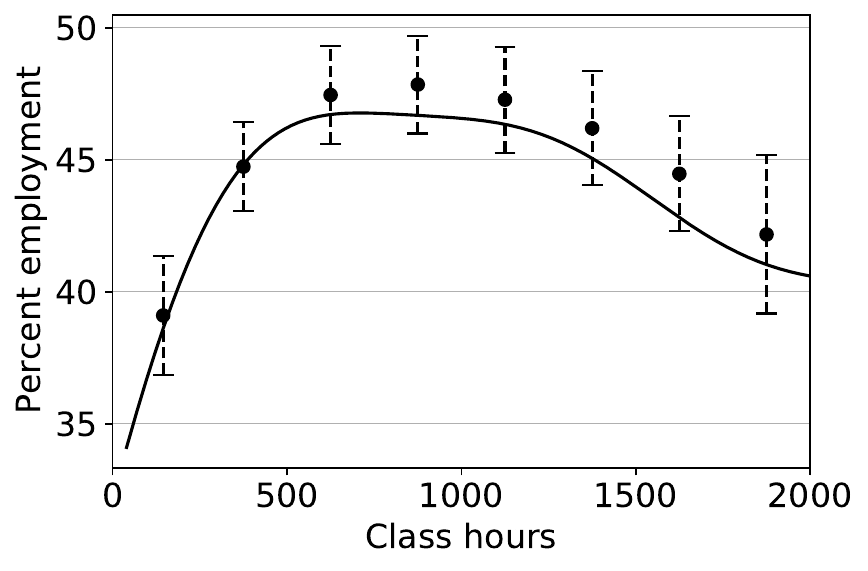}
         \caption{Discrete treatment effects.}
     \end{subfigure}
     \hfill
     \begin{subfigure}[b]{0.45\textwidth}
         \centering
         \includegraphics[width=\textwidth]{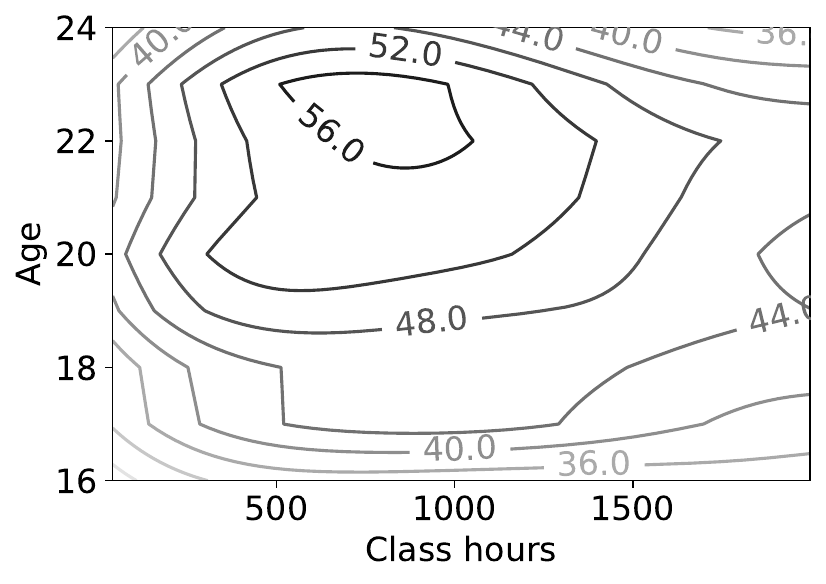}
         \caption{Heterogeneous response curve.}
     \end{subfigure}
\par
\caption{\label{fig:JC}
Effect of job training on employment. We implement our estimators for dose, heterogeneous, and incremental response curves (\texttt{RKHS}, solid). For comparison, we also implement the dose response curve estimator of \cite{colangelo2020double} (\texttt{DR2}, dashes) as well as the discrete treatment effects of \cite{singh2021debiased} (\texttt{DR3}, vertical bars).}
\end{centering}
\end{figure}

To demonstrate how kernel methods for causal functions are a practical addition to the empirical economic toolkit, we conduct a real world program evaluation. Specifically, we estimate dose, heterogeneous, and incremental response curves of the Jobs Corps, the largest job training program for disadvantaged youth in the US. The Job Corps is financed by the US Department of Labor, and it serves about 50,000 participants annually. Participation is free for individuals who meet low income requirements. Access to the program was randomized from November 1994 to February 1996; see \cite{schochet2008does} for details. Many studies focus on data from this period to evaluate the effect of job training on employment \cite{flores2012estimating,colangelo2020double}. Though access to the program was randomized, individuals could decide whether to participate and for how many hours. From a causal perspective, we assume selection on observables: conditional on observed covariates, participation was exogenous on the extensive and intensive margins. From a statistical perspective, we assume that different intensities of job training have smooth effects on counterfactual employment, and that those effects are smoothly modified by age, assumptions motivated by labor market theory.

In this setting, the continuous treatment $D\in\mathbb{R}$ is total hours spent in academic or vocational classes in the first year after randomization, and the continuous outcome $Y\in\mathbb{R}$ is the proportion of weeks employed in the second year after randomization. The covariates $X\in\mathbb{R}^{40}$ include age, gender, ethnicity, language competency, education, marital
status, household size, household income, previous receipt of social aid, family background, health, and health related behavior at base line. As in \cite{colangelo2020double}, we focus on the $n=3,906$ observations for which $D\geq 40$, i.e. individuals who completed at least one week of training. We implement various causal parameters in Figure~\ref{fig:JC}: the dose response curve; the incremental response curve; the discrete treatment effects with confidence intervals of \cite{singh2021debiased}; and the heterogeneous response curve with respect to age. For the discrete effects, we discretize treatment into roughly equiprobable bins: $[40,250]$, $(250,500]$, $(500,750]$ $(750,1000]$, $(1000,1250]$, $(1250,1500]$, $(1500,1750]$, and $(1750,2000]$ class hours. As far as we know, the heterogeneous response of class hours, a continuous treatment, has not been previously studied in this empirical setting. In Supplement~\ref{sec:application}, we provide implementation details and verify that our results are robust to the choice of sample.

The dose response curve plateaus and achieves its maximum around $d=500$, corresponding to 12.5 weeks of classes. Our global estimate (\texttt{RKHS}, solid) has the same overall shape but is smoother and slightly lower than the collection of local estimates from \cite{colangelo2020double} (\texttt{DR2}, dashes). The smoothness of our estimator is a consequence of the RKHS assumptions, and we see how it is a virtue for empirical economic research; a smooth dose response curve is more economically plausible in this setting. The first 12.5 weeks of classes confer most of the gain in employment: from 35\% employment to more than 47\% employment for the average participant. The incremental response curve (\texttt{RKHS}, solid) is the derivative of the dose response curve, and it visualizes where the greatest gain happens. The discrete treatment effects of \cite{singh2021debiased} (\texttt{DR3}, vertical bars) corroborate our dose response curve, and the 95\% confidence intervals contain the dose response curve of \cite{colangelo2020double} (\texttt{DR2}, dashes) as well as our own (\texttt{RKHS}, solid). 
Finally, the heterogeneous response curve (\texttt{RKHS}, solid) shows that age plays a substantial role in the effectiveness of the intervention. For the youngest participants, the intervention has a small effect: employment only increases from 28\% to at most 36\%. For older participants, the intervention has a large effect: employment increases from 40\% to 56\%. 
Our policy recommendation is therefore 12--14 weeks of classes targeting individuals 21--23 years old.

\section{Counterfactual distributions}\label{sec:distribution}

\subsection{Definition}

In the main text, we study causal functions defined as means of potential outcomes. In this section, we extend the estimators and analyses presented in the main text to counterfactual distributions of potential outcomes. A counterfactual distribution can be encoded by a kernel mean embedding using a new feature map $\phi(y)$ for a new scalar valued RKHS $\mathcal{H}_{\mathcal{Y}}$. We now allow $\mathcal{Y}$ to be a Polish space (Assumption~\ref{assumption:original}).
\begin{definition}[Counterfactual distributions and embeddings] We define 
\begin{enumerate}
    \item Counterfactual distribution: $
\theta_0^{D:ATE}(d)=\text{\normalfont pr}\{Y^{(d)}\}
$ is the counterfactual distribution of outcomes given intervention $D=d$ for the entire population.
    \item Counterfactual distribution with distribution shift: $
\theta_0^{D:DS}(d,\tilde{\text{\normalfont pr}})=\tilde{\text{\normalfont pr}}\{Y^{(d)}\}
$ is the counterfactual distribution of outcomes given intervention $D=d$ for an alternative population with data distribution $\tilde{\text{\normalfont pr}}$ (elaborated in Assumption~\ref{assumption:covariate}).
    \item Conditional counterfactual distribution: $
\theta_0^{D:ATT}(d,d')=\text{\normalfont pr}\{Y^{(d')} \mid D=d\}
$ is the counterfactual distribution of outcomes given intervention $D=d'$ for the subpopulation who actually received treatment $D=d$.
    \item Heterogeneous counterfactual distribution: $
\theta_0^{D:CATE}(d,v)=\text{\normalfont pr}\{Y^{(d)} \mid V=v\}
$ is the counterfactual distribution of outcomes given intervention $D=d$ for the subpopulation with covariate value $V=v$.
\end{enumerate}
Likewise we define counterfactual distribution embeddings, e.g. $
\check{\theta}_0^{D:ATE}(d)=E\{\phi(Y^{(d)})\}.
$
\end{definition}

Our strategy is to estimate the embedding of a  counterfactual distribution. At that point, the analyst may use the embedding to (i) estimate moments of the counterfactual distribution \cite{kanagawa2014recovering} or (ii) sample from the counterfactual distribution \cite{welling2009herding}. Since we already analyze means in the main text, we focus on (ii) in this supplement.

\subsection{Identification}

The same identification results apply to counterfactual distributions.

\begin{lemma}[Identification of counterfactual distributions]\label{theorem:id_treatment_dist}
If Assumption~\ref{assumption:selection} holds,
\begin{enumerate}
     \item $\{\theta_0^{D:ATE}(d)\}(y)=\int \text{\normalfont pr}(y \mid d,x)\mathrm{d}\text{\normalfont pr}(x)$.
    \item If in addition Assumption~\ref{assumption:covariate} holds, then $\{\theta_0^{D:DS}(d,\tilde{\text{\normalfont pr}})\}(y)=\int \text{\normalfont pr}(y \mid d,x)\mathrm{d}\tilde{\text{\normalfont pr}}(x)$.
    \item $\{\theta_0^{D:ATT}(d,d')\}(y)=\int \text{\normalfont pr}(y \mid d',x)\mathrm{d}\text{\normalfont pr}(x \mid d)$ \cite{chernozhukov2013inference}.
    \item $\{\theta_0^{D:CATE}(d,v)\}(y)=\int \text{\normalfont pr}(y \mid d,v,x)\mathrm{d}\text{\normalfont pr}(x \mid v)$.
\end{enumerate}
Likewise for embeddings of counterfactual distributions. For example, if in addition Assumption~\ref{assumption:RKHS} holds, then $\check{\theta}_0^{D:ATE}(d)=\int E\{\phi(Y) \mid D=d,X=x\}\mathrm{d}\text{\normalfont pr}(x)
$.
\end{lemma}
The identification results for embeddings of counterfactual distributions resemble those presented in the main text. Define the generalized regressions
$
\gamma_0(d,x)=E\{\phi(Y) \mid D=d,X=x\}
$ and $\gamma_0(d,v,x)=E\{\phi(Y) \mid D=d,V=v,X=x\}$. Then we can express these results in the familiar form, e.g. $\check{\theta}_0^{D:ATE}(d)=\int \gamma_0(d,x)\mathrm{d}\text{\normalfont pr}(x)$.

\subsection{Closed form solution}

To estimate counterfactual distributions, we extend the RKHS construction in Section~\ref{sec:algorithm}. As before, define scalar valued RKHSs for treatment $D$ and covariates $X$. Define an additional scalar valued RKHS for outcome $Y$. Because the regression $\gamma_0$ is now a conditional mean embedding, we present a construction involving a conditional expectation operator. Define the conditional expectation operator
$
E_3:\mathcal{H}_{\mathcal{Y}}\rightarrow \mathcal{H}_{\mathcal{D}}\otimes \mathcal{H}_{\mathcal{X}},\; f(\cdot)\mapsto E\{f(Y) \mid D=\cdot,X=\cdot \}
$. By construction
$
\gamma_0(d,x)=E_3^*\{\phi(d)\otimes \phi(x)\}
$.  As before, we replace $X$ with $(V,X)$ for $\theta_0^{D:CATE}$. We place regularity conditions on this RKHS construction, similar to those in Section~\ref{sec:algorithm}, to represent counterfactual distributions as evaluations of $E_3^*$. This representation allows for continuous treatment, unlike the representation in \cite[eq. 16, 17, 20]{muandet2021counterfactual}.

\begin{theorem}[Decoupling via kernel mean embeddings]\label{theorem:representation_dist}
Suppose the conditions of Lemma~\ref{theorem:id_treatment_dist} hold. Further suppose Assumption~\ref{assumption:RKHS} holds and $E_3\in\mathcal{L}_2(\mathcal{H}_{\mathcal{Y}},\mathcal{H}_{\mathcal{D}}\otimes \mathcal{H}_{\mathcal{X}})$. Then
\begin{enumerate}
    \item $\check{\theta}_0^{D:ATE}(d)=E_3^*\{\phi(d)\otimes \mu_x\} $ where $\mu_x=\int\phi(x) \mathrm{d}\text{\normalfont pr}(x) $.
    \item $\check{\theta}_0^{D:DS}(d,\tilde{\text{\normalfont pr}})=E_3^*\{\phi(d)\otimes \nu_x\}$ where $\nu_x=\int\phi(x) \mathrm{d}\tilde{\text{\normalfont pr}}(x) $.
    \item $\check{\theta}_0^{D:ATT}(d,d')=E_3^*\{\phi(d')\otimes \mu_x(d)\} $  where $\mu_x(d)=\int\phi(x) \mathrm{d}\text{\normalfont pr}(x \mid d)$.
    \item $\check{\theta}_0^{D:CATE}(d,v)=E_3^*\{\phi(d)\otimes \phi(v)\otimes \mu_{x}(v)\}$ where $\mu_{x}(v)= \int \phi(x) \mathrm{d}\text{\normalfont pr}(x \mid v)$.
\end{enumerate}
For $\theta_0^{D:CATE}$, we instead assume $E_3\in\mathcal{L}_2(\mathcal{H}_{\mathcal{Y}},\mathcal{H}_{\mathcal{D}}\otimes\mathcal{H}_{\mathcal{V}} \otimes  \mathcal{H}_{\mathcal{X}})$.
\end{theorem}

See Supplement~\ref{sec:derivation} for the proof. The mean embeddings are the same as in Theorem~\ref{theorem:representation_treatment}. They encode the reweighting distributions as elements in the RKHS such that the counterfactual distribution embeddings can be expressed as evaluations of $E_3^*$. 

As in Section~\ref{sec:algorithm}, the abstract representation helps to define estimators with closed form solutions that can be easily computed. In particular, the representation separates the three steps necessary to estimate a counterfactual distribution: estimating a conditional distribution, which may involve many covariates; estimating the distribution for reweighting; and using one distribution to integrate another. For example, for $\check{\theta}_0^{D:CATE}(d,v)$, our estimator is $\hat{\theta}^{D:CATE}(d,v)=\hat{E}_3^*\{\phi(d)\otimes \phi(v) \otimes \hat{\mu}_x(v)\}$. $\hat{E}_3$ and $\hat{\mu}_x(v)$ are generalized kernel ridge regressions, and the latter can be used to integrate the former by simply multiplying the two. This algorithmic insight is a key innovation of the present work, and the reason why our estimators have simple closed form solutions despite complicated causal integration.

\begin{algorithm}[Estimation of counterfactual distribution embeddings]\label{algorithm:dist}
Denote the empirical kernel matrices
$
K_{DD}, K_{XX}, K_{YY}\in\mathbb{R}^{n\times n}
$. Let $(\tilde{X}_i)$ $(i=1,...,\tilde{n})$ be observations drawn from population $\tilde{\text{\normalfont pr}}$. Denote by $\odot$ the elementwise product. The distribution embedding estimators have the closed form solutions
\begin{enumerate}
    \item $\{\hat{\theta}_0^{D:ATE}(d)\}(y)=n^{-1}\sum_{i=1}^n K_{yY}(K_{DD}\odot K_{XX}+n\lambda_3  I )^{-1}(K_{Dd}\odot K_{Xx_i})$;
    \item $\{\hat{\theta}_0^{D:DS}(d)\}(y)=\tilde{n}^{-1}\sum_{i=1}^{\tilde{n}} K_{yY}(K_{DD}\odot K_{XX}+n\lambda_3  I )^{-1}(K_{Dd}\odot K_{X\tilde{x}_i})$;
    \item $\{\hat{\theta}_0^{D:ATT}(d,d')\}(y)=K_{yY}(K_{DD}\odot K_{XX}+n\lambda_3  I )^{-1}[K_{Dd'}\odot \{K_{XX}(K_{DD}+n\lambda_1  I )^{-1}K_{Dd}\}]$;
     \item $\{\hat{\theta}_0^{D:CATE}(d,v)\}(y)=K_{yY}(K_{DD}\odot K_{VV}\odot K_{XX}+n\lambda_3  I )^{-1}[K_{Dd}\odot K_{Vv} \odot \{K_{XX}(K_{VV}+n\lambda_2  I )^{-1}K_{Vv}\}]$;
\end{enumerate}
where $(\lambda_1,\lambda_2,\lambda_3)$ are ridge regression penalty hyperparameters.
\end{algorithm}
We derive these estimators in Supplement~\ref{sec:derivation}. We give theoretical values for $(\lambda_1,\lambda_2,\lambda_3)$ that optimally balance bias and variance in Theorem~\ref{theorem:consistency_treatment} below. Supplement~\ref{sec:tuning} gives practical tuning procedures, one of which is asymptotically optimal. 
We avoid the estimation and inversion of propensity scores in \cite[eq. 21]{muandet2021counterfactual}.

Algorithm~\ref{algorithm:dist} estimates counterfactual distribution embeddings. The ultimate parameters of interest are counterfactual distributions. We present a deterministic procedure that uses the distribution embedding to provide samples $(\tilde{Y}_j)$ from the distribution. In Theorem~\ref{theorem:conv_dist} below, we prove that these samples converge in distribution to the counterfactual distribution. The procedure is a variant of kernel herding  \cite{welling2009herding,muandet2021counterfactual}.

\begin{algorithm}[Estimation of counterfactual distributions]\label{algorithm:herding}
Recall that $\hat{\theta}_0^{D:ATE}(d)$ is a mapping from $\mathcal{Y}$ to $\mathbb{R}$. 
Given $\hat{\theta}_0^{D:ATE}(d)$, calculate
\begin{enumerate}
    \item $\tilde{Y}_1=\argmax_{y\in\mathcal{Y}} \left[ \{\hat{\theta}_0^{D:ATE}(d)\}(y)\right]$;
    \item $\tilde{Y}_{j}=\argmax_{y\in\mathcal{Y}} \left[ \{\hat{\theta}_0^{D:ATE}(d)\}(y)-(j+1)^{-1}\sum_{\ell=1}^{j-1}k_{\mathcal{Y}}(\tilde{Y}_{\ell},y)\right]$ for $j>1$.
\end{enumerate}
Likewise for the other counterfactual distributions, replacing $\hat{\theta}_0^{D:ATE}(d)$ with the other quantities in Algorithm~\ref{algorithm:dist}.
\end{algorithm}
By this procedure, samples from counterfactual distributions are straightforward to compute. With such samples, one may visualize a histogram as an estimator of the counterfactual density of potential outcomes. Alternatively, one may test statistical hypotheses.

\subsection{Convergence in distribution}

Towards a guarantee of uniform consistency, we place regularity conditions on the original spaces as in Assumption~\ref{assumption:original}. Importantly, we relax the condition that $\mathcal{Y}\subset \mathbb{R}$; instead, we assume $\mathcal{Y}$ is a Polish space. Next, we assume the regression $\gamma_0$ is smooth and quantify the spectral decay of its RKHS, parameterized in terms of the conditional expectation operator $E_3$. Likewise we assume the conditional mean embeddings $\mu_x(d)$ and $\mu_x(v)$ are smooth and quantify their spectral decay. With these assumptions, we arrive at our next main result.
\begin{theorem}[Uniform consistency of counterfactual distribution embeddings]\label{theorem:consistency_dist}
Suppose Assumptions~\ref{assumption:selection},~\ref{assumption:RKHS},~\ref{assumption:original}, and~\ref{assumption:smooth_op} hold with $\mathcal{A}_3=\mathcal{Y}$ and $\mathcal{B}_3=\mathcal{D}\times \mathcal{X}$ (or $\mathcal{B}_3=\mathcal{D}\times \mathcal{V}\times \mathcal{X}$ for $\theta_0^{D:CATE}$). Set $(\lambda_1,\lambda_2,\lambda_3)=\{n^{-1/(c_1+1/b_1)},n^{-1/(c_2+1/b_2)},n^{-1/(c_3+1/b_3)}\}$, which is rate optimal regularization.
\begin{enumerate}
    \item Then with high probability
    $$
    \sup_{d\in\mathcal{D}}\|\hat{\theta}^{D:ATE}(d)-\check{\theta}_0^{D:ATE}(d)\|_{\mathcal{H}_{\mathcal{Y}}}=O\left[n^{-(c_3-1)/\{2(c_3+1/b_3)\}}\right].
$$
 \item If in addition Assumption~\ref{assumption:covariate} holds, then with high probability
      $$
     \sup_{d\in\mathcal{D}}\|\hat{\theta}^{D:DS}(d,\tilde{\text{\normalfont pr}})-\check{\theta}_0^{D:DS}(d,\tilde{\text{\normalfont pr}})\|_{\mathcal{H}_{\mathcal{Y}}}=O\left[ n^{-(c_3-1)/\{2(c_3+1/b_3)\}}+\tilde{n}^{-1/2}\right].
    $$
    \item If in addition Assumption~\ref{assumption:smooth_op} holds with $\mathcal{A}_1=\mathcal{X}$ and $\mathcal{B}_1=\mathcal{D}$, then with high probability
        $$
    \sup_{d,d'\in\mathcal{D}}\|\hat{\theta}^{D:ATT}(d,d')-\check{\theta}_0^{D:ATT}(d,d')\|_{\mathcal{H}_{\mathcal{Y}}}=O\left[n^{-(c_3-1)/\{2(c_3+1/b_3)\}}+n^{-(c_1-1)/\{2(c_1+1/b_1)\}}\right].
$$
 \item If in addition Assumption~\ref{assumption:smooth_op} holds with $\mathcal{A}_2=\mathcal{X}$ and $\mathcal{B}_2=\mathcal{V}$, then with high probability
      $$
     \sup_{d\in\mathcal{D},v\in\mathcal{V}}\|\hat{\theta}^{D:CATE}(d,v)-\check{\theta}_0^{D:CATE}(d,v)\|_{\mathcal{H}_{\mathcal{Y}}}=O\left[n^{-(c_3-1)/\{2(c_3+1/b_3)\}}+n^{-(c_2-1)/\{2(c_2+1/b_2)\}}\right].
    $$
\end{enumerate}
\end{theorem}
Explicit constants hidden by the $O(\cdot)$ notation are indicated in Supplement~\ref{sec:proof}, as well as explicit specializations of Assumption~\ref{assumption:smooth_op}. Again, these rates approach $n^{-1/4}$ when $(c_1,c_2,c_3)=2$ and $(b_1,b_2,b_3)\rightarrow \infty$, i.e. when the regressions are smooth and when the effective dimensions are finite. Our assumptions do not include an assumption on the smoothness of an explicit density ratio, which appears in \cite[Theorem 11]{fukumizu2013kernel} and \cite[Assumption 3]{muandet2021counterfactual}. Finally, we state an additional regularity condition under which we can prove that the samples $(\tilde{Y}_j)$ calculated from the distribution embeddings weakly converge to the desired distribution.
\begin{assumption}[Additional regularity]\label{assumption:regularity}
Assume
\begin{enumerate}
    \item $\mathcal{Y}$ is locally compact.
    \item $\mathcal{H}_{\mathcal{Y}}\subset\mathcal{C}_0$, where $\mathcal{C}_0$ is the space of bounded, continuous, real valued functions that vanish at infinity.
\end{enumerate}
\end{assumption}
As discussed by \cite{simon2020metrizing}, the combined assumptions that $\mathcal{Y}$ is Polish and locally compact impose weak restrictions. In particular, if $\mathcal{Y}$ is a Banach space, then to satisfy both conditions it must be finite dimensional. Trivially, $\mathcal{Y}=\mathbb{R}^{dim(Y)}$ satisfies both conditions. We arrive at our final result of this section.
\begin{theorem}[Convergence in distribution of counterfactual distributions]\label{theorem:conv_dist}
Suppose the conditions of Theorem~\ref{theorem:consistency_dist} hold, as well as Assumption~\ref{assumption:regularity}. Suppose samples $(\tilde{Y}_j)$ are calculated for $\theta_0^{D:ATE}(d)$ as described in Algorithm~\ref{algorithm:herding}. Then $(\tilde{Y}_j)\rightsquigarrow \theta_0^{D:ATE}(d)$. Likewise for the other counterfactual distributions, replacing $\hat{\theta}_0^{D:ATE}(d)$ with the other quantities in Algorithm~\ref{algorithm:dist}.
\end{theorem}
See Supplement~\ref{sec:proof} for the proof. Samples are drawn for given value $d$. Though our nonparametric consistency result is uniform across treatment values, this convergence in distribution result is for a fixed treatment value.
\section{Graphical models}\label{sec:graphical}

In the main text, we study causal functions defined in the potential outcomes framework and identified by selection on observables. In this supplement, we study causal functions and counterfactual distributions defined in the directed acyclic graph (DAG) framework and identified by Pearl's front and back door criteria. We derive estimators, then prove uniform consistency and convergence in distribution.

\subsection{DAG background}

DAGs provide another popular language for causal inference \cite{pearl2009causality}. Rather than reasoning about $\text{\normalfont pr}\{Y^{(d)}\}$, one reasons about $\text{\normalfont pr}\{Y \mid do(D=d)\}$, where both expressions are concerned with the distribution of outcome $Y$ given intervention $D=d$. For a specific setting, graphical criteria in terms of the DAG can help verify conditional independence statements in terms of potential outcomes. In this section, we provide results in terms of causal DAGs, analogous to the results in terms of potential outcomes given in the main text. In particular, we focus on the front and back door criteria, which are the fundamental building blocks of DAG-based causal inference. 

Assume the analyst has access to a causal DAG $G$ with vertex set $W$, partitioned into four disjoint sets $W=(Y,D,X,U)$. $Y$ is the outcome, $D$ is the set of treatments, $X$ is the set of covariates, and $U$ is the set of unobserved variables. Since counterfactual inquiries involve intervention on the graph $G$, we require notation for graph modification. Denote by $G_{\bar{D}}$ the graph obtained by deleting from $G$ all arrows pointing into nodes in $D$. Denote by $G_{\underline{D}}$ the graph obtained by deleting from $G$ all arrows emerging from nodes in $D$. We denote $d$-separation by $\indep_d$. $d$-separation implies statistical independence. Throughout this section, we make the standard faithfulness assumption: $d$-connection implies statistical dependence. 

\subsection{Identification}

We define causal functions and counterfactual distributions in terms of the $do$ operator on the DAG. For clarity of exposition, we focus on the case where $(D,Y)$ are nodes rather than sets.
\begin{definition}[Causal function and counterfactual distribution: DAG]
$\theta_0^{do}(d)=E\{Y \mid do(D=d)\}$ is the counterfactual mean outcome given intervention $D=d$ for the entire population. Likewise we define the counterfactual distribution $\theta_0^{D:do}(d)=\text{\normalfont pr}\{Y \mid do(D=d)\}$ and counterfactual distribution embedding $
\check{\theta}_0^{D:do}(d)=E\{\phi(Y) \mid do(D=d)\}
$ as in Supplement~\ref{sec:distribution}.
\end{definition}
In seminal works, \cite{pearl1993comment,pearl1995causal} states sufficient conditions under which such effects, philosophical quantities defined in terms of interventions on the graph, can be measured from empirical quantities such as outcomes $Y$, treatments $D$, and covariates $X$. We present two sets of sufficient conditions, known as the back door and front door criteria.

\begin{assumption}[Back door criterion]\label{assumption:back-door}
Assume
\begin{enumerate}
    \item No node in $X$ is a descendent of $D$.
    \item $X$ blocks every path between $D$ and $Y$ that contains an arrow into $D$:
$
(Y\indep_d D \mid X)_{G_{\underline{D}}}.
$
\end{enumerate}
\end{assumption}
Intuitively, the analyst requires sufficiently many and sufficiently well placed covariates $X$ in the context of the graph $G$. Assumption~\ref{assumption:back-door} is satisfied if there is no unobserved confounder $U$, or if any unobserved confounder $U$ with a back door path into treatment $D$ is blocked by $X$. 

\begin{assumption}[Front door criterion]\label{assumption:front-door}
Assume
\begin{enumerate}
    \item $X$ intercepts all directed paths from $D$ to $Y$.
    \item There is no unblocked back door path from $D$ to $X$.
    \item All back door paths from $X$ to $Y$ are blocked by $D$.
    \item $\text{\normalfont pr}(D,X)>0$ almost surely.
\end{enumerate}
\end{assumption}
Intuitively, these conditions ensure that $X$ serves to block all spurious paths from $D$ to $Y$; to leave all directed paths unperturbed; and to create no new spurious paths. As before, define the regression
$
\gamma_0(d,x)=E(Y \mid D=d,X=x)
$.
\begin{lemma}[Identification of causal function: DAG \cite{pearl1993comment,pearl1995causal}]\label{theorem:id_treatment_dag}
Depending on which criterion holds, the causal parameter $\theta_0^{do}(d)$ has different expressions.
\begin{enumerate}
    \item If Assumption~\ref{assumption:back-door} holds then
$
\theta_0^{do}(d)=\int \gamma_0(d,x)\mathrm{d}\text{\normalfont pr}(x).
$
    \item If Assumption~\ref{assumption:front-door} holds then
$
\theta_0^{do}(d)=\int \gamma_0(d',x)\mathrm{d}\text{\normalfont pr}(d')\mathrm{d}\text{\normalfont pr}(x \mid d).
$
\end{enumerate}
If in addition Assumption~\ref{assumption:RKHS} holds then the analogous result holds for counterfactual distribution embeddings using $\gamma_0(d,x)=E\{\phi(Y) \mid D=d,X=x\}$ instead, as in Supplement~\ref{sec:distribution}.
\end{lemma}

Comparing Lemma~\ref{theorem:id_treatment_dag} with Lemma~\ref{theorem:id_treatment}, we see that if Assumption~\ref{assumption:back-door} holds then
our dose response estimator $\hat{\theta}^{ATE}(d)$ in Section~\ref{sec:algorithm} is also a uniformly consistent estimator of $\theta_0^{do}(d)$. Similarly our counterfactual distribution estimator $\hat{\theta}^{D:ATE}(d)$ converges in distribution to $\hat{\theta}^{D:do}(d)$. In the remainder of this section, we therefore focus on what happens if Assumption~\ref{assumption:front-door} holds instead. We study the causal function and counterfactual distribution.

\subsection{Closed form solutions}

We maintain notation from Section~\ref{sec:algorithm}.

\begin{theorem}[Decoupling via kernel mean embedding: DAG]\label{theorem:representation_treatment_dag}
Suppose Assumptions~\ref{assumption:RKHS} and~\ref{assumption:front-door} hold.
\begin{enumerate}
    \item If an addition $\gamma_0\in\mathcal{H}$ then
$
\theta_0^{do}(d)=\langle \gamma_0,\mu_d\otimes \mu_{x}(d)\rangle_{\mathcal{H}};
$
    \item If in addition $E_3\in \mathcal{L}_2(\mathcal{H}_{\mathcal{Y}},\mathcal{H}_{\mathcal{D}}\otimes \mathcal{H}_{\mathcal{X}})$ then 
    $
\check{\theta}_0^{do}(d)=E^*_3\{\mu_d\otimes \mu_{x}(d)\};
    $
\end{enumerate}
where $\mu_d=\int \phi(d)\mathrm{d}\text{\normalfont pr}(d)$ and $\mu_{x}(d)=\int \phi(x)\mathrm{d}\text{\normalfont pr}(x \mid d)$. 
\end{theorem}
See Supplement~\ref{sec:derivation} for the proof. The quantity $\mu_d=\int \phi(d)\mathrm{d}\text{\normalfont pr}(d)$ is the mean embedding of $\text{\normalfont pr}(d)$. The quantity $\mu_{x}(d)=\int \phi(x)\mathrm{d}\text{\normalfont pr}(x \mid d)$ is the conditional mean embedding of $\text{\normalfont pr}(x \mid d)$. This representation helps to derive an estimator with a closed form solution. For $\theta_0^{do}(d)$, our estimator will be $\hat{\theta}^{FD}(d)=\langle \hat{\gamma}, \hat{\mu}_d\otimes \hat{\mu}_x(d)\rangle_{\mathcal{H}}$, where $\hat{\gamma}$ is a standard kernel ridge regression, $\hat{\mu}_d$ is an empirical mean, and $\hat{\mu}_x(d)$ is an appropriately defined kernel ridge regression.

\begin{algorithm}[Estimation of causal functions: DAG]\label{algorithm:dag}
Denote the empirical kernel matrices
$
K_{DD}, K_{XX}, K_{YY}\in\mathbb{R}^{n\times n}
$
calculated from observations drawn from population $\text{\normalfont pr}$. Denote by $\odot$ the elementwise product. The front door criterion estimators have the closed form solutions
\begin{enumerate}
    \item $
\hat{\theta}^{FD}(d)=n^{-1}\sum_{i=1}^n Y^{\top}(K_{DD}\odot K_{XX}+n\lambda  I )^{-1}[K_{Dd_i}\odot \{K_{XX}(K_{DD}+n\lambda_1  I )^{-1}K_{Dd}\}]
$
    \item $
\{\hat{\theta}^{D:FD}(d)\}(y)=n^{-1}\sum_{i=1}^n K_{yY}(K_{DD}\odot K_{XX}+n\lambda_3  I )^{-1}[K_{Dd_i}\odot \{K_{XX}(K_{DD}+n\lambda_1  I )^{-1}K_{Dd}\}]    
$
\end{enumerate}
where $(\lambda,\lambda_1,\lambda_3)$ are ridge regression penalty hyperparameters.
\end{algorithm}
We derive this estimator in Supplement~\ref{sec:derivation}. We give theoretical values for $(\lambda,\lambda_1,\lambda_3)$ that optimally balance bias and variance in Theorem~\ref{theorem:consistency_dag} below. Supplement~\ref{sec:tuning} gives practical tuning procedures, one of which is asymptotically optimal.

\subsection{Uniform consistency and convergence in distribution}

Towards a guarantee of uniform consistency, we place the same assumptions as in Section~\ref{sec:algorithm}.

\begin{theorem}[Uniform consistency of causal functions: DAG]\label{theorem:consistency_dag}
Suppose the conditions of Theorem~\ref{theorem:representation_treatment_dag} hold, as well as Assumptions~\ref{assumption:original} and~\ref{assumption:smooth_op} with $\mathcal{A}_1=\mathcal{X}$ and $\mathcal{B}_1=\mathcal{D}$. Set $(\lambda,\lambda_1,\lambda_3)=\{n^{-1/(c+1/b)},n^{-1/(c_1+1/c_1)},n^{-1/(c_3+1/b_3)}\}$, which is rate optimal regularization.
\begin{enumerate}
    \item If in addition Assumption~\ref{assumption:smooth_gamma} holds then with high probability $$    \|\hat{\theta}^{FD}-\theta_0^{do}\|_{\infty}=O\left[n^{-(c-1)/\{2(c+1/b)\}}+n^{-(c_1-1)/\{2(c_1+1/b_1)\}}\right].
$$
    \item If in addition Assumption~\ref{assumption:smooth_op} holds with $\mathcal{A}_3=\mathcal{Y}$ and $\mathcal{B}_3=\mathcal{D}\times \mathcal{X}$ then with high probability
    $$
    \sup_{d\in\mathcal{D}}\|\hat{\theta}^{D:FD}(d)-\check{\theta}_0^{D:do}(d)\|_{\mathcal{H}_{\mathcal{Y}}}=O\left[n^{-(c_3-1)/\{2(c_3+1/b_3)\}}+n^{-(c_1-1)/\{2(c_1+1/b_1)}\right].
$$
\end{enumerate}

\end{theorem}
Explicit constants hidden by the $O(\cdot)$ notation are indicated in Supplement~\ref{sec:proof}. The rate is at best $n^{-1/4}$ when $(c,c_1,c_3)=2$ and $(b,b_1,b_3)\rightarrow \infty$, i.e. when the regressions are smooth and when the effective dimensions are finite. 
 Finally, we present a convergence in distribution result.

\begin{theorem}[Convergence in distribution of counterfactual distributions: DAG]\label{theorem:conv_dist_dag}
Suppose the conditions of Theorem~\ref{theorem:consistency_dag} hold, as well as Assumption~\ref{assumption:regularity}. Suppose samples $(\tilde{Y}_j)$ are calculated for $\theta_0^{D:FD}(d)$ as described in Algorithm~\ref{algorithm:herding}. Then $(\tilde{Y}_j)\rightsquigarrow \theta_0^{D:do}(d)$.
\end{theorem}
See Supplement~\ref{sec:proof} for the proof. 
\section{Algorithm derivation}\label{sec:derivation}

In this supplement, we derive estimators for (i) causal functions, (ii) counterfactual distributions, and (iii) graphical models. Before we do so, we compare kernel methods to series estimation. For intuition, consider $\hat{\theta}^{ATE}(d)$ with linear kernels $k(d,d')=d d'$ and $k(x,x')=x^{\top}x'$. Then by singular value decomposition,
$$
\hat{\theta}^{ATE}(d)=\left(d n^{-1}\sum_{i=1}^nX_i\right)^{\top}\left(n^{-1}\sum_{i=1}^n D_i^2X_iX_i^{\top}+\lambda  I \right)^{-1}\left(n^{-1}\sum_{i=1}^n D_iX_i Y_i\right).
$$
This formulation is interpretable as a regularized series estimator with basis function $\phi(d,x)=dx$. However, it requires scalar treatment, finite dimensional covariate, linear ridge regression, and computation $O\{dim(X)^3\}$. By contrast, the formulation in Algorithm~\ref{algorithm:treatment} allows for generic treatment, generic covariate, nonlinear ridge regression, and computation $O(n^3)$.

\subsection{Causal functions}

\begin{proof}[Proof of Theorem~\ref{theorem:representation_treatment}]
In Assumption~\ref{assumption:RKHS}, we impose that the scalar kernels are bounded. This assumption has several implications. First, the feature maps are Bochner integrable \cite[Definition A.5.20]{steinwart2008support}. Bochner integrability permits us to interchange expectation and inner product. Second, the mean embeddings exist. Third, the product kernel is also bounded and hence the tensor product RKHS inherits these favorable properties. By Lemma~\ref{theorem:id_treatment} and linearity of expectation,
 \begin{align*}
    \theta_0^{ATE}(d)&= \int \gamma_0(d,x)\mathrm{d}\text{\normalfont pr}(x)\\
    &=\int \langle \gamma_0, \phi(d)\otimes \phi(x)\rangle_{\mathcal{H}}  \mathrm{d}\text{\normalfont pr}(x) \\
    &= \langle \gamma_0, \phi(d)\otimes \int\phi(x) \mathrm{d}\text{\normalfont pr}(x) \rangle_{\mathcal{H}} \\
    &= \langle \gamma_0, \phi(d)\otimes \mu_x \rangle_{\mathcal{H}}.
\end{align*}
Likewise for $\theta_0^{DS}(d,\tilde{\text{\normalfont pr}})$. Next,
\begin{align*}
    \theta_0^{ATT}(d,d')&= \int \gamma_0(d',x)\mathrm{d}\text{\normalfont pr}(x \mid d)\\
    &=\int \langle \gamma_0, \phi(d')\otimes \phi(x)\rangle_{\mathcal{H}}  \mathrm{d}\text{\normalfont pr}(x \mid d) \\
    &= \langle \gamma_0, \phi(d')\otimes \int\phi(x) \mathrm{d}\text{\normalfont pr}(x \mid d) \rangle_{\mathcal{H}}\\
    &= \langle \gamma_0, \phi(d')\otimes \mu_x(d) \rangle_{\mathcal{H}}.
\end{align*}
Finally, 
 \begin{align*}
    \theta_0^{CATE}(d,v)&= \int \gamma_0(d,v,x)\mathrm{d}\text{\normalfont pr}(x \mid v)\\
    &=\int \langle \gamma_0, \phi(d)\otimes \phi(v) \otimes \phi(x)\rangle_{\mathcal{H}}  \mathrm{d}\text{\normalfont pr}(x \mid v) \\
    &= \langle \gamma_0, \phi(d)\otimes \phi(v) \otimes \int\phi(x) \mathrm{d}\text{\normalfont pr}(x \mid v) \rangle_{\mathcal{H}} \\
    &= \langle \gamma_0, \phi(d)\otimes \phi(v) \otimes \mu_{x}(v) \rangle_{\mathcal{H}}.
\end{align*}
\cite[Lemma 4.34]{steinwart2008support} guarantees that the derivative feature map $\nabla_d\phi(d)$ exists, is continuous, and is Bochner integrable since $$\kappa_d'=\left\{\sup_{d,d'\in\mathcal{D}}\nabla_d\nabla_{d'}k(d,d')\right\}^{1/2}<\infty.$$ Therefore the derivations remain valid for incremental functions.
\end{proof}

\begin{proof}[Proof of Algorithm~\ref{algorithm:treatment}]
By standard arguments \cite{kimeldorf1971some}
\begin{align*}
\hat{\gamma}(d,x)&=
\langle \hat{\gamma}, \phi(d)\otimes \phi(x) \rangle_{\mathcal{H}}=Y^{\top}(K_{DD}\odot K_{XX}+n\lambda  I )^{-1}(K_{Dd}\odot K_{Xx}).    
\end{align*}
    The results for $\hat{\theta}^{ATE}(d)$ holds by substitution:
   $$
   \hat{\mu}_x=n^{-1}\sum_{i=1}^n \phi(x_i),\quad \hat{\theta}^{ATE}(d)=\langle \hat{\gamma}, \phi(d)\otimes \hat{\mu}_x\rangle_{\mathcal{H}}.
   $$
    Likewise for $\hat{\theta}^{DS}(d,\tilde{\text{\normalfont pr}})$. 
    
    The results for $\hat{\theta}^{ATT}(d,d')$ and $\hat{\theta}^{CATE}(d,v)$ use the closed form of the conditional mean embedding from \cite[Algorithm 1]{singh2019kernel}. Specifically,
    $$
    \hat{\mu}_x(d)=K_{\cdot X}(K_{DD}+n\lambda_1 I )^{-1}K_{Dd},\quad  \hat{\theta}^{ATT}(d,d')=\langle \hat{\gamma}, \phi(d')\otimes \hat{\mu}_x(d)\rangle_{\mathcal{H}} 
    $$
and
    $$
    \hat{\mu}_{x}(v)=K_{\cdot X}(K_{VV}+n\lambda_2 I )^{-1}K_{Vv},\quad \hat{\theta}^{CATE}(d,v)=\langle \hat{\gamma}, \phi(d)\otimes \phi(v)\otimes \hat{\mu}_{x}(v)\rangle_{\mathcal{H}}.
    $$
For incremental functions, replace $\hat{\gamma}(d,x)$ with
\begin{align*}
\nabla_d\hat{\gamma}(d,x)&=
\langle \hat{\gamma}, \nabla_d \phi(d)\otimes \phi(x) \rangle_{\mathcal{H}}=Y^{\top}(K_{DD}\odot K_{XX}+n\lambda I )^{-1}(\nabla_d K_{Dd}\odot K_{Xx}).
\end{align*}
\end{proof}

\subsection{Counterfactual distributions}

\begin{proof}[Proof of Theorem~\ref{theorem:representation_dist}]
Assumption~\ref{assumption:RKHS} implies Bochner integrability, which permits us to interchange expectation and evaluation. Therefore by Lemma~\ref{theorem:id_treatment} and linearity of expectation,
    \begin{align*}
   \check{\theta}_0^{D:ATE}(d)&= \int \gamma_0(d,x)\mathrm{d}\text{\normalfont pr}(x)\\
   &=\int E_3^*\{\phi(d)\otimes \phi(x)\} \mathrm{d}\text{\normalfont pr}(x) \\
    &= E_3^*\{\phi(d)\otimes \int \phi(x)\mathrm{d}\text{\normalfont pr}(x)\} \\
    &= E_3^*\{\phi(d)\otimes \mu_x\}.
\end{align*}
Likewise for $\check{\theta}_0^{D:DS}(d,\tilde{\text{\normalfont pr}})$. Next,
    \begin{align*}
   \check{\theta}_0^{D:ATT}(d)&= \int \gamma_0(d',x)\mathrm{d}\text{\normalfont pr}(x \mid d)\\
   &=\int E_3^*\{\phi(d')\otimes \phi(x)\} \mathrm{d}\text{\normalfont pr}(x \mid d) \\
    &= E_3^*\{\phi(d')\otimes \int \phi(x)\mathrm{d}\text{\normalfont pr}(x \mid d)\}  \\
    &= E_3^*\{\phi(d')\otimes \mu_x(d)\}. 
\end{align*}
Finally,
    \begin{align*}
   \check{\theta}_0^{D:CATE}(d)&= \int \gamma_0(d,v,x)\mathrm{d}\text{\normalfont pr}(x \mid v) \\
   &=\int E_3^*\{\phi(d)\otimes \phi(v)\otimes \phi(x)\} \mathrm{d}\text{\normalfont pr}(x \mid v) \\
    &= E_3^*\{\phi(d)\otimes \phi(v)\otimes \int \phi(x)\mathrm{d}\text{\normalfont pr}(x \mid v)\}  \\
    &= E_3^*\{\phi(d)\otimes \phi(v) \otimes \mu_x(v)\}.
\end{align*}
\end{proof}

\begin{proof}[Proof of Algorithm~\ref{algorithm:dist}]
By \cite[Algorithm 1]{singh2019kernel},
$$
\hat{\gamma}(d,x)=\hat{E}_3^*\{\phi(d)\otimes \phi(x)\}=K_{\cdot Y}(K_{DD}\odot K_{XX}+n\lambda_3 I )^{-1}(K_{Dd}\odot K_{Xx}).
$$
The result for $\hat{\theta}^{D:ATE}$ follows by substitution: 
$$
\hat{\mu}_x=n^{-1}\sum_{i=1}^n \phi(x_i),\quad \hat{\theta}^{D:ATE}(d)=\hat{E}_3^*\{\phi(d)\otimes \hat{\mu}_x\}.
$$
Likewise for $\hat{\theta}^{D:DS}$. Both $ \hat{\theta}^{D:ATT}$ and $\hat{\theta}^{D:CATE}$ appeal to the closed form for conditional mean embeddings from \cite[Algorithm 1]{singh2019kernel}. Specifically,
\begin{align*}
     \hat{\mu}_x(d)&=K_{\cdot X}(K_{DD}+n\lambda_1 I )^{-1}K_{Dd},\quad \hat{\theta}^{D:ATT}(d,d')=\hat{E}_3^*\{\phi(d')\otimes \hat{\mu}_x(d)\}; \\
         \hat{\mu}_{x}(v)&=K_{\cdot X}(K_{VV}+n\lambda_2 I )^{-1}K_{Vv},\quad \hat{\theta}^{D:CATE}(d,v)=\hat{E}_3^*\{\phi(d)\otimes \phi(v)\otimes \hat{\mu}_{x}(v)\}.
\end{align*}
\end{proof}

\subsection{Graphical models}

\begin{proof}[Proof of Theorem~\ref{theorem:representation_treatment_dag}]
Assumption~\ref{assumption:RKHS} implies Bochner integrability, which permits us to interchange expectation and inner product. Therefore
\begin{align*}
    \theta_0^{do}(d)&=\int \gamma_0(d',x)\mathrm{d}\text{\normalfont pr}(d')\mathrm{d}\text{\normalfont pr}(x \mid d)  \\
    &=\int \langle \gamma_0,\phi(d')\otimes \phi(x) \rangle_{\mathcal{H}} \mathrm{d}\text{\normalfont pr}(d')\mathrm{d}\text{\normalfont pr}(x \mid d) \\
    &=\langle \gamma_0,\int \phi(d')\mathrm{d}\text{\normalfont pr}(d') \otimes \int \phi(x)\mathrm{d}\text{\normalfont pr}(x \mid d) \rangle_{\mathcal{H}}  \\
    &= \langle \gamma_0,\mu_d \otimes \mu_x(d) \rangle_{\mathcal{H}}.
\end{align*}
Similarly,  
\begin{align*}
    \check{\theta}_0^{D:do}(d)&=\int \gamma_0(d',x)\mathrm{d}\text{\normalfont pr}(d')\mathrm{d}\text{\normalfont pr}(x \mid d)  \\
    &=\int E_3^*\{\phi(d')\otimes \phi(x)\} \mathrm{d}\text{\normalfont pr}(d')\mathrm{d}\text{\normalfont pr}(x \mid d) \\
    &=E_3^*\left\{\int \phi(d')\mathrm{d}\text{\normalfont pr}(d') \otimes \int \phi(x)\mathrm{d}\text{\normalfont pr}(x \mid d) \right\}  \\
    &= E_3^*\{\mu_d \otimes \mu_x(d)\}.
\end{align*}
\end{proof}

\begin{proof}[Proof of Algorithm~\ref{algorithm:dag}]
Consider $\hat{\theta}^{do}$. By standard arguments \cite{kimeldorf1971some}
\begin{align*}
\hat{\gamma}(d,x)&=
\langle \hat{\gamma}, \phi(d)\otimes \phi(x) \rangle_{\mathcal{H}}=Y^{\top}(K_{DD}\odot K_{XX}+n\lambda I )^{-1}(K_{Dd}\odot K_{Xx}).  
\end{align*}
By \cite[Algorithm 1]{singh2019kernel}, write the mean embedding and conditional mean embedding as
  $$
  \hat{\mu}_x=n^{-1}\sum_{i=1}^n \phi(x_i),\quad \hat{\mu}_x(d)=K_{\cdot X}(K_{DD}+n\lambda_1 I )^{-1}K_{Dd}.
  $$
    Substitute these quantities to obtain
    $
    \hat{\theta}^{do}(d)=\langle \hat{\gamma}, \hat{\mu}_d\otimes \hat{\mu}_x(d)\rangle_{\mathcal{H}}
    $. 
    Next consider $\hat{\theta}^{D:do}$. By \cite[Algorithm 1]{singh2019kernel}
$$
\hat{\gamma}(d,x)=\hat{E}_3^*\{\phi(d)\otimes \phi(x)\}=K_{\cdot Y}(K_{DD}\odot K_{XX}+n\lambda_3 I )^{-1}(K_{Dd}\odot K_{Xx}).
$$
Substitution of the mean embeddings gives
    $
    \hat{\theta}^{D:do}(d)=\hat{E}_3^* \{\hat{\mu}_d\otimes \hat{\mu}_x(d)\}
    $. 
\end{proof}
\section{Tuning}\label{sec:tuning}

In the present work, we propose a family of novel estimators that are combinations of kernel ridge regressions. As such, the same two kinds of hyperparameters that arise in kernel ridge regressions arise in our estimators: ridge regression penalties and kernel hyperparameters. In this section, we describe practical tuning procedures for such hyperparameters. To simplify the discussion, we focus on the regression of $Y$ on $W$. Recall that the closed form solution of the regression estimator using all observations is
$$
\hat{f}(w)=K_{wW}(K_{WW}+n\lambda  I )^{-1}Y.
$$

\subsection{Ridge penalty}

It is convenient to tune $\lambda$ by leave-one-out cross validation (LOOCV) or generalized cross validation (GCV), since the validation losses have closed form solutions.

\begin{algorithm}[Ridge penalty tuning by LOOCV]
Construct the matrices
$$
H_{\lambda}= I -K_{WW}(K_{WW}+n\lambda  I )^{-1}\in\mathbb{R}^{n\times n},\quad \tilde{H}_{\lambda}=diag(H_{\lambda})\in\mathbb{R}^{n\times n}
$$
where $\tilde{H}_{\lambda}$ has the same diagonal entries as $H_{\lambda}$ and off diagonal entries of zero. Then set
$$
\lambda^*=\argmin_{\lambda \in\Lambda} n^{-1}\|\tilde{H}_{\lambda}^{-1}H_{\lambda} Y\|_2^2,\quad \Lambda\subset\mathbb{R}.
$$
\end{algorithm}

\begin{proof}
We prove that $n^{-1}\|\tilde{H}_{\lambda}^{-1} H_{\lambda}Y\|_2^2$ is the LOOCV loss. By definition, the LOOCV loss is
$
\mathcal{E}(\lambda)=n^{-1}\sum_{i=1}^n \{Y_i-\hat{f}_{-i}(W_i)\}^2
$
where $\hat{f}_{-i}$ is the regression estimator using all observations except the $i$th observation.

Let $\Phi$ be the matrix of features, with $i$th row $\phi(W_i)^{\top}$, and let $Q=\Phi^{\top}\Phi+n\lambda I$. By the regression first order condition,
\begin{align*}
    \hat{f}&=Q^{-1}\Phi^{\top}Y,\quad 
    \hat{f}_{-i}=\{Q-\phi(W_i)\phi(W_i)^{\top}\}^{-1}\{\Phi^{\top}Y-\phi(W_i)Y_i\}.
\end{align*}
Recall the Sherman-Morrison formula for rank one updates:
$$
(A+uv^{\top})^{-1}=A^{-1}-\frac{A^{-1}uv^{\top} A^{-1}}{1+v^{\top}A^{-1}u}.
$$
Hence
$$
\{Q-\phi(W_i)\phi(W_i)^{\top}\}^{-1}=Q^{-1}+\frac{Q^{-1}\phi(W_i)\phi(W_i)^{\top}Q^{-1}}{1-\phi(W_i)^{\top}Q^{-1}\phi(W_i)}.
$$
Let $\beta_i=\phi(W_i)^{\top} Q^{-1} \phi(W_i)$. Then
\begin{align*}
    \hat{f}_{-i}(W_i)&=\phi(W_i)^{\top} \left\{Q^{-1}+\frac{Q^{-1}\phi(W_i)\phi(W_i)^{\top}Q^{-1}}{1-\beta_i}\right\}\{\Phi^{\top}Y-\phi(W_i)Y_i\} \\
    &=\phi(W_i)^{\top} \left\{I+\frac{Q^{-1}\phi(W_i)\phi(W_i)^{\top}}{1-\beta_i}\right\}\{\hat{f}-Q^{-1}\phi(W_i)Y_i\} \\
    &=\left(1 +\frac{\beta_i}{1-\beta_i}\right)\phi(W_i)^{\top}\{\hat{f}-Q^{-1}\phi(W_i)Y_i\}\\
    &=\left(1 +\frac{\beta_i}{1-\beta_i}\right)\{\hat{f}(W_i)-\beta_iY_i\} \\
    &=\frac{1}{1-\beta_i}\{\hat{f}(W_i)-\beta_iY_i\},
\end{align*}
i.e. $\hat{f}_{-i}$ can be expressed in terms of $\hat{f}$.
Note that
\begin{align*}
    Y_i-\hat{f}_{-i}(W_i)&=Y_i-\frac{1}{1-\beta_i}\{\hat{f}(W_i)-\beta_iY_i\} \\
    &=Y_i+\frac{1}{1-\beta_i}\{\beta_iY_i-\hat{f}(W_i)\} \\
    &=\frac{1}{1-\beta_i}\{Y_i-\hat{f}(W_i)\}.
\end{align*}
Substituting back into the LOOCV loss
\begin{align*}
    n^{-1}\sum_{i=1}^n \left\{Y_i-\hat{f}_{-i}(W_i)\right\}^2 
    &=n^{-1}\sum_{i=1}^n \left[\{Y_i-\hat{f}(W_i)\}\left(\frac{1}{1-\beta_i}\right)\right]^2 \\
    &= n^{-1}\|\tilde{H}_{\lambda}^{-1} \{Y-K_{WW}(K_{WW}+n\lambda  I )^{-1}Y\}\|_2^2 \\
    &=n^{-1}\|\tilde{H}_{\lambda}^{-1} H_{\lambda}Y\|_2^2,
\end{align*}
since
$$
(\tilde{H}_{\lambda}^{-1})_{ii}=\frac{1}{(\tilde{H}_{\lambda})_{ii}}=\frac{1}{(H_{\lambda})_{ii}}=\frac{1}{1-\{K_{WW}(K_{WW}+n\lambda  I )^{-1}\}_{ii}}
$$
and
$$
K_{WW}(K_{WW}+n\lambda  I )^{-1}=\Phi\Phi^{\top}(\Phi\Phi^{\top}+n\lambda  I )^{-1}=\Phi(\Phi^{\top}\Phi+n\lambda I)^{-1}\Phi^{\top}=\Phi Q^{-1}\Phi^{\top}.
$$
\end{proof}

\begin{algorithm}[Ridge penalty tuning by GCV]
Construct the matrix
$$
H_{\lambda}= I -K_{WW}(K_{WW}+n\lambda  I )^{-1}\in\mathbb{R}^{n\times n}.
$$
Then set
$$
\lambda^*=\argmin_{\lambda \in\Lambda} n^{-1}\|\{\text{\normalfont tr}(H_{\lambda})\}^{-1} H_{\lambda} Y\|_2^2,\quad \Lambda\subset\mathbb{R}.
$$
\end{algorithm}

\begin{proof}
We match symbols with the classic derivation of \cite{craven1978smoothing}. Observe that
$$
\begin{Bmatrix} \hat{f}(W_1) \\ \vdots \\ f(W_n) \end{Bmatrix}=K_{WW}(K_{WW}+n\lambda  I )^{-1}Y=A_{\lambda}Y,\quad A_{\lambda}=K_{WW}(K_{WW}+n\lambda  I )^{-1}.
$$
Therefore
$$
H_{\lambda}= I -K_{WW}(K_{WW}+n\lambda  I )^{-1}= I -A_{\lambda}.
$$
\end{proof}

GCV can be viewed as a rotation invariant modification of LOOCV. In practice, we find that LOOCV and GCV provide almost identical hyperparameter values.

\subsection{Kernel}

The exponentiated quadratic kernel is the most popular kernel among machine learning researchers:
$$
k(w,w')=\exp\left\{-\frac{1}{2}\frac{(w-w')^2}{\iota^2}\right\}.
$$
Importantly, this kernel satisfies the required properties; it is continuous, bounded, and characteristic.

\cite[Section 4.3]{williams2006gaussian} characterize the exponentiated quadratic RKHS as an attenuated series of the form
$$
\mathcal{H}=\left(f=\sum_{j=1}^{\infty}f_j\varphi_j:\;\sum_{j=1}^{\infty} \frac{f_j^2}{\eta_j}<\infty\right),\quad \langle f,f' \rangle_{\mathcal{H}}=\sum_{j=1}^{\infty} \frac{f_jf_j'}{\eta_j}.
$$
For simplicity, take $\mathcal{W}=\mathbb{R}$ and take the measure $\nu$ to be the standard Gaussian distribution (more generally, it can be the population distribution $\text{\normalfont pr}$). Recall that the generalization of Mercer's Theorem permits $\mathcal{W}$ to be separable. Then the induced RKHS is characterized by
$$
\eta_j=\left(\frac{2\bar{a}}{\bar{A}}\right)^{1/2}\bar{B}^j,\quad \varphi_j(w)=\exp\{-(\bar{c}-\bar{a})w^2\}H_j\{w(2\bar{c})^{1/2}\}.
$$
$H_j$ is the $j$th Hermite polynomial, and the constants $(\bar{a},\bar{b},\bar{c},\bar{A},\bar{B})>0$ are
$$
\bar{a}=\frac{1}{4},\quad \bar{b}=\frac{1}{2\iota^2},\quad \bar{c}=(\bar{a}^2+2\bar{a}\bar{b})^{1/2},\quad \bar{A}=\bar{a}+\bar{b}+\bar{c},\quad \bar{B}=\frac{\bar{b}}{\bar{A}}<1.
$$
The eigenvalues $(\eta_j)$ geometrically decay, and the series $(\varphi_j)$ consists of weighted Hermite polynomials. For a function to belong to this RKHS, its coefficients on higher order weighted Hermite polynomials must be small.

Observe that the exponentiated quadratic kernel has a hyperparameter: the lengthscale $\iota$. A convenient heuristic is to set the lengthscale equal to the median interpoint distance of $(W_i)$ $(i=1,...,n)$, where the interpoint distance between observations $i$ and $j$ is $\|W_i-W_j\|_{\mathcal{W}}$. When the input $W$ is multidimensional, we use the kernel obtained as the product of scalar kernels for each input dimension. For example, if $\mathcal{W}\subset \mathbb{R}^d$ then
$$
k(w,w')=\prod_{j=1}^d \exp\left\{-\frac{1}{2}\frac{(w_j-w_j')^2}{\iota_j^2}\right\}.
$$
Each lengthscale $\iota_j$ is set according to the median interpoint distance for that input dimension.

In principle, we could instead use LOOCV or GCV to tune kernel hyperparameters in the same way that we use LOOCV or GCV to tune ridge penalties. However, given our choice of product kernel, this approach becomes impractical in high dimensions. For example, in the dose response curve design, $D\in\mathbb{R}$ and $X\in\mathbb{R}^{100}$ leading to a total of 101 lengthscales $(\iota_j)$. Even with a closed form solution for LOOCV and GCV, searching over this high dimensional grid becomes cumbersome.
\section{Balancing weight proof}\label{sec:balancing}

In this section, we provide the proofs to relate our algorithm with the balancing weight algorithms in previous work.

\begin{proof}[Proof of Proposition~\ref{prop:balance_exists}]
This result is standard in causal inference textbooks, e.g.  \cite[Technical Point 3.1]{hernan2010causal}. We state the proof for clarity:
\begin{align*}
 \int \gamma (d,x)\alpha_0(d,x)\mathrm{d}\text{\normalfont pr}(d,x)
    &= \int \int  \gamma(d,x) \frac{1(d=d^*)}{\text{\normalfont pr}(D=d^* \mid x)}  \mathrm{d}\text{\normalfont pr}(d \mid x) \mathrm{d}\text{\normalfont pr}(x)   \\
    &= \int \frac{1}{\text{\normalfont pr}(D=d^* \mid x)}  \int \gamma(d,x)1(d=d^*) \mathrm{d}\text{\normalfont pr}(d \mid x)\mathrm{d}\text{\normalfont pr}(x)  \\
    &= \int \frac{1}{\text{\normalfont pr}(D=d^* \mid x)} \gamma(d^*,x) \text{\normalfont pr}(D=d^* \mid x) \mathrm{d}\text{\normalfont pr}(x) \\
    &=\int \gamma(d^*,x)\mathrm{d}\text{\normalfont pr}(x).
\end{align*}
The variance of $\alpha_0$ is finite since $\text{\normalfont pr}(D=d^* \mid X)$ is bounded away from zero almost surely.
\end{proof}

\begin{proof}[Proof of Proposition~\ref{prop:balance_dne}]
The result follows from the Riesz representation theorem in $\mathbb{L}^2$, e.g. \cite[Theorem 5.3]{luenberger1997optimization} and \cite[Lemma 2.1]{chernozhukov2022debiased}. It is alluded to in e.g. \cite{van1991differentiable,newey1994asymptotic}. We state the proof for clarity.

Consider the functional $F:\gamma \mapsto \int \gamma(d^*,x)\mathrm{d}\text{\normalfont pr}(x)$ over $\mathbb{L}^2$. A Riesz representer $\alpha_0\in \mathbb{L}^2$ exists if and only if the functional $F$ is bounded and linear. Clearly the functional $F$ is linear in the sense that, for any scalar $c\in\mathbb{R}$, $F(c\gamma)=cF(\gamma)$. A linear functional is bounded over $\mathbb{L}^2$  if and only if it is continuous over $\mathbb{L}^2$ \cite[Proposition 5.1]{luenberger1997optimization}. We will show that this functional is not continuous over $\mathbb{L}^2$.

Consider the zero function $\tilde{0}\in \mathbb{L}^2$. The definition of continuity of $F$ at $\tilde{0}$ is as follows: for all $\epsilon>0$, there exists some $\delta>0$ such that for all $\gamma\in \mathbb{L}^2$, $\|\gamma-\tilde{0}\|_{\mathbb{L}^2}<\delta$ implies $|F(\gamma)-F(\tilde{0})|<\epsilon$. To violate continuity, we must show that there exists some $\epsilon>0$ such that for all $\delta>0$, there exists a $\tilde{\gamma}\in \mathbb{L}^2$ whereby $\|\tilde{\gamma}-\tilde{0}\|_{\mathbb{L}^2}<\delta$ yet $|F(\tilde{\gamma})-F(\tilde{0})|>\epsilon$.

To serve as this counterexample, define the function $\tilde{\gamma}$ such that $\tilde{\gamma}(d^*,x)=1$ for any $x$, and $\tilde{\gamma}(d,x)=0$ for any $x$ and any $d\neq d^*$. Observe that, because treatment is continuous, the set of values for which $d=d^*$ is a set with measure zero. Therefore $\|\tilde{\gamma}-\tilde{0}\|_2=0$ yet $|F(\tilde{\gamma})-F(\tilde{0})|=|1-0|=1$.

In summary, we have shown that $F$ is linear but not continuous over $\mathbb{L}^2$ and therefore not bounded over $\mathbb{L}^2$. Therefore its Riesz representer in $\mathbb{L}^2$ does not exist.
\end{proof}

\begin{proof}[Proof of Corollary~\ref{cor:connect}]
We proceed in steps. For clarity, we focus on the formulation of \cite{hirshberg2019minimax}, who consider estimation of $\theta_0^{ATE}(0)$ for binary treatment. We maintain the notation $\gamma_0(d,x)=E(Y \mid D=d,X=x)$.
\begin{enumerate}
    \item Reformulation of \cite{hirshberg2019minimax}. 
    
    The authors propose the estimator
$$
\tilde{\theta}^{ATE}(0)=n^{-1}\sum_{i=1}^n 1(D_i=0)\hat{w}(X_i) Y_i
$$
where $\hat{w}(x)$ is their estimator of $1/\text{\normalfont pr}(d=0 \mid x)$. Define
$$
\hat{\alpha}_i=\hat{\alpha}(D_i,X_i)=1(D_i=0)\hat{w}(X_i)
$$
so that 
$$
\tilde{\theta}^{ATE}(0)=n^{-1}\sum_{i=1}^n \hat{\alpha}_i Y_i.
$$

    \item Equivalence. 
    
    As noted in \cite[Lemma 1]{hirshberg2019minimax},
    $$
    \tilde{\theta}^{ATE}(0)=n^{-1}\sum_{i=1}^n\hat{f}(X_i)
    $$
    where $\hat{f}(x)$ is a kernel ridge regression estimator of $\gamma_0(0,x)$, which is estimated by subsetting to the untreated observations $(i:D_i=0)$ and then regressing $(Y_i)_{i:D_i=0}$ on $(X_i)_{i:D_i=0}$.
    
    \item Reformulation of our proposal.
    
    In Algorithm~\ref{algorithm:treatment}, we propose
    $$
    \hat{\theta}^{ATE}(0)=n^{-1}\sum_{i=1}^n Y^{\top}(K_{DD}\odot K_{XX}+n\lambda  I )^{-1}(K_{D0}\odot K_{Xx_i})=n^{-1}\sum_{i=1}^n \hat{\gamma}(0,X_i),
    $$
    where $\hat{\gamma}(d,x)$ is a kernel ridge regression estimator of $\gamma_0(d,x)$, which is estimated with all of the observations $(i=1,...,n)$. Take $k_{\mathcal{D}}(d,d')=1(d=d')$ in the product kernel $k(d,x;d',x')=k_{\mathcal{D}}(d,d')k_{\mathcal{X}}(x,x')$ of the tensor product RKHS $\mathcal{H}=\mathcal{H}_{\mathcal{D}}\otimes \mathcal{H}_{\mathcal{X}}$. Then it is numerically equivalent to set $\hat{\gamma}(0,x)=\hat{f}(x)$, where $\hat{f}(x)$ is the kernel ridge regression described above.
    
    \item Collecting results.
    
    In summary, we have shown
    $$
    n^{-1}\sum_{i=1}^n \hat{\alpha}_i Y_i=\tilde{\theta}^{ATE}(0)=n^{-1}\sum_{i=1}^n\hat{f}(X_i)=n^{-1}\sum_{i=1}^n \hat{\gamma}(0,X_i)=\hat{\theta}^{ATE}(0).
    $$
\end{enumerate}

\end{proof}

\begin{proof}[Proof of Corollary~\ref{cor:extend}]
We consider each case, appealing to Algorithm~\ref{algorithm:treatment}. Let $e_j\in \mathbb{R}^n$ be the vector of zeroes whose $j$th component is one. 
\begin{enumerate}
    \item $\hat{\theta}^{ATE}(d)=n^{-1}\sum_{i=1}^n Y^{\top}(K_{DD}\odot K_{XX}+n\lambda  I )^{-1}(K_{Dd}\odot K_{Xx_i})  $;
 
    Write $Z=n^{-1}\sum_{i=1}^n(K_{DD}\odot K_{XX}+n\lambda  I )^{-1}(K_{Dd}\odot K_{Xx_{i}})$. Then
    $$
    \hat{\theta}^{ATE}(d)=Y^{\top}Z=\sum_{i=1}^n Y_iZ_i=n^{-1}\sum_{i=1}^n Y_i nZ_i.
    $$
    Therefore
    $$
    \hat{\alpha}_j^{ATE}=nZ_j=ne_j^{\top} Z=e_j^{\top}\sum_{i=1}^n(K_{DD}\odot K_{XX}+n\lambda  I )^{-1}(K_{Dd}\odot K_{Xx_{i}}).
    $$
    
    
     \item $\hat{\theta}^{DS}(d,\tilde{\text{\normalfont pr}})=\tilde{n}^{-1}\sum_{i=1}^{\tilde{n}} Y^{\top}(K_{DD}\odot K_{XX}+n\lambda  I )^{-1}(K_{Dd}\odot K_{X\tilde{x}_i}) $;
     
     The argument is as above, taking $Z=\tilde{n}^{-1}\sum_{i=1}^{\tilde{n}}(K_{DD}\odot K_{XX}+n\lambda  I )^{-1}(K_{Dd}\odot K_{X\tilde{x}_i}) $. Therefore
     $$
     \hat{\alpha}_j^{DS}=n\tilde{n}^{-1}e_j^{\top}\sum_{i=1}^{\tilde{n}}(K_{DD}\odot K_{XX}+n\lambda  I )^{-1}(K_{Dd}\odot K_{X\tilde{x}_i}).
     $$
     
    \item $\hat{\theta}^{ATT}(d,d')=Y^{\top}(K_{DD}\odot K_{XX}+n\lambda  I )^{-1}[K_{Dd'}\odot \{K_{XX}(K_{DD}+n\lambda_1  I )^{-1}K_{Dd}\}]$;
    
    The argument is as above, taking $Z=(K_{DD}\odot K_{XX}+n\lambda  I )^{-1}[K_{Dd'}\odot \{K_{XX}(K_{DD}+n\lambda_1  I )^{-1}K_{Dd}\}]$. Therefore
     $$
     \hat{\alpha}_j^{ATT}=ne_j^{\top}(K_{DD}\odot K_{XX}+n\lambda  I )^{-1}[K_{Dd'}\odot \{K_{XX}(K_{DD}+n\lambda_1  I )^{-1}K_{Dd}\}].
     $$
     
    \item $\hat{\theta}^{CATE}(d,v)=Y^{\top}(K_{DD}\odot K_{VV}\odot K_{XX} +n\lambda  I )^{-1}[K_{Dd}\odot K_{Vv}\odot \{K_{XX}(K_{VV}+n\lambda_2  I )^{-1}K_{Vv} \}] $;
    
    The argument is as above, taking $Z=(K_{DD}\odot K_{VV}\odot K_{XX} +n\lambda  I )^{-1}[K_{Dd}\odot K_{Vv}\odot \{K_{XX}(K_{VV}+n\lambda_2  I )^{-1}K_{Vv} \}]$. Therefore
     $$
    \hat{\alpha}_j^{CATE}=ne_j^{\top}(K_{DD}\odot K_{VV}\odot K_{XX} +n\lambda  I )^{-1}[K_{Dd}\odot K_{Vv}\odot \{K_{XX}(K_{VV}+n\lambda_2  I )^{-1}K_{Vv} \}].
     $$
\end{enumerate}
Likewise for incremental functions, e.g.
$$
\hat{\alpha}_j^{\nabla:ATE}=e_j^{\top}\sum_{i=1}^n(K_{DD}\odot K_{XX}+n\lambda  I )^{-1}(\nabla_dK_{Dd}\odot K_{Xx_{i}}).
$$

\end{proof}


\section{Uniform consistency and convergence in distribution proof}\label{sec:proof}

In this supplement, we (i) present an equivalent definition of smoothness and relate our key assumptions with previous work; (ii) present technical lemmas for regression, unconditional mean embeddings, and conditional mean embeddings; (iii) appeal to these lemmas to prove uniform consistency of causal functions as well as convergence in distribution for counterfactual distributions.

\subsection{Assumptions revisited}

\textbf{Alternative representations of smoothness}

\begin{lemma}[Alternative representation of smoothness; Remark 2 of \cite{caponnetto2007optimal}]\label{lemma:alt}
If the input measure and Mercer measure are the same then there are equivalent formalisms for the source conditions in Assumptions~\ref{assumption:smooth_gamma} and~\ref{assumption:smooth_op}.
\begin{enumerate}
    \item The source condition in Assumption~\ref{assumption:smooth_gamma} holds if and only if the regression $\gamma_0$ is a particularly smooth element of $\mathcal{H}$. Formally, define the covariance operator $T$ for $\mathcal{H}$.
    We assume there exists $ g\in \mathcal{H}$ such that $\gamma_0=T^{(c-1)/2}g$, $c\in(1,2]$, and $\|g\|^2_{\mathcal{H}}\leq\zeta$.
    \item The source condition in Assumption~\ref{assumption:smooth_op} holds if and only if the conditional expectation operator $E_{\ell}$ is a particularly smooth element of $\mathcal{L}_2(\mathcal{H}_{\mathcal{A}_{\ell}},\mathcal{H}_{\mathcal{B}_{\ell}})$. Formally, define the covariance operator $T_{\ell}=E\{\phi(B_{\ell})\otimes \phi(B_{\ell})\}$ for $\mathcal{L}_2(\mathcal{H}_{\mathcal{A}_{\ell}},\mathcal{H}_{\mathcal{B}_{\ell}})$.
    We assume there exists $ G_{\ell}\in \mathcal{L}_2(\mathcal{H}_{\mathcal{A}_{\ell}},\mathcal{H}_{\mathcal{B}_{\ell}})$ such that $E_{\ell}=T_{\ell}^{(c_{\ell}-1)/2}\circ G_{\ell}$, $c_{\ell}\in(1,2]$, and $\|G_{\ell}\|^2_{\mathcal{L}_2(\mathcal{H}_{\mathcal{A}_{\ell}},\mathcal{H}_{\mathcal{B}_{\ell}})}\leq\zeta_{\ell}$.
\end{enumerate}
\end{lemma}

\begin{remark}[Covariance operator]
The covariance operator $T$ for the RKHS $\mathcal{H}$ depends on the setting.
    \begin{enumerate}
        \item $\theta_0^{ATE}$, $\theta_0^{DS}$, $\theta_0^{ATT}$: $T=E[\{\phi(D)\otimes \phi(X)\}\otimes \{\phi(D)\otimes \phi(X)\}]$;
        \item $\theta_0^{CATE}$: $T=E[\{\phi(D)\otimes \phi(V)\otimes \phi(X)\}\otimes \{\phi(D)\otimes \phi(V)\otimes \phi(X)\}]$.
    \end{enumerate}
\end{remark}

\cite{singh2019kernel} prove that $T_{\ell}$ and its powers are well defined under Assumption~\ref{assumption:RKHS}.

\textbf{Specific representations of smoothness}

Next, we instantiate the source condition in Assumption~\ref{assumption:smooth_op} for the different settings considered in the main text.

\begin{assumption}[Smoothness of mean embedding $\mu_x(d)$]\label{assumption:smooth_ATT}
Assume
\begin{enumerate}
\item The conditional expectation operator $E_1$ is well specified as a Hilbert--Schmidt operator between RKHSs, i.e. $E_1\in \mathcal{L}_2(\mathcal{H}_{\mathcal{X}},\mathcal{H}_{\mathcal{D}})$, where
    $
    E_1:\mathcal{H}_{\mathcal{X}} \rightarrow \mathcal{H}_{\mathcal{D}},\; f(\cdot)\mapsto E\{f(X) \mid D=\cdot\}.
    $
    \item The conditional expectation operator is a particularly smooth element of $\mathcal{L}_2(\mathcal{H}_{\mathcal{X}},\mathcal{H}_{\mathcal{D}})$. Formally, define the covariance operator $T_1=E\{\phi(D)\otimes \phi(D)\}$ for $\mathcal{L}_2(\mathcal{H}_{\mathcal{X}},\mathcal{H}_{\mathcal{D}})$.
    We assume there exists $ G_1\in \mathcal{L}_2(\mathcal{H}_{\mathcal{X}},\mathcal{H}_{\mathcal{D}})$ such that $E_1=T_1^{(c_1-1)/2}\circ G_1$, $c_1\in(1,2]$, and $\|G_1\|^2_{\mathcal{L}_2(\mathcal{H}_{\mathcal{X}},\mathcal{H}_{\mathcal{D}})}\leq\zeta_1$.
    \end{enumerate}
\end{assumption}

\begin{assumption}[Smoothness of mean embedding $\mu_x(v)$]\label{assumption:smooth_CATE}
Assume
\begin{enumerate}
\item The conditional expectation operator $E_2$ is well specified as a Hilbert--Schmidt operator between RKHSs, i.e. $E_2\in \mathcal{L}_2(\mathcal{H}_{\mathcal{X}},\mathcal{H}_{\mathcal{V}})$, where
    $
    E_2:\mathcal{H}_{\mathcal{X}} \rightarrow \mathcal{H}_{\mathcal{V}},\; f(\cdot)\mapsto E\{f(X) \mid V=\cdot\}.
    $
    \item The conditional expectation operator is a particularly smooth element of $\mathcal{L}_2(\mathcal{H}_{\mathcal{X}},\mathcal{H}_{\mathcal{V}})$. Formally, define the covariance operator $T_2=E\{\phi(V)\otimes \phi(V)\}$ for $\mathcal{L}_2(\mathcal{H}_{\mathcal{X}},\mathcal{H}_{\mathcal{V}})$.
    We assume there exists $ G_2\in \mathcal{L}_2(\mathcal{H}_{\mathcal{X}},\mathcal{H}_{\mathcal{V}})$ such that $E_2=T_2^{(c_2-1)/2}\circ G_2$, $c_2\in(1,2]$, and $\|G_2\|^2_{\mathcal{L}_2(\mathcal{H}_{\mathcal{X}},\mathcal{H}_{\mathcal{V}})}\leq\zeta_2$.
    \end{enumerate}
\end{assumption}

\begin{assumption}[Smoothness of conditional expectation operator $E_3$]\label{assumption:smooth_D:ATE}
Assume
\begin{enumerate}
\item The conditional expectation operator $E_3$ is well specified as a Hilbert--Schmidt operator between RKHSs, i.e. $E_3\in \mathcal{L}_2(\mathcal{H}_{\mathcal{Y}},\mathcal{H}_{\mathcal{D}}\otimes \mathcal{H}_{\mathcal{X}})$, where
    $
    E_3:\mathcal{H}_{\mathcal{Y}} \rightarrow \mathcal{H}_{\mathcal{D}}\otimes \mathcal{H}_{\mathcal{X}},\; f(\cdot)\mapsto E\{f(Y) \mid D=\cdot,X=\cdot \}.
    $
    \item The conditional expectation operator is a particularly smooth element of $\mathcal{L}_2(\mathcal{H}_{\mathcal{Y}},\mathcal{H}_{\mathcal{D}}\otimes \mathcal{H}_{\mathcal{X}})$. Formally, define the covariance operator $T_3=E[\{\phi(D)\otimes \phi(X)\} \otimes  \{\phi(D)\otimes \phi(X)\}]$ for $\mathcal{L}_2(\mathcal{H}_{\mathcal{Y}},\mathcal{H}_{\mathcal{D}}\otimes \mathcal{H}_{\mathcal{X}})$.
    We assume there exists $ G_3\in \mathcal{L}_2(\mathcal{H}_{\mathcal{Y}},\mathcal{H}_{\mathcal{D}}\otimes \mathcal{H}_{\mathcal{X}})$ such that $E_3=T_3^{(c_3-1)/2}\circ G_3$, $c_3\in(1,2]$, and $\|G_3\|^2_{\mathcal{L}_2(\mathcal{H}_{\mathcal{Y}},\mathcal{H}_{\mathcal{D}}\otimes \mathcal{H}_{\mathcal{X}})}\leq\zeta_3$.
    \end{enumerate}
\end{assumption}

\textbf{Interpreting smoothness for tensor products}

Another way to interpret the smoothness assumption for a tensor product RKHS follows from manipulation of the product kernel. For simplicity, consider the RKHS construction for $\theta_0^{ATE}$, take $k_{\mathcal{D}}$ to be the exponentiated quadratic kernel over $\mathcal{D}\subset\mathbb{R}$, and take $k_{\mathcal{X}}$ to be the exponentiated quadratic kernel over $\mathcal{X}\subset\mathbb{R}$. Define the vector of differences 
$$
v
=\begin{pmatrix} v_1 \\ v_2 \end{pmatrix}=\begin{pmatrix} d \\ x \end{pmatrix} - \begin{pmatrix} d' \\ x' \end{pmatrix} =\begin{pmatrix} d-d' \\ x-x' \end{pmatrix}. $$ 
Then $$
k(d,x;d',x')
=\exp\left(-\frac{1}{2}\frac{v_1^2}{\iota_1^2}\right)
\exp\left(-\frac{1}{2}\frac{v_2^2}{\iota_2^2}\right)
=\exp\left\{-\frac{1}{2}v^{\top} \begin{pmatrix}\iota_1^{-2} & 0 \\ 0 & \iota_2^{-2}\end{pmatrix}  v  \right\}.
$$
In summary, the product of exponentiated quadratic kernels over scalars is an exponentiated quadratic kernel over vectors. Therefore a tensor product of exponentiated quadratic RKHSs $\mathcal{H}_{\mathcal{D}}$ and $\mathcal{H}_{\mathcal{X}}$ begets an exponentiated quadratic RKHS $\mathcal{H}=\mathcal{H}_{\mathcal{D}} \otimes \mathcal{H}_{\mathcal{X}}$, for which the smoothness and spectral decay conditions admit their usual interpretation. The same is true anytime that a product of kernels begets a recognizable kernel.

\textbf{Matching assumptions with previous work}

Finally, we relate our approximation assumptions with previous work. Specifically, we match symbols with \cite{fischer2017sobolev}.

\begin{remark}[Matching assumptions]
Recall our main approximation assumptions.
\begin{enumerate}
    \item Source condition $c\in(1,2]$. \cite{fischer2017sobolev} refer to the source condition as SRC parametrized by $\beta$. Matching symbols, $c=\beta$. A larger value of $c$ is a stronger assumption.
    \item Effective dimension $b\geq 1$. \cite{fischer2017sobolev} refer to the effective dimension condition as EVD parametrized by $p$. Matching symbols, $b=1/p$. A larger value of $b$ is a stronger assumption.
    \item Embedding property $a\in(0,1]$. \cite{fischer2017sobolev} place an additional assumption EMB parametrized by $\alpha\in(0,1]$. In our setting of interest, $c\geq 1$ and the kernel is bounded. Together, these conditions imply $\alpha\leq  1$. Matching symbols, $a=\alpha$. A larger value of $a$ is a weaker assumption
\end{enumerate}
\end{remark}

In our algorithm derivation, we have already assumed correct specification and bounded kernels, i.e. we have already assumed that $c\geq 1$, $b\geq 1$, and $a\leq 1$. By placing explicit source and effective dimension conditions, we derive rates that adapt to stronger assumptions $c>1$ and $b>1$. 

It turns out that a further assumption of $a<1$ does not improve the rate, so we omit that additional complexity. Observe that $c\geq 1$ and $b\geq 1$ imply $c+1/b>1 \geq a$ for any value $a\in(0,1]$. The regime in which the inequality $c+1/b>a$ holds is the regime in which the rate does not depend on $a$ \cite[Theorem 1.ii]{fischer2017sobolev}, so the weakest version of the embedding property is sufficient for our purpose. We pose as a question for future work how to analyze the misspecified case, in which the stronger assumption of $a<1$ may play an important role.


\subsection{Lemmas}

\textbf{Regression}

For expositional purposes, we summarize classic results for the kernel ridge regression estimator $\hat{\gamma}$ for $\gamma_0(w)=E(Y \mid W=w)$. Consider the definitions
\begin{align*}
    \gamma_0&=\argmin_{\gamma\in\mathcal{H}}\mathcal{E}(\gamma),\quad \mathcal{E}(\gamma)=E[\{Y-\gamma(W)\}^2]; \\
    \hat{\gamma}&=\argmin_{\gamma\in\mathcal{H}}\hat{\mathcal{E}}(\gamma),\quad \hat{\mathcal{E}}(\gamma)=n^{-1}\sum_{i=1}^n\{Y_i-\gamma(W_i)\}^2+\lambda\|\gamma\|^2_{\mathcal{H}}.
\end{align*}

\begin{proposition}[Regression rate]\label{theorem:regression}
Suppose Assumptions~\ref{assumption:RKHS},~\ref{assumption:original}, and \ref{assumption:smooth_gamma} hold. Set $\lambda=n^{-1/(c+1/b)}$. Then with probability $1-\delta$, for $n$ sufficiently large, we have that
$$
\|\hat{\gamma}-\gamma_0\|_{\mathcal{H}}\leq r_{\gamma}(n,\delta,b,c)=C\log(4/\delta) \cdot n^{-\frac{1}{2}\frac{c-1}{c+1/b}},
$$
where $C$ is a constant independent of $n$ and $\delta$.
\end{proposition}

Remark~\ref{remark:big_n} in the subsequent, technical supplement elaborates on the meaning of the phrase ``$n$ sufficiently large''.

\begin{proof}
We verify the conditions of \cite[Theorem 1.ii]{fischer2017sobolev}. By Assumption~\ref{assumption:RKHS}, the kernel is bounded and measurable. Separability of the original spaces together with boundedness of the kernel imply that $\mathcal{H}$ is separable \cite[Lemma 4.33]{steinwart2008support}. By Assumption~\ref{assumption:original}, $\int y^2\mathrm{d}\text{\normalfont pr}(y)<\infty$. Since Assumption~\ref{assumption:smooth_gamma} implies $\gamma_0\in\mathcal{H}$, we have that $\|\gamma_0\|_{\infty}\leq \kappa_w \|\gamma_0\|_{\mathcal{H}}$ by Cauchy-Schwarz inequality. 

Next, we verify the assumptions called EMB, EVD, SRC, and MOM. Boundedness of the kernel implies EMB with $a=1$. EVD is the assumption we call effective dimension, parametrized by $b\geq 1$. SRC is the assumption we call the source condition, parametrized by $c\in(1,2]$ in our case. MOM is a Bernstein moment condition satisfied by hypothesis. We study the RKHS norm guarantee, which corresponds to Hilbert scale equal to one. We are in regime (ii) of the theorem, since $c+1/b>1$. For the exact finite sample constant, see \cite[Theorem 16]{fischer2017sobolev}.
\end{proof}

\textbf{Unconditional mean embedding}

For expositional purposes, we summarize classic results for the unconditional mean embedding estimator $\hat{\mu}_w$ for $\mu_w=E\{\phi(W)\}$.

\begin{lemma}[Bennett inequality; Lemma 2 of \cite{smale2007learning}]\label{lemma:prob}
Let $(\xi_i)$ be i.i.d. random variables drawn from distribution $\text{\normalfont pr}$ taking values in a real separable Hilbert space $\mathcal{K}$. Suppose there exists $ M$ such that
$
    \|\xi_i\|_{\mathcal{K}} \leq M<\infty$ almost surely and $
    \sigma^2(\xi_i)=E(\|\xi_i\|_{\mathcal{K}}^2)$. Then for all $n\in\mathbb{N}$ and for all $\delta\in(0,1)$,
$$
\text{\normalfont pr}\bigg[\bigg\|\dfrac{1}{n}\sum_{i=1}^n\xi_i-E(\xi)\bigg\|_{\mathcal{K}}\leq\dfrac{2M\log(2/\delta)}{n}+\left\{\dfrac{2\sigma^2(\xi)\log(2/\delta)}{n}\right\}^{1/2}\bigg]\geq 1-\delta.
$$
\end{lemma}

\begin{proposition}[Mean embedding rate]\label{theorem:mean}
Suppose Assumptions~\ref{assumption:RKHS} and~\ref{assumption:original} hold. Then with probability $1-\delta$, 
$$
\|\hat{\mu}_w-\mu_w\|_{\mathcal{H}_{\mathcal{W}}}\leq r_{\mu}(n,\delta)=\frac{4\kappa_w \log(2/\delta)}{n^{1/2}}.
$$
\end{proposition}

\begin{proof}
The result follows from Lemma~\ref{lemma:prob} with $\xi_i=\phi(W_i)$, since
$$
\left\|n^{-1}\sum_{i=1}^n \phi(W_i)-E\{\phi(W)\}\right\|_{\mathcal{H}_{\mathcal{W}}}
\leq \frac{2 \kappa_w\log(2/\delta)}{n}+\left\{\dfrac{2\kappa^2_w\log(2/\delta)}{n}\right\}^{1/2}\leq \frac{4\kappa_w \log(2/\delta)}{n^{1/2}}.
$$
\cite[Theorem 15]{altun2006unifying} originally prove this rate by McDiarmid inequality. See \cite[Theorem 2]{smola2007hilbert} for an argument via Rademacher complexity. See \cite[Proposition A.1]{tolstikhin2017minimax} for an improved constant and the proof that the rate is minimax optimal.
\end{proof}

\begin{remark}[Kernel bound]\label{remark:2}
In various applications, $\kappa_w$ varies.
    \begin{enumerate}
        \item $\theta_0^{ATE}$ and $\check{\theta}_0^{D:ATE}$: with probability $1-\delta$, $
\|\hat{\mu}_x-\mu_x\|_{\mathcal{H}_{\mathcal{X}}}\leq r_{\mu}(n,\delta)=4\kappa_x \log(2/\delta)n^{-1/2}.
$
        \item $\theta_0^{DS}$ and $\check{\theta}_0^{D:DS}$: with probability $1-\delta$, 
$
\|\hat{\nu}_x-\nu_x\|_{\mathcal{H}_{\mathcal{X}}}\leq r_{\nu}(\tilde{n},\delta)=4\kappa_x \log(2/\delta)\tilde{n}^{-1/2}.
$
 \item $\theta_0^{FD}$ and $\check{\theta}_0^{D:FD}$: with probability $1-\delta$, 
$
\|\hat{\mu}_d-\mu_d\|_{\mathcal{H}_{\mathcal{D}}}\leq r_{\mu}(n,\delta)=4\kappa_d \log(2/\delta)n^{-1/2}.
$
    \end{enumerate}
\end{remark}

\textbf{Conditional expectation operator and conditional mean embedding}

Next, we present original results for the generalized kernel ridge regression estimator $\hat{E}_{\ell}$ of the conditional expectation operator $E_{\ell}:\mathcal{H}_{\mathcal{A}_{\ell}}\rightarrow\mathcal{H}_{\mathcal{B}_{\ell}}$, $f(\cdot)\mapsto E\{f(A_{\ell}) \mid B_{\ell}=\cdot\}$. We prove these results and compare them with previous work in Supplement~\ref{sec:operator}.

Consider the definitions
\begin{align*}
    E_{\ell}&=\argmin_{E\in\mathcal{L}_2(\mathcal{H}_{\mathcal{A}_{\ell}},\mathcal{H}_{\mathcal{B}_{\ell}})}\mathcal{E}(E),\quad \mathcal{E}(E)=E[\{\phi(A_{\ell})-E^*\phi(B_{\ell})\}^2]; \\
    \hat{E}_{\ell}&=\argmin_{E\in\mathcal{L}_2(\mathcal{H}_{\mathcal{A}_{\ell}},\mathcal{H}_{\mathcal{B}_{\ell}})}\hat{\mathcal{E}}(E),\quad \hat{\mathcal{E}}(E)=n^{-1}\sum_{i=1}^n\{\phi(A_{\ell i})-E^*\phi(B_{\ell i})\}^2+\lambda_{\ell}\|E\|^2_{\mathcal{L}_2(\mathcal{H}_{\mathcal{A}_{\ell}},\mathcal{H}_{\mathcal{B}_{\ell}})}.
\end{align*}

\begin{proposition}[Conditional mean embedding rate]\label{theorem:conditional}
Suppose Assumptions~\ref{assumption:RKHS},~\ref{assumption:original}, and~\ref{assumption:smooth_op} hold. Set $\lambda_{\ell}=n^{-1/(c_{\ell}+1/b_{\ell})}$. Then with probability $1-\delta$, for $n$ sufficiently large,
$$
\|\hat{E}_{\ell}-E_{\ell}\|_{\mathcal{L}_2}\leq r_E(\delta,n,b_{\ell},c_{\ell})=C\log(4/\delta)\cdot n^{-(c_{\ell}-1)/\{2(c_{\ell}+1/b_{\ell})\}},
$$
where $C$ is a constant independent of $n$ and $\delta$. Moreover, for all $b\in\mathcal{B}_{\ell}$
$$
  \|\hat{\mu}_a(b)-\mu_a(b)\|_{\mathcal{H}_{\mathcal{A}_{\ell}}}\leq r_{\mu}(\delta,n,b_{\ell},c_{\ell})=\kappa_{b}\cdot
  r_E(\delta,n,b_{\ell},c_{\ell}).
    $$
\end{proposition}

Remark~\ref{remark:big_n} in the subsequent, technical supplement elaborates on the meaning of the phrase ``$n$ sufficiently large''.

\begin{proof}
We delay the proof of this result to the next supplement due to its technicality.
\end{proof}

\begin{remark}[Kernel bounds]\label{remark:3}
In various applications, $\kappa_a$ and $\kappa_b$ vary.
    \begin{enumerate}
        \item $\theta_0^{ATT}$ and $\check{\theta}_0^{ATT}$: $\kappa_a=\kappa_x$, $\kappa_b=\kappa_d$;
        \item $\theta_0^{CATE}$ and $\check{\theta}_0^{CATE}$: $\kappa_a=\kappa_x$, $\kappa_b=\kappa_v$;
        \item Counterfactual distributions: $\kappa_a=\kappa_y$, $\kappa_b=\kappa_d\kappa_x$.
    \end{enumerate}
\end{remark}

\subsection{Main results}

Appealing to Propositions~\ref{theorem:regression},~\ref{theorem:mean}, and~\ref{theorem:conditional}, we now prove consistency for (i) causal functions, (ii) counterfactual distributions, and (iii) graphical models. 

\textbf{Causal functions}

\begin{proof}[Proof of Theorem~\ref{theorem:consistency_treatment}]
We initially consider $\theta_0^{ATE}$. 
    \begin{align*}
        &\hat{\theta}^{ATE}(d)-\theta_0^{ATE}(d)
        =\langle \hat{\gamma} , \phi(d)\otimes \hat{\mu}_x \rangle_{\mathcal{H}} - \langle \gamma_0 , \phi(d)\otimes \mu_x \rangle_{\mathcal{H}} \\
        &=\langle \hat{\gamma} , \phi(d)\otimes(\hat{\mu}_x-\mu_x) \rangle_{\mathcal{H}} + \langle (\hat{\gamma}-\gamma_0), \phi(d) \otimes \mu_x \rangle_{\mathcal{H}} \\
        &=\langle (\hat{\gamma}-\gamma_0), \phi(d)\otimes(\hat{\mu}_x-\mu_x) \rangle_{\mathcal{H}} + \langle \gamma_0, \phi(d)\otimes(\hat{\mu}_x-\mu_x) \rangle_{\mathcal{H}}+\langle (\hat{\gamma}-\gamma_0), \phi(d) \otimes \mu_x \rangle_{\mathcal{H}}.
    \end{align*}
    Therefore by Propositions~\ref{theorem:regression} and~\ref{theorem:mean}, with probability $1-2\delta$,
   \begin{align*}
       &|\hat{\theta}^{ATE}(d)-\theta_0^{ATE}(d)|\leq 
       \|\hat{\gamma}-\gamma_0\|_{\mathcal{H}}\|\phi(d)\|_{\mathcal{H}_{\mathcal{D}}} \|\hat{\mu}_x-\mu_x\|_{\mathcal{H}_{\mathcal{X}}}\\
       &\quad +
       \|\gamma_0\|_{\mathcal{H}}\|\phi(d)\|_{\mathcal{H}_{\mathcal{D}}}\|\hat{\mu}_x-\mu_x\|_{\mathcal{H}_{\mathcal{X}}} \\
       &\quad + 
       \|\hat{\gamma}-\gamma_0\|_{\mathcal{H}}\|\phi(d)\|_{\mathcal{H}_{\mathcal{D}}} \|\mu_x\|_{\mathcal{H}_{\mathcal{X}}}
      \\
      &\leq \kappa_d \cdot r_{\gamma}(n,\delta,b,c) \cdot r_{\mu}(n,\delta)+\kappa_d\cdot\|\gamma_0\|_{\mathcal{H}} \cdot r_{\mu}(n,\delta)+\kappa_d\kappa_x \cdot r_{\gamma}(n,\delta,b,c)\\
      &=O\left(n^{-\frac{1}{2}\frac{c-1}{c+1/b}}\right).
   \end{align*}
By the same argument, with probability $1-2\delta$,
    \begin{align*}
   &|\hat{\theta}^{DS}(d,\tilde{\text{\normalfont pr}})-\theta_0^{DS}(d,\tilde{\text{\normalfont pr}})| \\
    &\leq \kappa_d \cdot r_{\gamma}(n,\delta,b,c) \cdot r_{\nu}(\tilde{n},\delta)+\kappa_d\cdot\|\gamma_0\|_{\mathcal{H}} \cdot r_{\nu}(\tilde{n},\delta)+\kappa_d\kappa_x \cdot r_{\gamma}(n,\delta,b,c)\\
      &=O\left( n^{-\frac{1}{2}\frac{c-1}{c+1/b}}+\tilde{n}^{-\frac{1}{2}}\right).
    \end{align*}
    Next, consider $\theta_0^{ATT}$.
    \begin{align*}
        &\hat{\theta}^{ATT}(d,d')-\theta_0^{ATT}(d,d') =\langle \hat{\gamma} , \phi(d')\otimes \hat{\mu}_x(d) \rangle_{\mathcal{H}} - \langle \gamma_0 , \phi(d')\otimes \mu_x(d) \rangle_{\mathcal{H}} \\
        &=\langle \hat{\gamma} , \phi(d')\otimes\{\hat{\mu}_x(d)-\mu_x(d)\} \rangle_{\mathcal{H}} + \langle (\hat{\gamma}-\gamma_0), \phi(d') \otimes \mu_x(d) \rangle_{\mathcal{H}} \\
        &=\langle (\hat{\gamma}-\gamma_0), \phi(d')\otimes\{\hat{\mu}_x(d)-\mu_x(d)\} \rangle_{\mathcal{H}} \\
        &\quad + \langle \gamma_0, \phi(d')\otimes\{\hat{\mu}_x(d)-\mu_x(d)\} \rangle_{\mathcal{H}}\\
        &\quad +\langle (\hat{\gamma}-\gamma_0), \phi(d') \otimes \mu_x(d) \rangle_{\mathcal{H}}.
    \end{align*}
    Therefore by Propositions~\ref{theorem:regression} and~\ref{theorem:conditional}, with probability $1-2\delta$,
   \begin{align*}
       &|\hat{\theta}^{ATT}(d,d')-\theta_0^{ATT}(d,d')|
       \leq 
       \|\hat{\gamma}-\gamma_0\|_{\mathcal{H}}\|\phi(d')\|_{\mathcal{H}_{\mathcal{D}}} \|\hat{\mu}_x(d)-\mu_x(d)\|_{\mathcal{H}_{\mathcal{X}}}\\
       &\quad +\|\gamma_0\|_{\mathcal{H}}\|\phi(d')\|_{\mathcal{H}_{\mathcal{D}}}\|\hat{\mu}_x(d)-\mu_x(d)\|_{\mathcal{H}_{\mathcal{X}}} \\
       &\quad +
       \|\hat{\gamma}-\gamma_0\|_{\mathcal{H}}\|\phi(d')\|_{\mathcal{H}_{\mathcal{D}}} \|\mu_x(d)\|_{\mathcal{H}_{\mathcal{X}}}
      \\
      &\leq \kappa_d \cdot r_{\gamma}(n,\delta,b,c) \cdot r_{\mu}^{ATT}(n,\delta,b_1,c_1)+\kappa_d\cdot\|\gamma_0\|_{\mathcal{H}} \cdot r_{\mu}^{ATT}(n,\delta,b_1,c_1)
      +\kappa_d\kappa_x \cdot r_{\gamma}(n,\delta,b,c)
      \\
      &=O\left(n^{-\frac{1}{2}\frac{c-1}{c+1/b}}+n^{-\frac{1}{2}\frac{c_1-1}{c_1+1/b_1}}\right).
   \end{align*}
    Finally, consider $\theta_0^{CATE}$.
    \begin{align*}
        &\hat{\theta}^{CATE}(d,v)-\theta_0^{CATE}(d,v)=\langle \hat{\gamma} , \phi(d)\otimes \phi(v)\otimes \hat{\mu}_{x}(v) \rangle_{\mathcal{H}} - \langle \gamma_0 , \phi(d )\otimes \phi(v) \otimes \mu_{x}(v) \rangle_{\mathcal{H}} \\
        &=\langle \hat{\gamma} , \phi(d)\otimes \phi(v)\otimes\{\hat{\mu}_{x}(v)-\mu_{x}(v)\} \rangle_{\mathcal{H}} + \langle (\hat{\gamma}-\gamma_0), \phi(d)\otimes \phi(v) \otimes \mu_{x}(v) \rangle_{\mathcal{H}} \\
        &=\langle (\hat{\gamma}-\gamma_0), \phi(d)\otimes \phi(v)\otimes\{\hat{\mu}_{x}(v)-\mu_{x}(v)\} \rangle_{\mathcal{H}}  \\
        &\quad + \langle \gamma_0, \phi(d)\otimes \phi(v)\otimes\{\hat{\mu}_{x}(v)-\mu_{x}(v)\} \rangle_{\mathcal{H}}\\
        &\quad +\langle (\hat{\gamma}-\gamma_0), \phi(d)\otimes \phi(v) \otimes \mu_{x}(v) \rangle_{\mathcal{H}}.
    \end{align*}
    Therefore by Propositions~\ref{theorem:regression} and~\ref{theorem:conditional}, with probability $1-2\delta$,
   \begin{align*}
       &|\hat{\theta}^{CATE}(d,v)-\theta_0^{CATE}(d,v)|\leq 
       \|\hat{\gamma}-\gamma_0\|_{\mathcal{H}}\|\phi(d)\|_{\mathcal{H}_{\mathcal{D}}}\|\phi(v)\|_{\mathcal{H}_{\mathcal{V}}} \|\hat{\mu}_{x}(v)-\mu_{x}(v)\|_{\mathcal{H}_{\mathcal{X}}}\\
      &\quad+
       \|\gamma_0\|_{\mathcal{H}}\|\phi(d)\|_{\mathcal{H}_{\mathcal{D}}}\|\phi(v)\|_{\mathcal{H}_{\mathcal{V}}}\|\hat{\mu}_{x}(v)-\mu_{x}(v)\|_{\mathcal{H}_{\mathcal{X}}} \\
       &\quad+
       \|\hat{\gamma}-\gamma_0\|_{\mathcal{H}}\|\phi(d)\|_{\mathcal{H}_{\mathcal{D}}}\|\phi(v)\|_{\mathcal{H}_{\mathcal{V}}} \|\mu_{x}(v)\|_{\mathcal{H}_{\mathcal{X}}}
      \\
      &\leq \kappa_d\kappa_{v} \cdot r_{\gamma}(n,\delta,b,c) \cdot r_{\mu}^{CATE}(n,\delta,b_2,c_2)\\
      &\quad+\kappa_d\kappa_{v}\cdot\|\gamma_0\|_{\mathcal{H}} \cdot r_{\mu}^{CATE}(n,\delta,b_2,c_2)
      +\kappa_d\kappa_{v} \kappa_{x} \cdot r_{\gamma}(n,\delta,b,c)
      \\
      &=O\left(n^{-\frac{1}{2}\frac{c-1}{c+1/b}}+n^{-\frac{1}{2}\frac{c_2-1}{c_2+1/b_2}}\right).
   \end{align*}
    For incremental functions, replace $\phi(d)$ with $\nabla_d \phi(d)$ and hence replace $\|\phi(d)\|_{\mathcal{H}_{\mathcal{D}}}\leq \kappa_d$ with $\|\nabla_d \phi(d)\|_{\mathcal{H}_{\mathcal{D}}}\leq \kappa_d'$.
\end{proof}

\textbf{Counterfactual distributions}

\begin{proof}[Proof of Theorem~\ref{theorem:consistency_dist}]
   The argument is analogous to Theorem~\ref{theorem:consistency_treatment}. By Propositions~\ref{theorem:mean} and~\ref{theorem:conditional}, for all $d\in\mathcal{D}$, with probability $1-2\delta$,
   \begin{align*}
       &\|\hat{\theta}^{D:ATE}(d)-\check{\theta}_0^{D:ATE}(d)\|_{\mathcal{H}_{\mathcal{Y}}}
      \\
      &\leq \kappa_d \cdot r_{E}(n,\delta,b_3,c_3) \cdot r_{\mu}(n,\delta)+\kappa_d\cdot\|E_3\|_{\mathcal{L}_2} \cdot r_{\mu}(n,\delta)
    +\kappa_d\kappa_x \cdot r_{E}(n,\delta,b_3,c_3)\\
      &=O\left(n^{-\frac{1}{2}\frac{c_3-1}{c_3+1/b_3}}\right).
   \end{align*}
Likewise, with probability $1-2\delta$,
    \begin{align*}
   &\|\hat{\theta}^{D:DS}(d,\tilde{\text{\normalfont pr}})-\check{\theta}_0^{D:DS}(d,\tilde{\text{\normalfont pr}})\|_{\mathcal{H}_{\mathcal{Y}}} \\
   &\leq
   \kappa_d \cdot r_{E}(n,\delta,b_3,c_3) \cdot r_{\nu}(\tilde{n},\delta)+\kappa_d\cdot\|E_3\|_{\mathcal{L}_2} \cdot r_{\nu}(\tilde{n},\delta)
    +\kappa_d\kappa_x \cdot r_{E}(n,\delta,b_3,c_3)\\
      &=O\left(n^{-\frac{1}{2}\frac{c_3-1}{c_3+1/b_3}}+\tilde{n}^{-\frac{1}{2}}\right).
    \end{align*}
By Proposition~\ref{theorem:conditional}, for all $d,d'\in\mathcal{D}$, with probability $1-2\delta$,
   \begin{align*}
       &\|\hat{\theta}^{D:ATT}(d,d')-\check{\theta}_0^{D:ATT}(d,d')\|_{\mathcal{H}_{\mathcal{Y}}}
      \\
      &\leq \kappa_d \cdot r_{E}(n,\delta,b_3,c_3) \cdot r^{ATT}_{\mu}(n,\delta,b_1,c_1) \\
      &\quad +\kappa_d\cdot\|E_3\|_{\mathcal{L}_2} \cdot r^{ATT}_{\mu}(n,\delta,b_1,c_1) 
      +\kappa_d\kappa_x \cdot r_{E}(n,\delta,b_3,c_3)\\
      &=O\left(n^{-\frac{1}{2}\frac{c_1-1}{c_1+1/b_1}}+n^{-\frac{1}{2}\frac{c_3-1}{c_3+1/b_3}}\right).
   \end{align*}
By Proposition~\ref{theorem:conditional}, for all $d\in\mathcal{D}$ and $v\in\mathcal{V}$ with probability $1-2\delta$,
   \begin{align*}
       &\|\hat{\theta}^{D:CATE}(d,v)-\check{\theta}_0^{D:CATE}(d,v)\|_{\mathcal{H}_{\mathcal{Y}}}
      \leq \kappa_d\kappa_{v} \cdot r_{E}(n,\delta,b_3,c_3) \cdot r_{\mu}^{CATE}(n,\delta,b_2,c_2) \\
      &\quad +\kappa_d\kappa_{v}\cdot\|E_3\|_{\mathcal{L}_2} \cdot r_{\mu}^{CATE}(n,\delta,b_2,c_2) \\
      &\quad +\kappa_d\kappa_{v} \kappa_{x} \cdot r_{E}(n,\delta,b_3,c_3)
      \\
      &=O\left(n^{-\frac{1}{2}\frac{c_2-1}{c_2+1/b_2}}+n^{-\frac{1}{2}\frac{c_3-1}{c_3+1/b_3}}\right).
   \end{align*}
      
\end{proof}

\begin{proof}[Proof of Theorem~\ref{theorem:conv_dist}]
    Fix $d$. By Theorem~\ref{theorem:consistency_dist}
    $$
    \|\hat{\theta}^{D:ATE}(d)-\check{\theta}_0^{D:ATE}(d)\|_{\mathcal{H}_{\mathcal{Y}}}=O_p\left(n^{-1/2\frac{c_3-1}{c_3+1/b_3}}\right).
    $$
    Denote the samples constructed by Algorithm~\ref{algorithm:herding} by $\tilde{Y}_j$ $(j=1,...,m)$. Then by \cite[Section 4.2]{bach2012equivalence},
    $$
    \left\|\hat{\theta}^{D:ATE}(d)-\frac{1}{m}\sum_{j=1}^m \phi(\tilde{Y}_j)\right\|_{\mathcal{H}_{\mathcal{Y}}}=O(m^{-1/2}).
    $$
    Therefore by triangle inequality,
    $$
    \left\|\frac{1}{m}\sum_{j=1}^m \phi(\tilde{Y}_j)-\check{\theta}_0^{D:ATE}(d)\right\|_{\mathcal{H}_{\mathcal{Y}}}=O_p\left(n^{-1/2\frac{c_3-1}{c_3+1/b_3}}+m^{-1/2}\right).
    $$
    The desired result follows from \cite{sriperumbudur2016optimal}, as quoted by \cite[Theorem 1.1]{simon2020metrizing}. The argument for other counterfactual distributions is identical.
\end{proof}

\textbf{Graphical models}

\begin{proposition}\label{prop:delta_dag}
If Assumptions~\ref{assumption:RKHS},~\ref{assumption:original}, and~\ref{assumption:smooth_op} hold with $\mathcal{A}_1=\mathcal{X}$ and $\mathcal{B}_1=\mathcal{D}$, then with probability $1-2\delta$,
$$
\|\hat{\mu}_d\otimes \hat{\mu}_x(d)-\mu_d\otimes \mu_x(d)\|_{\mathcal{H}_{\mathcal{D}}\otimes \mathcal{H}_{\mathcal{X}}}\leq \kappa_x \cdot r_{\mu} (n,\delta)+\kappa_d\cdot   r^{ATT}_{\mu}(n,\delta,b_1,c_1),
$$
where $r_{\mu}$ is as defined in Proposition~\ref{theorem:mean} (with $\kappa_w=\kappa_d$) and $r^{ATT}_{\mu}$ is as defined in Proposition~\ref{theorem:conditional}.
\end{proposition}

\begin{proof}
By triangle inequality,
\begin{align*}
    &\|\hat{\mu}_d\otimes \hat{\mu}_x(d)-\mu_d\otimes \mu_x(d)\|_{\mathcal{H}_{\mathcal{D}}\otimes \mathcal{H}_{\mathcal{X}}} \\
    &\leq \|\hat{\mu}_d\otimes \hat{\mu}_x(d)-\hat{\mu}_d\otimes \mu_x(d)\|_{\mathcal{H}_{\mathcal{D}}\otimes \mathcal{H}_{\mathcal{X}}} + 
    \|\hat{\mu}_d\otimes \mu_x(d)-\mu_d\otimes \mu_x(d)\|_{\mathcal{H}_{\mathcal{D}}\otimes \mathcal{H}_{\mathcal{X}}}.
\end{align*}
Focusing on the former term, by Proposition~\ref{theorem:conditional},
\begin{align*}
    \|\hat{\mu}_d\otimes \hat{\mu}_x(d)-\hat{\mu}_d\otimes \mu_x(d)\|_{\mathcal{H}_{\mathcal{D}}\otimes \mathcal{H}_{\mathcal{X}}}
    &\leq  \|\hat{\mu}_d\|_{\mathcal{H}_{\mathcal{D}}}\cdot \|\hat{\mu}_x(d)-\mu_x(d)\|_{\mathcal{H}_{\mathcal{X}}}\leq \kappa_d \cdot r^{ATT}_{\mu}(n,\delta,b_1,c_1).
\end{align*}
Focusing on the latter term, by Proposition~\ref{theorem:mean},
\begin{align*}
      \|\hat{\mu}_d\otimes \mu_x(d)-\mu_d\otimes \mu_x(d)\|_{\mathcal{H}_{\mathcal{D}}\otimes \mathcal{H}_{\mathcal{X}}}
    &\leq \|\hat{\mu}_d-\mu_d\|_{\mathcal{H}_{\mathcal{D}}}\cdot  \|\mu_x(d)\|_{\mathcal{H}_{\mathcal{X}}} \leq  \kappa_x \cdot r_{\mu} (n,\delta).
\end{align*}
\end{proof}

\begin{proof}[Proof of Theorem~\ref{theorem:consistency_dag}]
To begin, write
\begin{align*}
        \hat{\theta}^{FD}(d)-\theta_0^{do}(d)
        &=\langle \hat{\gamma} , \hat{\mu}_d\otimes \hat{\mu}_x(d) \rangle_{\mathcal{H}} - \langle \gamma_0 , \mu_d\otimes \mu_x(d) \rangle_{\mathcal{H}} \\
        &=\langle \hat{\gamma} , \{\hat{\mu}_d\otimes \hat{\mu}_x(d)-\mu_d\otimes \mu_x(d)\} \rangle_{\mathcal{H}} + \langle (\hat{\gamma}-\gamma_0), \mu_d\otimes \mu_x(d) \rangle_{\mathcal{H}} \\
        &=\langle (\hat{\gamma}-\gamma_0), \{\hat{\mu}_d\otimes \hat{\mu}_x(d)-\mu_d\otimes \mu_x(d)\}\rangle_{\mathcal{H}} \\
        &\quad + \langle \gamma_0, \{\hat{\mu}_d\otimes \hat{\mu}_x(d)-\mu_d\otimes \mu_x(d)\} \rangle_{\mathcal{H}}\\
        &\quad +\langle (\hat{\gamma}-\gamma_0), \mu_d\otimes \mu_x(d) \rangle_{\mathcal{H}}.
    \end{align*}
    Therefore by Propositions~\ref{theorem:regression},~\ref{theorem:mean},~\ref{theorem:conditional}, and~\ref{prop:delta_dag}, with probability $1-3\delta$,
  \begin{align*}
      |\hat{\theta}^{FD}(d)-\theta_0^{do}(d)|
      &\leq 
      \|\hat{\gamma}-\gamma_0\|_{\mathcal{H}}\|\hat{\mu}_d\otimes \hat{\mu}_x(d)-\mu_d\otimes \mu_x(d)\|_{\mathcal{H}_{\mathcal{D}}\otimes \mathcal{H}_{\mathcal{X}}}\\
      &\quad +
      \|\gamma_0\|_{\mathcal{H}}\|\hat{\mu}_d\otimes \hat{\mu}_x(d)-\mu_d\otimes \mu_x(d)\|_{\mathcal{H}_{\mathcal{D}}\otimes \mathcal{H}_{\mathcal{X}}} \\
      &\quad +
      \|\hat{\gamma}-\gamma_0\|_{\mathcal{H}}\|\mu_d\|_{\mathcal{H}_{\mathcal{D}}} \|\mu_x(d)\|_{\mathcal{H}_{\mathcal{X}}}
      \\
      &\leq r_{\gamma}(n,\delta,b,c)\{\kappa_x \cdot r_{\mu} (n,\delta)+\kappa_d\cdot   r^{ATT}_{\mu}(n,\delta,b_1,c_1)\} \\
      &\quad +\|\gamma_0\|_{\mathcal{H}}\{\kappa_x \cdot r_{\mu} (n,\delta)+\kappa_d\cdot   r^{ATT}_{\mu}(n,\delta,b_1,c_1)\}\\
      &\quad +\kappa_d\kappa_x\cdot r_{\gamma}(n,\delta,b,c) \\
      &=O\left(n^{-\frac{1}{2}\frac{c-1}{c+1/b}}+n^{-\frac{1}{2}\frac{c_1-1}{c_1+1/b_1}}\right).
  \end{align*}
  The argument for $\hat{\theta}^{D:FD}$ is analogous. By Propositions~\ref{theorem:mean},~\ref{theorem:conditional}, and~\ref{prop:delta_dag}, for all $d\in\mathcal{D}$, with probability $1-3\delta$,
   \begin{align*}
       &\|\hat{\theta}^{D:FD}(d)-\check{\theta}_0^{D:FD}(d)\|_{\mathcal{H}_{\mathcal{Y}}}
      \\
      &\leq r_{E}(n,\delta,b_3,c_3)\{\kappa_x \cdot r_{\mu} (n,\delta)+\kappa_d\cdot   r^{ATT}_{\mu}(n,\delta,b_1,c_1)\} \\
      &\quad +\|E_3\|_{\mathcal{L}_2}\{\kappa_x \cdot r_{\mu} (n,\delta)+\kappa_d\cdot   r^{ATT}_{\mu}(n,\delta,b_1,c_1)\}\\
      &\quad +\kappa_d\kappa_x\cdot r_{E}(n,\delta,b_3,c_3) \\
      &=O\left(n^{-\frac{1}{2}\frac{c_3-1}{c_3+1/b_3}}+n^{-\frac{1}{2}\frac{c_1-1}{c_1+1/b_1}}\right).
   \end{align*}
\end{proof}

\begin{proof}[Proof of Theorem~\ref{theorem:conv_dist_dag}]
The argument is identical to the proof of Theorem~\ref{theorem:conv_dist}.
\end{proof}
\section{Conditional expectation operator proof}\label{sec:operator}

In this supplement, we prove Proposition~\ref{theorem:conditional}, improving the rate of \cite{singh2019kernel} from $n^{-(c-1)/\{2(c+1)\}}$ to $n^{-(c-1)/\{2(c+1/b)\}}$. Our consideration of Hilbert--Schmidt norm departs from \cite{park2020measure,talwai2022sobolev}, who study surrogate risk and operator norm, respectively. Our assumptions also depart from \cite{singh2019kernel,park2020measure,talwai2022sobolev}. Instead, we directly generalize the assumptions of \cite{fischer2017sobolev} from the standard kernel ridge regression to the generalized kernel ridge regression that we use to estimate a conditional mean embedding. Our rate matches the minimax optimal rate shown in contemporaneous work of  \cite{li2022optimal}, who also study the misspecified case. We focus on the well specified case for reasons described in Section~\ref{sec:detail}, and employ a simpler proof strategy.

To lighten notation, we suppress the indexing of conditional expectation operators and conditional mean embeddings by $\ell$. Furthermore, to lighten notation, we abbreviate $\mathcal{L}_2=\mathcal{L}_2(\mathcal{H}_{\mathcal{A}},\mathcal{H}_{\mathcal{B}})$. In the simplified notation,
\begin{align*}
    E_0&=\argmin_{E\in\mathcal{L}_2}\mathcal{E}(E),\quad \mathcal{E}(E)=E[\{\phi(A)-E^*\phi(B)\}^2]; \\
    E_{\lambda}&=\argmin_{E\in\mathcal{L}_2}\mathcal{E}(E)+\lambda\|E\|^2_{\mathcal{L}_2};\\
    \hat{E}&=\argmin_{E\in\mathcal{L}_2}\hat{\mathcal{E}}(E),\quad \hat{\mathcal{E}}(E)=n^{-1}\sum_{i=1}^n\{\phi(A_{i})-E^*\phi(B_{i})\}^2+\lambda\|E\|^2_{\mathcal{L}_2}.
\end{align*}

\subsection{Bias}

\begin{proposition}[Conditional expectation operator bias; Theorem 6 of \cite{singh2019kernel}]\label{prop:bias}
Suppose Assumptions~\ref{assumption:RKHS},~\ref{assumption:original}, and the source condition in~\ref{assumption:smooth_op} hold. Then with probability one,
$$
\|E_{\lambda}-E_{0}\|_{\mathcal{L}_2} \leq \lambda^{\frac{c-1}{2}}\zeta^{1/2},
$$
where $\zeta$ is defined in Lemma~\ref{lemma:alt}.
\end{proposition}

\subsection{Variance}

\begin{lemma}[Helpful bounds]\label{lemma:bounds}
Suppose Assumptions~\ref{assumption:RKHS},~\ref{assumption:original}, and~\ref{assumption:smooth_op} hold. Let $\mu^{\lambda}_{a}(b)=E_{\lambda}^*\phi(b)$. We adopt the language of \cite{caponnetto2007optimal}.
\begin{enumerate}
    \item The generalized reconstruction error is $\mathcal{B}(\lambda)=\sup_{b\in\mathcal{B}}\|\mu^{\lambda}_a(b)-\mu_a(b)\|^2_{\mathcal{H}_{\mathcal{A}}} \leq \kappa^2_b \zeta \cdot \lambda^{c-1}$.
    \item The generalized effective dimension is $\mathcal{N}(\lambda)=\text{\normalfont tr}\{(T+\lambda I)^{-1}T\}\leq C (\pi/b) \{\sin(\pi/b)\}^{-1}\lambda^{-1/b}$.
\end{enumerate}
\end{lemma}

\begin{proof}
The first result is a corollary of Proposition~\ref{prop:bias}. 
The second result follows from \cite[eq. f]{sutherland2017fixing}, appealing to the effective dimension condition in Assumption~\ref{assumption:smooth_op}.
\end{proof}

\begin{lemma}[Decomposition of variance]\label{lemma:decomp}
Let $T_{AB}=E\{\phi(A)\otimes \phi(B)\}$ and let $E_n(\cdot)=n^{-1}\sum_{i=1}^n(\cdot)$. The following bound holds:
\begin{align*}
    \|\hat{E}-E_{\lambda}\|_{\mathcal{L}_2}
    &\leq \|\{\hat{T}_{AB}-T_{AB}(T_{BB}+\lambda I)^{-1}(\hat{T}_{BB}+\lambda I)\}(T_{BB}+\lambda I)^{-1/2}\|_{\mathcal{L}_2} \\
    &\quad \cdot \|(T_{BB}+\lambda I)^{1/2}(\hat{T}_{BB}+\lambda I)^{-1}(T_{BB}+\lambda I)^{1/2}\|_{op} \\
    &\quad \cdot \|(T_{BB}+\lambda I)^{-1/2}\|_{op}.
\end{align*}
Moreover, in the first factor,
\begin{align*}
    &\hat{T}_{AB}-T_{AB}(T_{BB}+\lambda I)^{-1}(\hat{T}_{BB}+\lambda I) \\
    &=E_n[\{\phi(A)-\mu^{\lambda}_a(B)\}\otimes \phi(B)]-E[\{\phi(A)-\mu^{\lambda}_a(B)\}\otimes \phi(B)].
\end{align*}
\end{lemma}

\begin{proof}
The result mirrors \cite[eq. 44]{fischer2017sobolev} and \cite[eq. 34]{talwai2022sobolev},
strengthening the RKHS norm to Hilbert--Schmidt norm via \cite[Proposition 22]{singh2019kernel}.
\end{proof}

\begin{lemma}[Bounding the first factor]\label{lemma:1}
Suppose Assumptions~\ref{assumption:RKHS} and~\ref{assumption:original} hold. Then with probability $1-\delta/2$, the first factor in Lemma~\ref{lemma:decomp} is bounded as 
\begin{align*}
  &\|\{\hat{T}_{AB}-T_{AB}(T_{BB}+\lambda I)^{-1}(\hat{T}_{BB}+\lambda I)\}(T_{BB}+\lambda I)^{-1/2}\|_{\mathcal{L}_2} \\
  &\leq 
    4\log(4/\delta) \left\{\frac{\kappa_a\kappa_b}{n\lambda^{1/2}}+\frac{\kappa_b\mathcal{B}(\lambda)^{1/2}}{n\lambda^{1/2}}+\frac{\kappa_a \mathcal{N}(\lambda)^{1/2}}{n^{1/2}}+\frac{\mathcal{B}(\lambda)^{1/2}\mathcal{N}(\lambda)^{1/2}}{n^{1/2}}\right\}.
\end{align*}
\end{lemma}

\begin{proof}
We verify the conditions of Lemma~\ref{lemma:prob}. Let
$$
\xi_i= [\{\phi(A_i)-\mu^{\lambda}_a(B_i)\} \otimes\phi(B_i)] (T_{BB}+\lambda I)^{-1/2}.
$$
We proceed in steps.
\begin{enumerate}
    \item First moment.
    
    Observe that 
\begin{align*}
    \|\xi_i\|_{\mathcal{L}_2}
    &=\|[\{\phi(A_i)-\mu^{\lambda}_a(B_i)\} \otimes\phi(B_i)] (T_{BB}+\lambda I)^{-1/2}\|_{\mathcal{L}_2} \\
    &=\| (T_{BB}+\lambda I)^{-1/2} [\phi(B_i)  \otimes \{\phi(A_i)-\mu^{\lambda}_a(B_i)\}] \|_{\mathcal{L}_2} \\
    &=\| (T_{BB}+\lambda I)^{-1/2} \phi(B_i) \|_{\mathcal{H}_{\mathcal{B}}} \cdot \|\phi(A_i)-\mu^{\lambda}_a(B_i)\|_{\mathcal{H}_{\mathcal{A}}}.
\end{align*}
Moreover
$$
\| (T_{BB}+\lambda I)^{-1/2} \phi(B_i) \|_{\mathcal{H}_{\mathcal{B}}} \leq \|(T_{BB}+\lambda I)^{-1/2}\|_{op} \| \phi(B_i) \|_{\mathcal{H}_{\mathcal{B}}}\leq \frac{\kappa_b}{\lambda^{1/2}}
$$
and
$$
\|\phi(A_i)-\mu^{\lambda}_a(B_i)\|_{\mathcal{H}_{\mathcal{A}}} \leq \|\phi(A_i)-\mu_a(B_i)\|_{\mathcal{H}_{\mathcal{A}}}+\|\mu_a(B_i)-\mu^{\lambda}_a(B_i)\|_{\mathcal{H}_{\mathcal{A}}} \leq 2\kappa_a+\mathcal{B}(\lambda)^{1/2}.
$$
In summary,
$$
\|\xi_i\|_{\mathcal{L}_2} \leq \frac{\kappa_b}{\lambda^{1/2}} \left\{2\kappa_a+\mathcal{B}(\lambda)^{1/2}\right\}.
$$
    \item Second moment.
    
    Next, write
    \begin{align*}
   &E(\|\xi_i\|^2_{\mathcal{L}_2}) \\
   &= \int \text{\normalfont tr}( [\{\phi(a)-\mu^{\lambda}_a(b)\} \otimes \phi(b)](T_{BB}+\lambda I)^{-1} [\phi(b)  \otimes \{\phi(a)-\mu^{\lambda}_a(b)\}]) \mathrm{d}\text{\normalfont pr}(a,b) \\
   &=\int \text{\normalfont tr}[ \{\phi(a)-\mu^{\lambda}_a(b)\} \langle \phi(b), (T_{BB}+\lambda I)^{-1} \phi(b)\rangle_{\mathcal{H}_{\mathcal{B}}} \langle \{\phi(a)-\mu^{\lambda}_a(b)\},\cdot \rangle_{\mathcal{H}_{\mathcal{A}}}] \mathrm{d}\text{\normalfont pr}(a,b)  \\
   &=\int \text{\normalfont tr}[ \langle \phi(b), (T_{BB}+\lambda I)^{-1} \phi(b)\rangle_{\mathcal{H}_{\mathcal{B}}} \langle \{\phi(a)-\mu^{\lambda}_a(b)\} ,\{\phi(a)-\mu^{\lambda}_a(b)\} \rangle_{\mathcal{H}_{\mathcal{A}}}] \mathrm{d}\text{\normalfont pr}(a,b)  \\
   &\leq \sup_{a,b} \|\phi(a)-\mu^{\lambda}_a(b)\|^2_{\mathcal{H}_{\mathcal{A}}} \cdot \int \text{\normalfont tr}\{ \langle \phi(b), (T_{BB}+\lambda I)^{-1} \phi(b)\rangle_{\mathcal{H}_{\mathcal{B}}}\} \mathrm{d}\text{\normalfont pr}(b).
    \end{align*}
    Focusing on the former factor,
    $$
     \|\phi(a)-\mu^{\lambda}_a(b)\|_{\mathcal{H}_{\mathcal{A}}}  \leq  \|\phi(a)-\mu_a(b)\|_{\mathcal{H}_{\mathcal{A}}}+\|\mu_a(b)-\mu^{\lambda}_a(b)\|_{\mathcal{H}_{\mathcal{A}}} \leq 2\kappa_a+\mathcal{B}(\lambda)^{1/2}.
    $$
    Therefore
    $$
    \sup_{a,b} \|\phi(a)-\mu^{\lambda}_a(b)\|^2_{\mathcal{H}_{\mathcal{A}}} \leq \left\{2\kappa_a+\mathcal{B}(\lambda)^{1/2}\right\}^2.
    $$
    Focusing on the latter factor, 
    \begin{align*}
        \int \text{\normalfont tr}\{ \langle \phi(b), (T_{BB}+\lambda I)^{-1} \phi(b)\rangle_{\mathcal{H}_{\mathcal{B}}}\} \mathrm{d}\text{\normalfont pr}(b)
        &=\int \text{\normalfont tr}[(T_{BB}+\lambda I)^{-1} \{\phi(b)\otimes\phi(b)\}] \mathrm{d}\text{\normalfont pr}(b) \\
        &= \text{\normalfont tr}\{(T_{BB}+\lambda I)^{-1} T_{BB}\} \\
        &=\mathcal{N}(\lambda).
    \end{align*}
    In summary,
    $$
    E(\|\xi_i\|^2_{\mathcal{L}_2}) \leq \mathcal{N}(\lambda)\left\{2\kappa_a+\mathcal{B}(\lambda)^{1/2}\right\}^2.
    $$
    \item Concentration.
    
    Therefore with probability $1-\delta/2$,
    \begin{align*}
        &\|E_n(\xi)-E(\xi)\|_{\mathcal{L}_2}  \\
        &\leq \frac{2 \log(4/\delta)}{n} \frac{\kappa_b}{\lambda^{1/2}} \left\{2\kappa_a+\mathcal{B}(\lambda)^{1/2}\right\} + \left[\frac{2 \log(4/\delta)}{n} \mathcal{N}(\lambda)\left\{2\kappa_a+\mathcal{B}(\lambda)^{1/2}\right\}^2\right]^{1/2} \\
        &\leq 4\log(4/\delta) \left\{\frac{\kappa_a\kappa_b}{n\lambda^{1/2}}+\frac{\kappa_b\mathcal{B}(\lambda)^{1/2}}{n\lambda^{1/2}}+\frac{\kappa_a \mathcal{N}(\lambda)^{1/2}}{n^{1/2}}+\frac{\mathcal{B}(\lambda)^{1/2}\mathcal{N}(\lambda)^{1/2}}{n^{1/2}}\right\}.
     \end{align*}

\end{enumerate}

\end{proof}

\begin{remark}[Sufficiently large $n$]\label{remark:big_n}
In the finite sample, we assume a certain inequality holds when bounding the second factor: 
\begin{equation}\label{eq:n_big}
  n\geq 8\kappa_b^2 \log(4/\delta) \cdot \lambda \cdot 
\log
\left\{ 2e\cdot \mathcal{N}(\lambda)\frac{\|T\|_{op}+\lambda}{\|T\|_{op}}
\right\}.  
\end{equation}
Ultimately, we will choose $\lambda=n^{-1/(c+1/b)}$ in Proposition~\ref{theorem:conditional}. This choice of $\lambda$ together with the bound on generalized effective dimension $\mathcal{N}(\lambda)$ in Lemma~\ref{lemma:bounds} imply that there exists an $n_0$ such that for all $n\geq n_0$,~\eqref{eq:n_big} holds, as argued by \cite[Proof of Theorem 1]{fischer2017sobolev}. We use the phrase ``$n$ sufficiently large'' when we appeal to this logic, and we summarize the final bound using $O(\cdot)$ notation.
\end{remark}

\begin{lemma}[Bounding the second factor]\label{lemma:2}
Suppose Assumptions~\ref{assumption:RKHS} and~\ref{assumption:original} hold. Further assume~\eqref{eq:n_big} holds. Then probability $1-\delta/2$, the second factor in Lemma~\ref{lemma:decomp} is bounded as 
$$
\|(T_{BB}+\lambda I)^{1/2}(\hat{T}_{BB}+\lambda I)^{-1}(T_{BB}+\lambda I)^{1/2}\|_{op} \leq 3.
$$
\end{lemma}

\begin{proof}
The result follows from \cite[eq. 44b, 47]{fischer2017sobolev}. In particular, our assumptions suffice for the properties used in \cite[Lemma 17]{fischer2017sobolev} to hold. Separability of $\mathcal{B}$ together with boundedness of the kernel $k_{\mathcal{B}}$ imply that $\mathcal{H}_{\mathcal{B}}$ is separable \cite[Lemma 4.33]{steinwart2008support}. Next, we verify the assumptions called EMB, EVD, and SRC. Boundedness of the kernel implies EMB with $a=1$. EVD is the assumption we call effective dimension, parametrized by $b\geq 1$. SRC is the assumption we call the source condition, parametrized by $c\in(1,2]$ in our case. 
\end{proof}

\begin{lemma}[Bounding the third factor]\label{lemma:3}
With probability one, the third factor in Lemma~\ref{lemma:decomp} is bounded as
$$
    \|(T_{BB}+\lambda I)^{-1/2}\|_{op} \leq \lambda^{-1/2}.
$$
\end{lemma}

\begin{proof}
The result follows from the definition of operator norm.
\end{proof}

\begin{proposition}[Conditional expectation operator variance]\label{prop:variance}
Suppose Assumptions~\ref{assumption:RKHS},~\ref{assumption:original}, and~\ref{assumption:smooth_op} hold. Further assume~\eqref{eq:n_big} holds and $\lambda\leq 1$. Then with probability $1-\delta$,
$$
\|\hat{E}-E_{\lambda}\|_{\mathcal{L}_2} \leq C\log(4/\delta)\left\{ \frac{1}{n\lambda}+\frac{1}{n^{1/2}\lambda^{1/(2b)+1/2}}\right\}.
$$
\end{proposition}

\begin{proof}
We combine the previous lemmas. By Lemmas~\ref{lemma:decomp},~\ref{lemma:1},~\ref{lemma:2}, and~\ref{lemma:3}, if~\eqref{eq:n_big} holds, then with probability $1-\delta$
\begin{align*}
    \|\hat{E}-E_{\lambda}\|_{\mathcal{L}_2} 
    &\leq \frac{12\log(4/\delta) }{\lambda^{1/2}} \left\{\frac{\kappa_a\kappa_b}{n\lambda^{1/2}}+\frac{\kappa_b\mathcal{B}(\lambda)^{1/2}}{n\lambda^{1/2}}+\frac{\kappa_a \mathcal{N}(\lambda)^{1/2}}{n^{1/2}}+\frac{\mathcal{B}(\lambda)^{1/2}\mathcal{N}(\lambda)^{1/2}}{n^{1/2}}\right\}.
\end{align*}
Next, recall the bounds in Lemma~\ref{lemma:bounds}. When $\lambda\leq 1$,
     $$
     \mathcal{B}(\lambda)^{1/2} \leq \kappa_b \zeta^{1/2} \lambda^{\frac{c-1}{2}} \leq \kappa_b \zeta^{1/2}.
     $$
     For brevity, write
     $$
     \mathcal{N}(\lambda)^{1/2}\leq C'\lambda^{-\frac{1}{2b}}.
     $$
     Therefore when $\lambda\leq 1$ the bound simplifies as
     \begin{align*}
         \|\hat{E}-E_{\lambda}\|_{\mathcal{L}_2} 
    \leq C\log(4/\delta) \left\{\frac{1}{n\lambda}+\frac{1}{n^{1/2}\lambda^{1/(2b)+1/2}}\right\}.
     \end{align*}
\end{proof}

\subsection{Collecting results}

\begin{proof}[Proof of Proposition~\ref{theorem:conditional}]
We combine and simplify Propositions~\ref{prop:bias} and~\ref{prop:variance}. Take $\lambda=n^{-1/(c+1/b)}$. For sufficiently large $n$,~\eqref{eq:n_big} holds and $\lambda\leq 1$ as explained in Remark~\ref{remark:big_n}. By triangle inequality, with probability $1-\delta$,
\begin{align*}
    \|\hat{E}-E_{\lambda}\|_{\mathcal{L}_2} &\leq  \|\hat{E}-E_{\lambda}\|_{\mathcal{L}_2}+ \|E_{\lambda}-E_0\|_{\mathcal{L}_2} \\
    &\leq  C\log(4/\delta) \left\{\frac{1}{n\lambda}+\frac{1}{n^{1/2}\lambda^{1/(2b)+1/2}}\right\}+C\lambda^{\frac{c-1}{2}}.
\end{align*}
Each term on the RHS simplifies as follows:
\begin{align*}
    &n^{-1}\lambda^{-1}
    =n^{-1}n^{1/(c+1/b)}
    =n^{1/(c+1/b)-1}
    =n^{\frac{1-c-1/b}{c+1/b}}
    =n^{-\frac{1}{2}\frac{2(c+1/b-1)}{c+1/b}} 
    \leq n^{-\frac{1}{2} \frac{c-1}{c+1/b}}; \\
    &n^{-1/2}\lambda^{-1\{1/(2b)+1/2\}}
    =n^{-1/2} n^{\frac{\{1/(2b)+1/2\}}{(c+1/b)}} 
    =n^{-\frac{1}{2}(1-\frac{1/b+1}{c+1/b})}
    =n^{-\frac{1}{2}(\frac{c+1/b-1/b-1}{c+1/b})}
    = n^{-\frac{1}{2}\frac{c-1}{c+1/b}};\\
    &\lambda^{\frac{c-1}{2}}
    = n^{-\frac{1}{c+1/b}\frac{c-1}{2}} =n^{-\frac{1}{2}\frac{c-1}{c+1/b}}.
\end{align*}
\end{proof}
\section{Simulation details}\label{sec:simulations}

In this appendix, we provide simulation details for (i) the dose response design, and (ii) the heterogeneous treatment effect design.

\subsection{Dose response curve}

A single observation consists of the triple $(Y,D,X)$ for outcome, treatment, and covariates where $Y,D\in\mathbb{R}$ and $X\in\mathbb{R}^{100}$. A single observation is generated is as follows. Draw unobserved noise as $\nu,\epsilon \overset{i.i.d.}{\sim}\mathcal{N}(0,1)$. Define the vector $\beta\in\mathbb{R}^{100}$ by $\beta_j=j^{-2}$. Define the matrix $\Sigma\in\mathbb{R}^{100\times 100}$ such that $\Sigma_{ii}=1$ and $\Sigma_{ij}=1(|i-j|=1)/2$ for $i\neq j$. Then draw $X\sim\mathcal{N}(0,\Sigma)$ and set
\begin{align*}
    D&=\Phi(3X^{\top}\beta)+0.75\nu,\quad 
    Y=1.2D+1.2X^{\top}\beta+D^2+DX_1+\epsilon.
\end{align*}

We implement our estimator $\hat{\theta}^{ATE}(d)$ (\texttt{RKHS}, white) described in Section~\ref{sec:algorithm}, with the tuning procedure described in Supplement~\ref{sec:tuning}. Specifically, we use ridge penalties determined by leave-one-out cross validation, and product exponentiated quadratic kernel with lengthscales set by the median heuristic. We implement \cite{kennedy2017nonparametric} (\texttt{DR1}, checkered white) using the default settings of the command \texttt{ctseff} in the \texttt{R} package \texttt{npcausal}. We implement \cite{colangelo2020double} (\texttt{DR2}, lined white) using default settings in \texttt{Python} code shared by the authors. Specifically, we use random forest for prediction, with the suggested hyperparameter values. For the Nadaraya--Watson smoothing, we select bandwidth that minimizes out-of-sample mean square error. We implement \cite{semenova2021debiased} (\texttt{DR-series}, gray) by modifying \texttt{ctseff}, as instructed by the authors. Importantly, we give \texttt{DR-series} the advantage of correct specification of the true dose response curve as a quadratic function.

\subsection{Heterogeneous treatment effect}

A single observations consists of the tuple $(Y,D,V,X)$, where outcome, treatment, and covariate of interest $Y,D,V\in\mathbb{R}$ and other covariates $X\in\mathbb{R}^3$. A single observation is generated as follows. Draw unobserved noise as $\epsilon_j \overset{i.i.d.}{\sim}\mathcal{U}(-1/2,1/2)$ $(j=1,...,4)$ and $\nu\sim \mathcal{N}(0,1/16)$. Then set
$$
V=\epsilon_1,\quad X=\begin{Bmatrix} 1+2V+\epsilon_2 \\ 1+2V+\epsilon_3 \\ (V-1)^2+\epsilon_4 \\ \end{Bmatrix}.
$$
Draw $D\sim Bernoulli[\Lambda \{(V+X_1+X_2+X_3)/2\}]$ where $\Lambda$ is the logistic link function. Finally set
$$
Y=\begin{cases} 0 &\text{ if } D=0;\\ V X_1 X_2 X_3+\nu &\text{ if } D=1. \end{cases}
$$
\cite{abrevaya2015estimating} also present a simpler version of this design.

We implement our estimator $\hat{\theta}^{CATE}(d,v)$ (\texttt{RKHS}, white) described in Section~\ref{sec:algorithm}, with the tuning procedure described in Supplement~\ref{sec:tuning}. Specifically, we use ridge penalties determined by leave-one-out cross validation. For multivariate functions, we use products of scalar kernels. For the binary treatment $D$, we use the binary kernel. For continuous variables, we use (product) exponentiated quadratic kernel with lengthscales set by the median heuristic. We implement \cite{abrevaya2015estimating} (\texttt{IPW}, lined gray) using default settings in the \texttt{MATLAB} code shared by the authors. We implement \cite{semenova2021debiased} (\texttt{DR-series}, gray) using the default settings of the command \texttt{best\_linear\_projection} in the \texttt{R} package \texttt{grf}. Importantly, we give \texttt{DR-series} the advantage of correct specification of the true heterogeneous treatment effect as the appropriate polynomial.
\section{Application details}\label{sec:application}

We implement our nonparametric estimators $\hat{\theta}^{ATE}(d)$, $\hat{\theta}^{\nabla:ATE}(d)$, and $\hat{\theta}^{CATE}(d,v)$ described in Section~\ref{sec:algorithm} (\texttt{RKHS}, solid).  We also implement the nonparametric estimator of \cite{colangelo2020double} (\texttt{DR2}, dashes) using default settings in \texttt{Python} code shared by the authors. Specifically, we use random forest for prediction, with the suggested hyperparameter values. Finally, we implement the semiparametric estimator of \cite{singh2021debiased} (\texttt{DR3}, vertical bars) with 95\% confidence intervals. Specifically, we reduce the continuous treatment into a discrete treatment that takes nine values corresponding to the roughly equiprobable bins $[40,250]$, $(250,500]$, $(500,750]$ $(750,1000]$, $(1000,1250]$, $(1250,1500]$, $(1500,1750]$, and $(1750,2000]$ class hours. Across estimators, we use the tuning procedure described in Supplement~\ref{sec:tuning}. Specifically, we use ridge penalties determined by leave-one-out cross validation, and product exponentiated quadratic kernel with lengthscales set by the median heuristic. 

\begin{figure}[H]
\begin{centering}
     \begin{subfigure}[b]{0.45\textwidth}
         \centering
         \includegraphics[width=\textwidth]{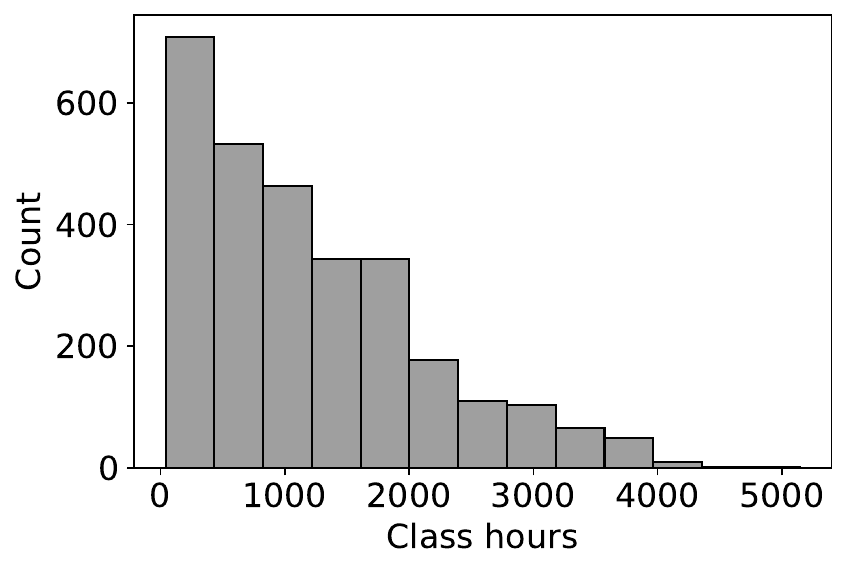}
         \vspace{-15pt}
         \caption{$D\geq 40$ and $Y>0$.}
     \end{subfigure}
     \hfill
     \begin{subfigure}[b]{0.45\textwidth}
         \centering
         \includegraphics[width=\textwidth]{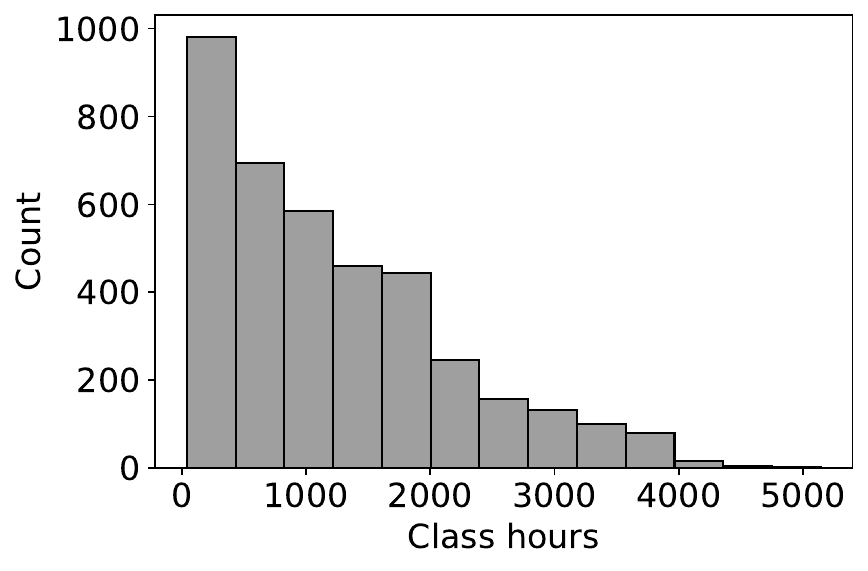}
         \vspace{-15pt}
         \caption{$D\geq 40$.}
     \end{subfigure}
\par
\caption{\label{fig:hist}
Class hours for different samples.}
\end{centering}
\end{figure}

\begin{figure}[ht]
\begin{centering}
     \begin{subfigure}[b]{0.45\textwidth}
         \centering
         \includegraphics[width=\textwidth]{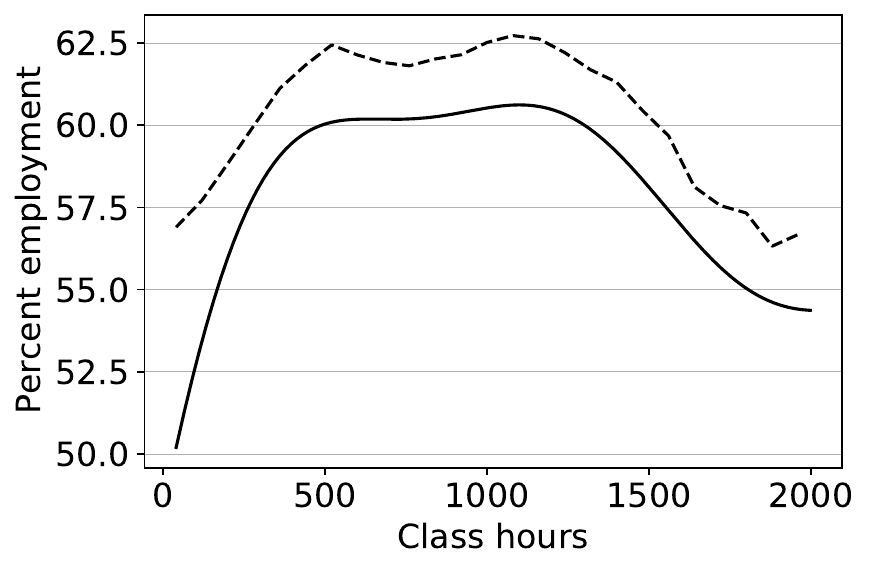}
         \caption{Dose response curve.}
     \end{subfigure}
     \hfill
     \begin{subfigure}[b]{0.45\textwidth}
         \centering
         \includegraphics[width=\textwidth]{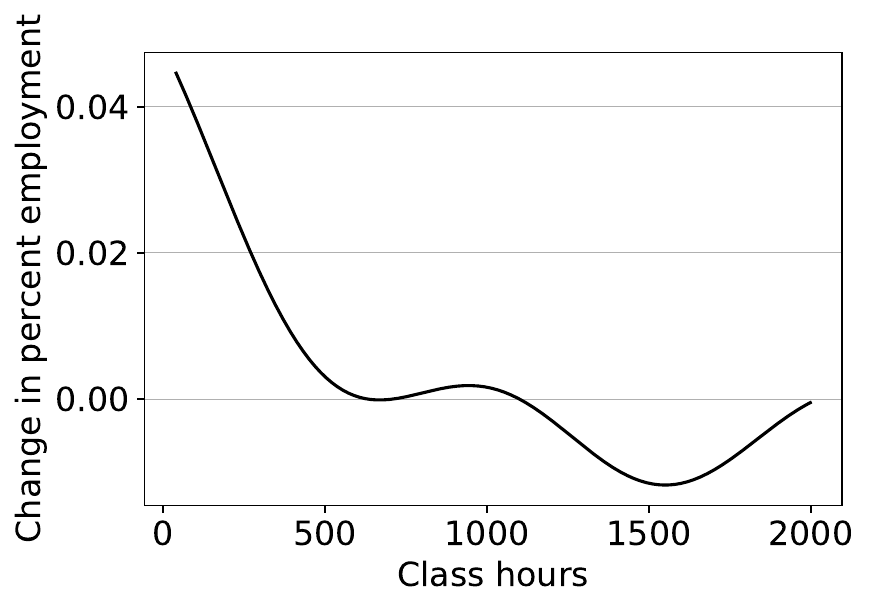}
         \caption{Incremental response curve.}
     \end{subfigure}
     \vskip\baselineskip
     \begin{subfigure}[b]{0.45\textwidth}
         \centering
         \includegraphics[width=\textwidth]{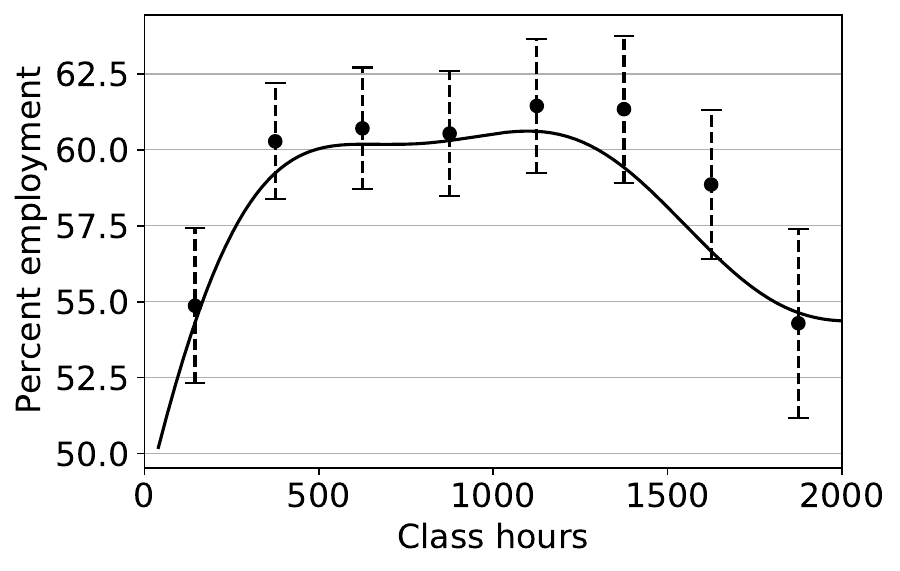}
         \caption{Discrete treatment effects.}
     \end{subfigure}
     \hfill
     \begin{subfigure}[b]{0.45\textwidth}
         \centering
         \includegraphics[width=\textwidth]{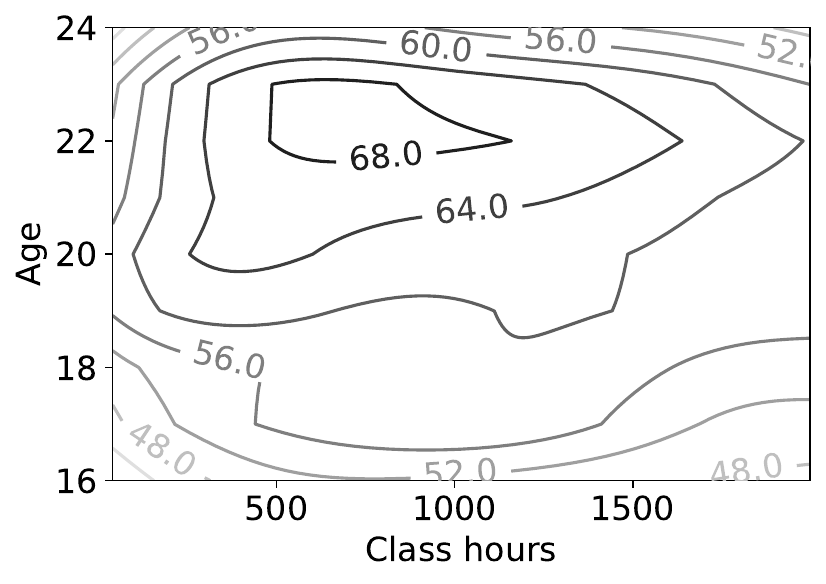}
         \caption{Heterogeneous response curve.}
     \end{subfigure}
\par
\caption{\label{fig:JC_d}
Effect of job training on employment: $D\geq 40$ and $Y>0$. We implement our estimators for dose, heterogeneous, and incremental response curves (\texttt{RKHS}, solid). For comparison, we also implement the dose response curve estimator of \cite{colangelo2020double} (\texttt{DR2}, dashes) as well as the discrete treatment effects of \cite{singh2021debiased} (\texttt{DR3}, vertical bars).}
\end{centering}
\end{figure}

We use the dataset published by \cite{huber2020direct}. In Supplement~\ref{sec:experiments}, we focus on the $n=3,906$ observations for which $D\geq 40$, i.e. individuals who completed at least one week of training. In this section, we verify that our results are robust to the choice of sample. Specifically, we consider the sample with $D\geq 40$ and $Y>0$, i.e. the $n=2,989$ individuals who completed at least one week of training and who found employment.

For each sample, we visualize class hours $D$ with a histogram in Figure~\ref{fig:hist}. The class hour distribution in the sample with $D\geq 40$ and $Y>0$ is similar to the class hour distribution in the sample with $D\geq 40$ that we use in Supplement~\ref{sec:experiments}. Next, we estimate the dose, heterogeneous, and incremental response curve for the new sample choice. Figure~\ref{fig:JC_d} 
visualizes results. For the sample with $D\geq 40$ and $Y>0$, the results mirror the results of the sample with $D\geq 40$ presented in Supplement~\ref{sec:experiments}. Excluding observations for which $Y=0$ leads to estimates that have the same shape but higher magnitudes, confirming the robustness of the results we present in Supplement~\ref{sec:experiments}.

\end{document}